\newif\ifhaspru
\hasprufalse

\newif\ifanon
\anonfalse

\documentclass[12pt]{article}

\usepackage[dvipsnames,svgnames,x11names,hyperref]{xcolor}
\usepackage{amsmath}
\usepackage{amsthm}
\usepackage[T1]{fontenc}
\usepackage{authblk}
\usepackage{comment}
\usepackage{fullpage}
\usepackage{thmtools}
\usepackage[colorlinks]{hyperref}
\hypersetup{colorlinks={true},linkcolor={blue},citecolor=magenta}

\usepackage[capitalize,compress]{cleveref}

\usepackage{amssymb,amsfonts,amsmath,amsthm,amscd,dsfont,mathrsfs,mathtools}
\usepackage{complexity}
\usepackage{tikz}
\usetikzlibrary{decorations.pathreplacing,backgrounds}
\usepackage{multirow}
\usepackage{mathpazo}
\usepackage{braket}
\usepackage{longfbox}

\usepackage[style=alphabetic,minalphanames=3,maxalphanames=4,maxnames=99,backref=true]{biblatex}

\usepackage{url}
\usepackage{bbm}
\usepackage{hyperref} 
\hypersetup{breaklinks=true}
\usepackage{cleveref}

\renewcommand{\vec}[1]{\mathbf{#1}}  
\newcommand{\bit}{\{0,1\}}
\newcommand{\reg}[1]{\mathsf{#1}}
\newcommand{\ignore}[1]{}

\newcommand{\proj}[1]{\ensuremath{|#1\rangle \!\langle #1|}}
\newcommand{\Z}{\mathbb{Z}}

\newcommand{\N}{\mathbb{N}}

\newcommand{\Tr}{\mathrm{Tr}}
\newcommand{\td}{\mathrm{td}}
\newcommand{\states}{\mathrm{S}}

\newcommand{\negl}{\mathsf{negl}}

\newcommand{\HPS}{\mathsf{HPS}}

\newcommand{\mb}{\mathbb}

\renewcommand{\vec}[1]{\mathbf{#1}}
\newcommand\id{\mathbb{I}}

\newcommand{\ot}{\otimes}

\newcommand{\norm}[1]{\left\lVert#1\right\rVert}

\newcommand{\opnorm}[1]{\norm{#1}_{\infty}}

\newcommand\algo{\mathcal}

\newcommand{\jnote}[1]{}
\newcommand{\alex}[1]{}
\newcommand{\jonas}[1]{}
\newcommand{\Dom}[1]{}
\newcommand{\dom}[1]{}

\usepackage{amsthm}
\usepackage{thmtools,thm-restate}

\declaretheoremstyle[bodyfont=\it,qed=\qedsymbol]{noproofstyle}

\declaretheorem[name=Observation,numbered=no]{observation*}

\declaretheorem[numberlike=equation]{theorem}

\declaretheorem[name=Theorem,numbered=no]{theorem*}

\declaretheorem[numberlike=equation]{lemma}
\declaretheorem[name=Lemma,numbered=no]{lemma*}

\declaretheorem[numberlike=equation]{corollary}
\declaretheorem[name=Corollary,numbered=no]{corollary*}

\declaretheorem[name=Proposition,numbered=no]{proposition*}

\declaretheorem[name=Claim,numbered=no]{claim*}

\declaretheorem[numberlike=equation]{conjecture}
\declaretheorem[name=Conjecture,numbered=no]{conjecture*}

\declaretheorem[name=Question,numbered=no]{question*}

\declaretheoremstyle[bodyfont=\it]{defstyle} 

\declaretheorem[numberlike=equation,style=defstyle]{definition}
\declaretheorem[unnumbered,name=Definition,style=defstyle]{definition*}

\declaretheorem[unnumbered,name=Example,style=defstyle]{example*}

\declaretheorem[unnumbered,name=Notation=defstyle]{notation*}

\declaretheorem[numberlike=equation,style=defstyle]{construction}
\declaretheorem[unnumbered,name=Construction,style=defstyle]{construction*}

\declaretheorem[numberlike=equation,style=defstyle]{remark}
\declaretheorem[unnumbered,name=Remark,style=defstyle]{remark*}

\declaretheorem[unnumbered,name=Remark,style=defstyle]{algorithm*}

\title{Efficient Quantum Pseudorandomness
from Hamiltonian Phase States}

\ifanon
\author{}
\else

\author[1]{John Bostanci}
\affil[1]{Columbia University}
\author[2]{Jonas Haferkamp}
\affil[2]{Harvard University}
\author[3,4]{Dominik Hangleiter}
\affil[3]{QuICS, University of Maryland \& NIST}
\affil[4]{Simons Institute for the Theory of Computing, UC Berkeley}
\author[5]{\quad Alexander Poremba}
\affil[5]{Massachusetts Institute of Technology}
\fi
\date{}

\begin{document}

\numberwithin{equation}{section}

\maketitle

\begin{abstract}
	
Quantum pseudorandomness has found applications in many areas of quantum information, ranging from entanglement theory, to models of scrambling phenomena in chaotic quantum systems, and, more recently, in the foundations of quantum cryptography. Kretschmer (TQC '21) showed that both pseudorandom states and pseudorandom unitaries exist even in a world without classical one-way functions. To this day, however, all known constructions require classical cryptographic building blocks which are themselves synonymous with the existence of one-way functions, and which are also challenging to implement on realistic quantum hardware.

In this work, we seek to make progress on both of these fronts simultaneously---by decoupling quantum pseudorandomness from classical cryptography altogether.
We introduce a quantum hardness assumption called the \emph{Hamiltonian Phase State} ($\mathsf{HPS}$) problem, which is the task of decoding output states of a random instantaneous quantum polynomial-time (IQP) circuit. Hamiltonian phase states can be generated very efficiently using only Hadamard gates, single-qubit $Z$ rotations and CNOT circuits. We show that the hardness of our problem reduces to a worst-case version of the problem, and we provide evidence that our assumption is plausibly \emph{fully quantum}; meaning, it cannot be used to construct one-way functions.
We also show information-theoretic hardness when only few copies of $\mathsf{HPS}$ are available by proving an approximate $t$-design property of our ensemble.
Finally, we show that our $\mathsf{HPS}$ assumption and its variants allow us to \emph{efficiently} construct many pseudorandom primitives, ranging from pseudorandom states, to quantum pseudoentanglement, to pseudorandom unitaries, and even primitives such as public-key encryption with quantum keys. \ifhaspru Along the way, we analyze a natural iterative construction of pseudorandom unitaries which resembles a candidate of Ji, Liu, and Song (CRYPTO'18).\fi
\end{abstract}

\newpage
\setcounter{tocdepth}{2}
\tableofcontents

\newpage

%!TEX root=./main.tex
%\section{Open questions}
%
%\begin{itemize}
%    \item Search-to-decision reduction for HPS? Potentially easier than longstanding problem of whether OWSG implies PRS? Potential idea: look into LWE search-to-decision reduction to guess one angle at a time
%
%    \item Non-cryptographic applications: hardness of quantum learning tasks?
%\end{itemize}

\section{Introduction}

Pseudorandomness~\cite{10.5555/1202577,TCS-010} is ubiquitous in theoretical computer science and has found applications in many areas, ranging from cryptography, to computational complexity, to the study of randomized algorithms, and even to combinatorics. 
The celebrated result of Håstad, Impagliazzo, Levin, and Luby~\cite{doi:10.1137/S0097539793244708} shows that one can construct a \emph{pseudorandom generator} from any one-way function---a function that is easy to evaluate but computationally hard to invert. Pseudorandom generators can then in turn be used to construct more advanced cryptographic primitives, such as \emph{pseudorandom functions}~\cite{10.1145/6490.6503}, i.e., keyed families of functions that appear random to any computationally bounded observer. This fact has elevated the notion of a one-way function as the minimal assumption in all of theoretical cryptography.
One-way functions are typically built from well-studied mathematical conjectures, such as the hardness of factoring~\cite{10.1145/359340.359342} and discrete logarithms~\cite{10.1145/359460.359473}, decoding error correcting codes~\cite{BFKL93,Alekhnovich03}, or finding short vectors in high-dimensional lattices~\cite{regev2009lattices}.
More advanced cryptographic primitives (which are believed  to lie beyond what is generically possible to construct from any one-way function), such as public-key encryption, tend to require highly structured assumptions which are more susceptible to algorithmic attacks---particularly by quantum computers~\cite{Shor_1997}, which has led to the design of \emph{post-quantum assumptions}~\cite{915326}.

In quantum cryptography\jnote{hi}, there has recently been a significant interest in so-called "fully quantum" cryptographic primitives (occasionally referred to as \emph{MicroCrypt} primitives) which are potentially \emph{weaker} than the conventional minimal assumptions used in classical cryptography. 
Here, the notion of \emph{quantum pseudorandomness} has emerged as the natural quantum analogue of pseudorandomness in the classical world~\cite{ji2018pseudorandom,kretschmer2021quantum,aaronson_et_al:LIPIcs.ITCS.2024.2}. In particular, Ji, Liu and Song~\cite{ji2018pseudorandom} proposed the notion of pseudorandom states~\cite{ji2018pseudorandom} and pseudorandom unitaries as the natural quantum analogues of pseudorandom generators~\cite{doi:10.1137/S0097539793244708} and pseudorandom functions~\cite{10.1145/6490.6503}, respectively.
The work of Kretschmer~\cite{kretschmer2021quantum, kretschmer2023quantum}
has shown that such fully quantum cryptographic primitives
can exist in a world in which no classical cryptography exists---including one-way functions.
At the same time, quantum pseudorandomness has applications in many areas of quantum information, ranging from entanglement theory~\cite{aaronson_et_al:LIPIcs.ITCS.2024.2,bouland_et_al:LIPIcs.CCC.2024.21,feng2024dynamicspseudoentanglement}, quantum learning theory~\cite{zhao2023learningquantumstatesunitaries}, to models of scrambling phenomena in chaotic quantum systems~\cite{Kim_2023,engelhardt2024cryptographiccensorship}, and, more generally, even in the foundations of quantum cryptography~\cite{ji2018pseudorandom,kretschmer2021quantum,kretschmer2023quantum,morimae2022one,ananth2022pseudorandom,brakerski2022computational,bostanci2023unitarycomplexityuhlmanntransformation,khurana2024commitments,batra2024commitments}. 

\paragraph{Limitations of existing constructions.}

Despite strong evidence that MicroCrypt primitives such as pseudorandom states and pseudorandom unitaries lie "below" one-way functions~\cite{kretschmer2021quantum, kretschmer2023quantum}, 
all known constructions implicitly make use of one-way functions (or other assumptions which are themselves synonymous with the existence of one-way functions)~\cite{ji2018pseudorandom,brakerski2019pseudo,metger2024simpleconstructionslineardepthtdesigns}. This begs the question:
\begin{center}
\emph{
Is it possible to construct fully quantum primitives, including quantum pseudorandomness, from quantum rather than classical hardness assumptions?}
\end{center}
Instantiating fully quantum primitives from a concrete and well-founded quantum hardness assumption (rather than from the existence of one-way functions) has remained a long standing open problem~\cite{10.1007/978-3-031-15802-5_8,morimae2022one}.

% A closer look at the state-of-the-art constructions of fully quantum cryptographic primitives from classical one-way functions, such as pseudorandom states~\cite{ji2018pseudorandom,brakerski2019pseudo} or pseudorandom unitaries~\cite{metger2024simpleconstructionslineardepthtdesigns} reveals that they depend on the use of pseudorandom functions,
Moreover, the fact that quantum pseudorandom states and unitaries are built from classical one-way functions makes them nearly impossible to realize on  realistic quantum hardware.
In some sense, this is inherent because cryptographic pseudorandom functions are highly complex by design~\cite{cryptoeprint:2024/1104}, and therefore require a massive computational overhead to implement coherently.
As a result, this severely limits the potential of using quantum pseudorandomness in practical applications; for example in the context of entanglement theory~\cite{aaronson_et_al:LIPIcs.ITCS.2024.2,bouland_et_al:LIPIcs.CCC.2024.21}, or when studying the
emergence of thermal
equilibria in isolated many-body systems~\cite{feng2024dynamicspseudoentanglement}, or when modeling scrambling phenomena in chaotic quantum systems~\cite{Kim_2023}. 
A second limitation of existing pseudorandom constructions is therefore also the notion of quantum efficiency, which begs the question: 
\begin{center}
\emph{Are there more efficient constructions of quantum pseudorandomness which can be implemented on realistic quantum hardware?}
\end{center}

Making progress on both of these  questions would not only lead to new insights in the foundations of quantum cryptography and the study of quantum hardness assumptions more generally, but also make quantum pseudorandomness more useful in practice.
To this day, however, no concrete fully quantum hardness assumption has been explored in an attempt to answer this question.

\paragraph{Towards a fully quantum assumption.} 

In order to plausibly claim that quantum pseudorandomness and other fully quantum cryptographic primitives exist in a world in which classical cryptography does not, we must construct these primitives from new assumptions that do not themselves imply classical cryptography.

The history of cryptography has taught us that finding good and well-founded cryptographic assumptions is not at all an easy task---even entirely plausible assumptions have often found surprising attacks~\cite{shor_algorithms_1994,cryptoeprint:2022/975,cryptoeprint:2022/214}.
What makes a new cryptographic assumption reasonable? While no widely agreed upon standards exist~\cite{cryptoeprint:2015/907}, the conventional belief is to use assumptions
\begin{itemize}
    \item which are rooted in a well-studied problem (ideally, a problem that has already been analyzed for many years) and which seems intractable in the worst case;

    \item for which there is a natural notion of what constitutes a "random instance" of the problem; moreover, such an instance can always be efficiently generated;
    
    \item for which there is evidence of average-case hardness, ideally in the form of a worst-case to average-case reduction;

    \item which can be connected to other assumptions or computational tasks that have been studied over the years, and

    \item which have enough structure to enable interesting cryptographic primitives.
\end{itemize}

A natural candidate for constructing quantum pseudorandomness (and other fully quantum cryptographic primitives) is via \emph{random quantum circuits}. 
In fact, the computational pseudorandomness of random quantum circuits appears to be a folklore conjecture and is widely believed among many quantum computer scientists.
As we are unaware of a concrete technical conjecture, we provide such a formulation here.
\begin{conjecture}[Random quantum circuits give rise to pseudorandom unitaries]
\label{conj:rqc} \ \\
Consider $n$-qubit random quantum circuits with $m$ gates defined by repeating the following process $m$ times independently at random: Draw a random pair $(i,j)$ of qubits and apply a gate from a universal gate set $\mathsf{G}\subset SU(4)$ to the qubits $i$ and $j$.
Then, there exist univeral constants $c>0$ and $C_{\mathsf{G}}>0$ (depending on the gate set $\mathsf{G}$) such that random quantum circuits with $m\geq C_{\mathsf{G}} n^c$ gates form ensembles of pseudorandom unitaries.
\end{conjecture}
We note that many other possible formulations (e.g. with specific geometric architectures or only regarding pseudorandom states) are also possible.
Indeed, if~\Cref{conj:rqc} holds even with exponential security, then Ref.~\cite{schuster_random_2024} implies that a simple ensemble of random quantum circuits in a 1D architecture of depth $\mathrm{polylog}(n)$ is also pseudorandom.

\cref{conj:rqc} can be seen as a direct quantum analogue of a claim that was first proposed by Gowers~\cite{gowers1996almost} who conjectured that random reversible quantum circuits form pseudorandom permutations on bitstrings.
This conjecture has inspired multiple recent works in classical cryptography.
For instance, it was recently proven by He and O'Donnell~\cite{he2024pseudorandom} that the Luby-Rackoff~\cite{luby1988construct} construction of pseudorandom permutations from pseudorandom functions can be implemented with reversible permutations.
Random reversible circuits have recently also inspired entirely new approaches for constructing program obfuscation schemes~\cite{canetti2024towards}.

Gowers originally conjectured the emergence of pseudorandomness when attempting to prove that random quantum circuits converge quickly to ensembles of $t$-wise independent permutations~\cite{gowers1996almost} (this bound was further improved later on towards an optimal scaling~\cite{hoory2005simple,brodsky2008simple,chen2024incompressibility}.
In fact, this property can itself be viewed as evidence for pseudorandomness as we further discuss in \cref{sec:hardness and quantumness}.
It turns out that random quantum circuits satisfy an analogous property by converging nearly optimally towards approximate $t$-designs~\cite{chen2024incompressibility,brandao2016local,haferkamp2022random,schuster_random_2024}.

However, we currently do not have rigorous evidence for \cref{conj:rqc}; for example, in terms of a worst-to-average reduction for a corresponding learning problem. 
Moreover, and maybe more importantly, it is unclear how one would use unstructured random quantum circuits to construct more advanced quantum cryptographic primitives (e.g., as discussed in \Cref{sec:intro-app}).
A similar situation arises for general one-way functions, which require additional structure to build more advanced cryptographic applications, such as public-key encryption.
It could very well be the case that random quantum circuits are simply \emph{too mixing} to be a useful in the context of quantum cryptography.
A natural way forward is to search for a sweet spot---an ensemble of random quantum circuits that is sufficiently structured to permit the construction of interesting cryptographic primitives but which, at the same time, is sufficiently mixing to guarantee security.
% \ \\
% \ \\
% \alex{Here, we can discuss the random quantum circuit conjecture which is folklore...(let's explicitly state it). @Jonas, would you like to take a stab at this? Some ideas...we can mention work of Gower's and O'Donnell on random classical circuits, and also attempts at constructing iO (e.g. Ran Canetti's work).
% Then transition: why are random circuits not a good candidate? Explain that it doesn't have structure to enable some applications. Is there a sweet spot? Barely enough structure, but still very good average-case hardness arguments, e.g., worst-case to average-case reductions? At the end, we can finish with a "cliff-hanger": random but structured quantum circuits seem like a good candidate, but how do argue average-case hardness? And then we transition to our contributions.
% }
% \jonas{Gave it a shot. Very tired though. Will go through it tomorrow again.}
% \dom{I went over it once. }

%!TEX root=./main.tex
\section{Our contributions}
\label{sec:contributions}

% \alex{Transition into what we do, starting here?}
In this work, we simultaneously address the two major open problems in the field of quantum pseudorandomness and propose the first well-founded and fully quantum hardness assumption. 
To this end, we follow the strategy sketched above, and propose a family of quantum states which we call \emph{Hamiltonian Phase States}. 
These states are a family of quantum states which are "maximally quantum" in the sense that the state has support on all bitstrings with amplitudes equal in magnitude, but varying phases. Hamiltonian Phase States are generated by a family of commuting \emph{instantaneous quantum polynomial-time} (IQP) circuits which generalize the $X$ programs proposed by Shepherd and Bremner~\cite{shepherd_temporally_2009}.
The corresponding circuits are highly structured in that they are generated by a Hamiltonian with only $Z$-type terms applied to the all-$\ket +$ state.
This structure makes them amenable to rigorous analysis \cite{kahanamoku-meyer_forging_2023,bremner_iqp_2023,gross_secret_2023}. 
At the same time, these circuits are also believed to be sufficiently mixing and hard to simulate classically~\cite{shepherd_temporally_2009,bremner_classical_2010,Bremner_2016,hangleiter_computational_2023}.
Moreover, since Hamiltonian phase states can be generated by a commuting Hamiltonian, they admit a highly efficient implementation in practice.

% which bares resemblance to the $\mathsf{LWE}$ problem. \alex{Hmm... I think that's far fetched :D, esepcially now that we don't reveal the architecture. I actually think the opposite is true: it bares little resemblence to anything classical, and that's a good thing?} 

Phase states are a natural direction to look at in the search for a fully-quantum cryptographic assumption with sufficient amounts of structure. 
On the one hand, this is because of their quantum advantage properties.
On the other hand, the (quantum) learnability of different ensembles of phase states has been studied extensively in recent work~\cite{arunachalam2023optimal}.
 % studied extensively the problem of learning phase states from quantum samples.  
There, the authors give optimal bounds for the sample complexity of learning many families of phase states from quantum samples, as well as upper bounds on the time complexity. 
Importantly, there are families of phase states generated using a small number of (long-range) gates, which cannot be learned from polynomially many samples.
Following this, our proposed cryptographic assumption is that Hamiltonian Phase States are hard to learn, given quantum samples and classical side information. 

Moreover, the known constructions for pseudorandom states with useful cryptographic applications are based on phase states~\cite{ji2018pseudorandom}. These are generated using a single-bit output quantum-secure pseudorandom function family $\{f_k\}_k$ with 
\begin{align}
    \label{eq:ji liu song PRS}
    \ket{\phi_k} \propto \sum_x \omega_q^{f_k(x)} \ket x,
\end{align}
where $\omega_q$ is a $q$-th root of unity, for example $q=2$~\cite{brakerski2019pseudo}. 
Because these states are based on a classical assumption, they require the reversible implementation of a classical PRF which requires a large number of Toffoli gates. 
These are extremely expensive in standard fault-tolerant constructions.
However, the results of Refs.~\cite{ji2018pseudorandom,arunachalam2023optimal} suggest that a more natural family of phase states which is generated by a quantum circuit with a small number of expensive gates can also yield quantum pseudorandomness. 
This would require gates affecting a large number of qubits, since low-degree phase states can be learned efficiently. 
As we show below, in spite of having terms with high support, the Hamiltonian Phase States can be generated highly efficiently using only local $Z$-rotations and CNOT gates.  

% \alex{Incoporate previous discussion on random circuits? The learning theory is going to help is with evidence of worst-case hardness?}  
% These are quantum input, classical output problems, many of which have information theoretic algorithms but no computationally efficient algorithms. 

\subsection{Hamiltonian Phase States}

\begin{figure}
    \includegraphics{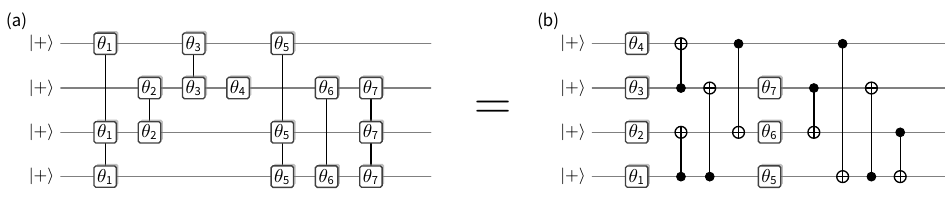}
    \caption{\label{fig:hps} 
    Hamiltonian Phase States (HPS) are generated by sequentially applying Ising-type rotations around angles $\theta_i$ to the state $\ket{+^n} =H^{\otimes n} \ket{0^n}$. 
    (a) Example of a HPS on 4 qubits. Connected boxes at sites $i,j,k$ with angle $\theta$ represent the unitary $\exp(i \theta Z_i Z_j Z_k ) $.
    (b) HPS can be implemented using only single-qubit $Z$ rotations interlaced with \textsf{CNOT} circuits.
    } 
\end{figure}
Let $\vec A \in \Z_2^{m \times n}$ be a binary matrix and let $\boldsymbol{\theta} = (\theta_1, \dots, \theta_m)$ be a set of uniformly random angles in the interval $[0,2\pi)$ according to some discretization into $q = \poly(n)$ parts. 
We consider phase states of the form 
\begin{align}
\ket{\Phi_{\boldsymbol{\theta}}^{\vec A}} = \exp\left(\mathrm{i} \, \sum_{i=1}^m \theta_i \bigotimes_{j=1}^n \mathsf{Z}^{\vec A_{ij}}  \right) H^{\otimes n} \ket{0^n}. \label{eq:HPS-state}
\end{align}
where, for $i \in [m]$, we denote the $i$-th row of $\vec A$ by $(\vec A_{i1}, \dots,\vec A_{in})$ and let
$$
\bigotimes_{j=1}^n \mathsf{Z}^{\vec A_{ij}} = \mathsf{Z}^{\vec A_{i1}} \otimes \dots \otimes \mathsf{Z}^{\vec A_{in}} \quad\quad \text{ for } \quad\quad \mathsf{Z}^0 = \id,\,\,\, \mathsf{Z}^1 = \mathsf{Z}.
$$
We call these states \emph{Hamiltonian Phase States}  since they can naturally be prepared as the result of a time evolution under a sparse Ising Hamiltonian. 
We also call the matrix $\vec A$ the \emph{architecture} of the states, as it specifies the overall structure/location of the Ising terms.
Hamiltonian Phase States with a single fixed angle $\theta_i \equiv \theta$ have been studied as a means to demonstrate verified quantum advantage, when measured in the $X$ basis, under the name $X$-programs~\cite{shepherd_temporally_2009}.
A Hamiltonian Phase State is therefore a generalized version of an $X$ program parameterized by the pair $(\vec A, \boldsymbol \theta)$. 
$X$ programs with  $\theta=\pi/8$ have the interesting property that its Fourier coefficients can be computed efficiently classically, but at the same time the simulation of such $X$-programs is believed to be classically intractable \cite{shepherd_temporally_2009,bremner_classical_2010,Bremner_2016,hangleiter_computational_2023}. 

Our cryptographic assumption rests on the  apparent hardness of \emph{learning} Hamiltonian Phase States (or generalized $X$ programs), which was highlighted in recent work~\cite{arunachalam2023optimal}. Concretely, our quantum computational assumption amounts to the conjecture that our ensemble of Hamiltonian phase states satisfies the following two properties:
\begin{itemize}
    \item Random Hamiltonian Phase States are \emph{hard to invert} in the following sense: given $\ket{\Phi_{\boldsymbol{\theta}}^{\vec A}}^{\otimes t}$, for any $t=\poly(n)$, it is computationally difficult to reverse-engineer the angles $\boldsymbol{\theta}$ and architecture $\vec A$. This means that the ensemble $\{\ket{\Phi_{\boldsymbol{\theta}}^{\vec A}}\}_{\boldsymbol{\theta},\vec A}$ gives rise to a so-called \emph{one-way state generator} (OWSG) as in \Cref{def:OWSG}. 
    
    \item Random Hamiltonian Phase States are \emph{hard to distinguish from Haar random states} in the following sense: given $\ket{\Phi_{\boldsymbol{\theta}}^{\vec A}}^{\otimes t}$, for any $t=\poly(n)$, it is computationally difficult to distinguish $\ket{\Phi_{\boldsymbol{\theta}}^{\vec A}}^{\otimes t}$ from $\ket{\Psi}^{\otimes t}$, where $\ket{\Psi}$ is a Haar random state. This means that the ensemble $\{\ket{\Phi_{\boldsymbol{\theta}}^{\vec A}}\}_{\boldsymbol{\theta},\vec A}$ gives rise to a so-called \emph{pseudorandom state generator} (PRSG) as in \Cref{{def:prsg}}.
\end{itemize}

We can call the two assumptions above the search (respectively, decision) variant of \emph{Hamiltonian Phase State} assumption ($\mathsf{HPS}_{n,m,q,\chi}$). Here, $n,m \in \N$ are circuit parameters, $q$ is a discretization parameter for the interval $[0,2\pi)$, and $\chi$ is a distribution over the choice of matrix $\vec A \in \Z_2^{m \times n}$; typically, $\chi$ is chosen to be the uniform distribution. 

There is some evidence that $\mathsf{HPS}_{n,m,q,\chi}$ is a reasonable assumption for constructing pseudorandom states.
Brakerski and Shmueli~\cite{brakerski2019pseudo} show that the states $D H^{\otimes n} \ket{0^n}$ form a state $t$-design when the diagonal operator $D$ consists of a $2t$-wise independent binary phase operator.
Previously, Nakata, Koashi and Murao~\cite{Nakata_2014} also showed that the states $D H^{\otimes n} \ket{0^n}$, where $D$ is a diagonal operator composed of appropriate diagonal gates with random phases, form a $t$-design. Starting from this intuition, we now provide rigorous evidence for the hardness of the $\mathsf{HPS}_{n,m,q,\chi}$ assumption.

\subsection{Overview of our Results}

In this work, we establish
$\mathsf{HPS}_{n,m,q,\chi}$ as a well-founded quantum computational assumption. Specifically, we address each of the meta-criteria we mentioned before:
\begin{itemize}
    \item (Evidence of worst-case hardness) The learnability of ensembles of phase states has been studied extensively in recent work~\cite{arunachalam2023optimal}, and has been found to have exponential time complexity in the worst case (despite only having polynomial sample complexity). We give a detailed discussion in \Cref{sssec:solving hps}.

    \item (Notion of a random instance) A random Hamiltonian phase state $\ket{\Phi_{\boldsymbol{\theta}}^{\vec A}}$, e.g., as in~\cref{eq:HPS-state}, is naturally defined in terms of a random binary matrix $\vec A \leftarrow \Z_q^{m \times n}$ and a random set of angles $\boldsymbol{\theta} = (\theta_1, \dots, \theta_m)$. Hence, it can be efficiently generated by a simple quantum circuit comprising $O(m/n \cdot n^2) $ \textsf{CNOT} gates, $n$ Hadamard gates, and $m$ single-qubit $Z$ rotations.

    \item (Evidence of average-case hardness) Our learning task admits a worst-case to average-case reduction. In \Cref{sec:worst-to-average}, we separately show how to re-randomize the architecture and the set of angles. Therefore, the hardness of our problem reduces to a worst-case version of the problem.

    \item (Relation to other problems) In \Cref{sec:worst-to-average}, we draw a connection between the task of learning Hamiltonian Phase states and the security of classical \emph{Goppa codes} and the well-known McEliece cryptosystem.

    \item (Cryptographic applications) 
    Hamiltonian Phase states have a sufficient amount of structure which suffices to construct a number of interesting cryptographic primitives which we sketch in detail in \Cref{sec:intro-app}.
\end{itemize}
In \Cref{ssec:fully quantumness of hps}, we also provide evidence that $\mathsf{HPS}_{n,m,q,\chi}$ is plausibly fully-quantum and does not allow one to construct one-way functions. In particular, we note that the result of \cite{kretschmer2021quantum} indicates that the idealized versions of any assumption that yields \emph{only} pseudorandom states can not be used to build one-way functions in a black-box way.  We further discuss the implications of the fact that HPS states are state $t$-designs on this reduction, noting that the resulting concentration properties by themselves rule out one-way function constructions that do not simulateneously measure many copies of the HPS state.
For all of the primitives we construct, in addition to just constructing these primitives from our hardness assumption, we argue that constructing them from our hardness assumption yields more efficient and practical implementations of these primitives (if and when fault-tolerant quantum computers become widely available).

\subsection{Applications}\label{sec:intro-app}

In this section, we give an overview of all the applications which are enabled by the HPS assumption. Besides the natural application of constructing efficient one-way state generators (see \Cref{def:OWSG}) and pseudorandom state generators (see \Cref{def:prsg}), which essentially follow by definition of our assumption, we also construct a number of other interesting applications that are relevant in quantum information science more broadly.

\paragraph{Quantum Trapdoor Functions and Public-Key Encryption with Quantum Public Keys.}

Recent work of Coladangelo~\cite{coladangelo2023quantumtrapdoorfunctionsclassical} introduced the notion of a quantum trapdoor function (QTF). This primitive is essentially a variant of a one-way state generator that also features a secret trapdoor which makes inversion possible. QTFs are interesting in the sense that they \emph{almost} enable public-key encryption: two parties can communicate classical messages over a quantum channel without ever exchanging a shared key in advance---the only caveat being that this requires the public keys to be quantum states~\cite{coladangelo2023quantumtrapdoorfunctionsclassical}. Using a construction based on binary-phase states,
Coladangelo~\cite{coladangelo2023quantumtrapdoorfunctionsclassical} showed that quantum trapdoor functions exist, if post-quantum one-way functions exist. However, to this day, it remains unclear how to construct QTFs from assumptions which are potentially weaker than one-way functions, such as the existence of pseudorandom states. 

In \Cref{sec:QTF}, we show how to construct QTFs from our (decisional) HPS assumption, which yields the first construction of QTFs from an assumption which is plausibly weaker than that of one-way functions. We believe that this application strongly highlights the versatility of Hamiltonian Phase states in the context of quantum cryptography; for example, it is far less clear how to construct QTFs from other, less structured, assumptions such as genuinely random quantum circuits via \Cref{conj:rqc}.

\paragraph{Quantum Pseudoentanglement.}

The notion of pseudoentanglement~\cite{aaronson_et_al:LIPIcs.ITCS.2024.2,bouland_et_al:LIPIcs.CCC.2024.21} has found many applications in quantum physics, for example to study the
emergence of thermal
equilibria in isolated many-body systems~\cite{feng2024dynamicspseudoentanglement}. Pseudoentangled
states have also been viewed as a potential tool for probing computational aspects of the AdS/CFT correspondence, which physicists believe may shed insight onto the behavior of black holes in certain simplified models of the universe. 
We note that it is currently not known how to construct these from \emph{any} assumption other than one-way functions.

In \Cref{sec:pseudoentanglement} we give a construction of pseudoentangled states from our HPS assumption, which yields the first construction of pseudoentanglement from an assumption which is plausibly weaker than that of one-way functions. Our proof sheds new light on the entanglement properties of random IQP circuits more generally.\footnote{To the best of our knowledge, such bounds for random IQP circuits were previously not known.} Therefore, we belief that this contribution is of independent interest.
Moreover, as we point out in the next section, our construction is also highly efficient and could enable implementations of quantum pseudoentanglement in practical scenarios.

\paragraph{Pseudorandom Unitaries.}

Pseudorandom unitaries are families of unitaries that are indistinguishable from Haar random unitaries in the presence of computationally bounded adversaries.  They are widely considered the most powerful fully-quantum primitive, and there has been a long line of work towards constructing them from the existence of one-way functions~\cite{ananth2022pseudorandom,brakerski2024real,metger2024simpleconstructionslineardepthtdesigns,chen2024efficient}, eventually resulting in the most recent breakthrough result by Ma and Huang~\cite{ma2024pseudorandom}.  

The result of \cite{ma2024pseudorandom} show that the ensemble of unitaries, colloquially known as the $\mathsf{PFC}$-ensemble~\cite{metger2024simpleconstructionslineardepthtdesigns}, form an approximation to a Haar random unitary.  However, this construction is not well suited for the HPS assumption, which, in some sense, provides a pseudorandom \emph{diagonal} unitary.  \ifhaspru Using the techniques introduced in~\cite{ma2024pseudorandom}, we show that another pseudorandom construction from \cite{ji2018pseudorandom}, the $\mathsf{FHFHFC}$-ensemble is also an approximation to a Haar random unitary when both functions are taken to be uniformly random binary phase oracles. This resolves a conjecture proposed by Ji, Liu and Song~\cite{ji2018pseudorandom}. We further show how diagonal matrices can be used to replace the Clifford operation with more applications of a random diagonal matrix and Hadamard gates.  With this result, we provide a plausible construction of efficient pseudorandom unitaries from the HPS assumption: alternating applications of HPS unitaries and Hadamards.\else \cite{ji2018pseudorandom} presented an alternative construction of pseudorandom unitaries that involves alternating random diagonal gates and $n$-qubit Hadamard gates.  We propose that instantiating the random diagonal gates using randomly sampled HPS instances yields a computationally secure pseudo-random unitary\fi See \Cref{sec:unitary} for more details.

\subsection{Physical Implementations}

Hamiltonian Phase States with $m$ terms on $n$ qubits can be generated very efficiently compared to phase states constructed from pseudorandom functions: 
to prepare a HPS, we require only a layer of Hadamard gates, followed by $\lceil m/n \rceil$ alternating layers of single-qubit $Z$ rotations and CNOT circuits. 
To see this, we observe two facts. First, 
\begin{align}
   \textsf{CNOT}_{k,l} e^{i \theta Z_l}  = e^{i \theta Z_k Z_l} \textsf{CNOT}_{k,l},
\end{align}
where $\textsf{CNOT}_{k,l}$ is  controlled on qubit $k$ and targeted  on qubit $l$.
Second, 
\begin{align}
    \textsf{CNOT}_{k,l} \ket{+^n} = \ket{+^n}. 
\end{align}
More generally, if $C$ is a circuit comprised of \textsf{CNOT} gates, then $C \ket x = \ket{\vec C \cdot x}$, where $\vec C \in \mathrm {GL}(n, \mathbb Z_2)$ is an invertible binary matrix. 
Given an HPS, decompose its architecture matrix $\vec A \in \mb Z_2^{n \times m}$ into $\lceil m/n \rceil$ submatrices of $n$ rows (except for the last one), and suppose each of those submatrices has full rank.
The HPS $\ket{\Phi_{\boldsymbol{\theta}}^{\vec A}}$ can then be prepared as
\begin{align}
     \ket{\Phi_{\boldsymbol{\theta}}^{\vec A}} = 
     \left ( C_{\lceil m/n \rceil} \prod_{i=(\lceil m/n \rceil-1)n}^m 
     e^{i \theta_i Z_{i}}\right) \textbf{}
     \cdots \left(C_1 \prod_{i=1}^n e^{i \theta_i Z_i}\right)\ket {+^n} 
\end{align} 
where the CNOT circuits $C_k$ are chosen such that the first $n$ rows of $\vec A$ are given by $\vec C_{\lceil m/n \rceil} \cdots \vec C_1$, the second $n$ rows by $\vec C_{\lceil m/n \rceil} \cdots \vec C_2$, and so on, see \cref{fig:hps} for an example.
If the rank condition above is not satisfied, decompose $\vec A$ into the minimal number $\ell$ of submatrices with full rank, and proceed as above. 
The smallest meaningful example of such states---with $n$ random, linearly independent terms---can thus be prepared using a single layer of rotations and a  CNOT circuit.
By the fact that $\mathrm{GL}(n,\mathbb Z_2)$ is a group, uniformly random architecture matrices $\vec C_i$ and phases $\theta_i$  generate a uniformly random HPS. 

This protocol for the implementation of HPS is also interesting from an early fault-tolerance perspective, since there are quantum codes in which all required operations are transversal, yielding a highly efficient fault-tolerant implementation. 
To see this, consider a $q=2^d$-fold discretization of the unit circle.  
Now, we observe that there are $d$-dimensional CSS codes with a transversal $Z^{1/2^{d-1}}$ gate such as the $[[ 2^d-1,1,3]]$ simplex code \cite{zeng_local_2007}. 
By the fact that they are  CSS codes, they also admit a transversal CNOT gate between code blocks. 
This means that HPS can be prepared using transversal in-block $Z^{1/2^{d-1}}$ gates as well as inter-block CNOT gates, making them amenable to implementations in early fault-tolerant architectures such as reconfigurable atom arrays \cite{bluvstein_logical_2024}, or trapped ion processors \cite{ryan-anderson_high-fidelity_2024,reichardt_demonstration_2024} in which arbitrary inter-block connectivities can be achieved.

\subsection{Related Work}

\paragraph{Minimal assumptions in cryptography.}
The study of minimal assumptions in quantum cryptography has a long history, dating back to the exploration of the possibility of unconditional cryptography in the 1980ies and 90ies~\cite{bennett1997strengths,lo1997quantum,mayers2001unconditional,lo1999unconditional}.  Modern research into minimal assumptions largely began with the work of \cite{ji2018pseudorandom}, who first proposed the idea of pseudo-random states and unitaries, which are classically samplable ensembles of pure states and unitaries that are indistinguishable from states and unitaries drawn from the Haar measure.  They also gave a construction of pseudo-random states from post-quantum one-way functions.  Since then, a number of works have proposed quantum analogs of minimal classical assumptions.  
\cite{morimae2022one} proposed the one-way state generator, a quantum analog of one-way functions.  \cite{brakerski2022computational} then proposed the EFI pair as a minimal quantum assumption, showing that it was required for many "useful" primitives, and posing the question of finding a concrete assumption that implies EFI pairs.  Since then, a number of works have provided oracle separations and reductions between these primitives~\cite{khurana2024commitments,batra2024commitments,bostanci2024oracle,behera2024oracle}, showing that unlike in classical cryptography, quantum cryptographic primitives seem to have distinct powers relative to each other.  

While almost all quantum cryptographic primitives are implied by some form of classical cryptography, it was not clear whether quantum primitives were strictly weaker than their classical counterparts.  
This question was resolved in \cite{kretschmer2021quantum}, which gave a quantum oracle relative to which one-way functions do not exist, but pseudo-random unitaries do.  Later this was extended to a classical oracle~\cite{kretschmer2023quantum}, although relative to this oracle only EFI pairs exist.  

Recently, there has been a large body of research constructing pseudo-random unitaries from one-way functions.  \cite{ji2018pseudorandom} conjectured several constructions, but were unable to prove their security from one-way functions.  \cite{metger2024simpleconstructionslineardepthtdesigns} proposed a new construction, the $\mathsf{PFC}$ ensemble, which was later proven to be a secure pseudo-random unitary family by \cite{ma2024pseudorandom}.  The recent work of \cite{ma2024pseudorandom} uses a new path recording framework, which we use heavily in our work.

Even more recently, a work has come out motivating the existence of quantum cryptography from conjectures related to quantum supremacy.  Namely, the work of \cite{khurana2024foundingquantumcryptographyquantum} argues that a sampling problem for which there exists a polynomial-time quantum sampler but no classical sampler that can approximate the output distribution (even to a very high error) implies one-way puzzles.

\paragraph{Concrete assumptions for quantum cryptography.}

A number of prior works have proposed assumptions that could be candidate concrete assumptions for building quantum cryptography.  The most common approach to building quantum crptography is to assume that random polynomial-depth quantum circuits become pseudo-random unitary or state ensembles after a certain depth.  The motivation to use random quantum circuits comes from their extensive study as related to quantum su\-pre\-macy \cite{bremner_classical_2010,aaronson_complexity-theoretic_2017,hangleiter_computational_2023}, as well as from the theory of black holes~\cite{brown2018second, brandao2021models}, which posits that random circuits should have ``scrambling'' properties similar to those expected in black holes.  \cite{bouland2019computational} proposes an even more bold foundation for quantum cryptography.  By formalizing assumptions and conjectures from physics, such as "complexity = volume" and the AdS-CFT correspondence, they argue that time evolution by a the Hamiltonian of a black hole in the CFT gives rise to pseudo-random states.

Another interesting approach to quantum cryptography is the result of \cite{kawachi2005computational}.  The authors propose the hardness of distinguishing signed coset states of the symmetric group as a basis for cryptography.  
They show a number of interesting properties of these states, firstly they prove a worst-case to average case reduction.  They reduce the hardness of this distinguishing task to the hardness of graph automorphism, which in turn is related to the hardness of the hidden subgroup problem on the symmetric group.  Finally, they build a public-key cryptography system with quantum keys (both secret and private), but for which trusted set-up is needed.  

\paragraph{$t$-designs and classical reversible circuits.}
% 
% \ \\
% \alex{Cite more quantum design papers here?}
% \jonas{haha, with pleasure!}
A common practice in cryptography used to justify the security of primitives is to prove an approximate $t$-wise independence property.  In particular, a long line of works have steadily improved the bounds on the number of rounds of a block cipher are needed to be applied in order to get certain $t$-wise independence~\cite{liu2021t, liu2023layout}.  
This is a popular justification for cryptographic security because it rules out a large class of natural attacks, such as linear attacks ($2$-wise independence rules this out), and differential attacks ($t$-wise independence rules out $\log_2(t)$ differential attacks).  

Recently, the study of random reversible circuits has led to a number of exciting results pertaining to pseudorandomness~\cite{he2024pseudorandom} and obsfucation~\cite{canetti2024towards}.  These results, motivated by early work on the $t$-wise independence of bounded depth reversible circuite~\cite{gowers1996almost,hoory2005simple,brodsky2008simple,chen2024incompressibility,gretta2024more,he2024pseudorandom}, have used the structure of random circuits, combined with an assumption that they are pseudo-random functions when given random inputs, to build indistinguishability obfuscators.  

Here, we show the approximate state $t$-design property for HPS in~\Cref{section:designbounds} with $O(nt^2)$ random $Z$-rotations.
Enormous effort went into proving the design property over various ensembles of circuits.
Notably, improving over previous bounds~\cite{brandao2016local,haferkamp2022random}, a near optimal bound (in $t$) of $\mathrm{poly}(n)t$ was found in Refs.~\cite{chen2024incompressibility}.
Previously, other constructions in depth $\mathrm{poly}(n)t$ where found in Refs.~\cite{chen2024efficient,metger2024simpleconstructionslineardepthtdesigns} for the weaker notion of additive error approximate designs.
We note that, while we do not achieve a linear scaling in $t$, our bound comes with very realistic constants and a surprisingly simple proof based on the analysis of Boolean functions.
Surprisingly, quantum circuits allow for a further, exponential improvement in the $n$-dependence, which is provably impossible for reversible circuits~\cite{schuster_random_2024}.
More concretely, Ref.~\cite{schuster_random_2024} shows that approximate unitary $t$-designs can be generated in depth $\tilde{O}(t\log(n))$, inheriting the $t$-dependence from earlier work~\cite{chen2024incompressibility}.
Last, very similar bounds of $O(nt^2)$ ``Pauli rotations'' was proven for the approximate unitary $t$-design property  in Ref.~\cite{haah2024efficient} by a non-trivial generalization of the argument presented here.\\
\\
\noindent\textbf{Note:} The results on the approximate state $t$-design property presented in~\Cref{section:designbounds} are contained in an earlier preprint~\cite{haferkamp_moments_2023}, which this paper supersedes.

\subsection{Concurrent Work}

Here we discuss concurrent work on building fully-quantum cryptography and the connections to this work.

\paragraph{The Hardness of Quantum Learning
and its Cryptographic Applications.} Ref.~\cite{fefferman2024hardness} proposes a pair of computational assumptions related to random circuits and construct quantum cryptography from them.  Specifically, they conjecture that the problem of learning the description of a random circuit of depth $\log^2(n)$, from copies of its output state, and the problem of cloning the output of a random circuit, are computationally hard problems.  The authors motivate this problem in a idealized model where the output state is truly Haar random, and subsequently build a number of cryptographic primitives from it.  We note that while our suite of results seem similar, there are large differences in both the assumptions themselves: random brickwork circuits versus Hamiltonian phase states, how we motivate the hardness of our problems: a query lower bound in an idealize model versus worst-to-average case reduction, analysis of design properties and analysis of best known algorithms, and the kinds of cryptography we build from it: OWSG and primitives implied by OWSG as well as noise-resistant versions of these primitives, compared our to public key cryptography and pseudorandom unitaries. Both works are interesting explorations into how fully-quantum cryptography might be realized in practice.

\paragraph{Quantum cryptography from \#P-hard problems.}

The recent and concurrent work of Khurana and Tomer~\cite{khurana2024foundingquantumcryptographyquantum} argues that one-way puzzles are equivalent to the hardness of average case sampling problems.  In particular, they point out that conjectures regarding the hardness of sampling the output distribution of random IQP circuits should imply quantum bit commitments and one-way state synthesis problems.  We believe that because our computational assumption takes a quantum input, and because it implies pseudo-random states instead of a classical communication protocol, our assumption is more plausibly ``fully-quantum''.  Additionally, we are able to build a wide plethora of quantum cryptographic primitives, as opposed to just building commitments and one-way puzzles.

\paragraph{Quantum cryptography from meta-complexity.}  The recent works of \cite{hiroka2024quantumcryptographymetacomplexity, cavalar2024meta} also constructs one-way puzzles (and therefore quantum bit commitments) from a different classical complexity theoretic assumption, namely the average case hardness of the gapped Kolmogorov complexity problem.  They do not use IQP circuit sampling or random circuit sampling to instantiate their one-way puzzles, and also do not construct primitives like quantum pseudo-random states.

\paragraph{Computational complexity of learning pure states.} The recent work of \cite{hiroka2024computationalcomplexitylearningefficiently} constructs EFI pairs and one-way state generators from the existance of learning problems that exhibit average-case hardness. 
We note that this work does not propose a concrete family of states for which learning might be hard, and by exploiting the properties of our state family, we are able to create quantum cryptographic primitives beyond one-way state generators and EFI pairs.

\paragraph{Quantum group actions.}
The recent work of \cite{morimae2024quantumgroupactions} propose another way to build quantum cryptography from so called ``one-way quantum group actions'', which generalize group actions to efficiently describable group representations (and starting states).  One candidate quantum group action they propose is a random $Z$-diagonal circuit, which bares resemblance to the HPS assumption.  The authors are also able to construct pseudo-random function-like state generators, but provide no evidence of the one-wayness of these random $Z$-diagonal circuits.  
\ifanon
\else
\section*{Acknowledgements}
This work was done in part while the authors were visiting the Simons Institute for the Theory of Computing, supported by NSF QLCI Grant No. 2016245.

The authors would like to thank Fermi Ma for pointing out a bug in an earlier version of the paper related to the pseudorandom states construction, and for helpful discussions.  The authors would also like to thank Aram Harrow, David Gross, Eli Goldin, Henry Yuen, Kabir Tomer, Luowen Qian, Natalie Parham, Ran Canetti, Saachi Mutreja, Soonwon Choi, Thiago Bergamaschi, Uma Girish and Vinod Vaikuntanathan for useful discussions.

JB is supported by Henry Yuen's AFORS (award FA9550-21-1-036) and NSF CAREER (award CCF2144219).
JH acknowledges funding from the Harvard Quantum Initiative postdoctoral fellowship.
DH acknowledges funding from the US Department of Defense through a QuICS Hartree fellowship. 
AP is supported by the U.S. Department of Energy, Office of Science, National Quantum Information Science Research Centers, Co-design Center for Quantum Advantage (C2QA) under contract number DE-SC0012704.
\fi

%!TEX root=./main.tex
\section{Background}
\label{sec:background}

\subsection{Notation}

The notation $x \sim X$ describes that an element $x$ is drawn uniformly at random from the set $X$. Similarly, if $\algo D$ is a distribution, we let
$x \sim \algo D$ denote sampling $x$ according to $\algo D$.
We denote the expectation value of a random variable $X$ 
by $\mathbb{E}[X] = \sum_{x} x \Pr[ X = x] $.

\subsection{Quantum Preliminaries}

A quantum \emph{register} $\reg{R}$ corresponds to a finite-dimensional complex Hilbert space.  We write $L(\reg{R})$ to denote the set of linear transformations on $\reg{R}$, $\mathrm{GL}(\reg{R})$ to denote the set of invertible linear transformations on $\reg{R}$, $\states(\reg{R})$ to denote the Hermitian, positive semi-definite matrices over $\reg{R}$, and $U(\reg{R})$ to denote the unitary group over $\reg{R}$.  When the register is clear from context, we sometimes write $U(d)$, where $d$ is the dimension of $\reg{R}$.  A linear transformation is called a quantum channel if it is completely positive and trace-preserving (i.e. for every positive input with trace $t$, the output of the map is a positive matrix with trace $t$).  

For a vector $\ket{\psi} \in \reg{R}$, we write $\psi$ to denote the density matrix $\proj{\psi}$, and for vectors $\ket{\psi}, \ket{\phi} \in \reg{R}$, we write $\braket{\psi|\phi}$ to denote the inner product.  For a vector $\ket{\psi} \in \reg{R}$, we write $\norm{\ket{\psi}}$ to denote the standard norm over $\reg{R}$, i.e. $\braket{\psi | \psi}$.  We write $\Tr(\cdot)$ to denote the trace and $\Tr_{\reg{R}}$ to denote the partial trace over a register $\reg{R}$.

For a linear operator $X \in L(\reg{R})$, let $\opnorm{X}$ be its operator norm and $\norm{X}_{1} = \Tr(\sqrt{X^{\dagger}X})$ be its trace norm.  For two density matrices $\rho, \sigma \in \states(\reg{R})$, let $\td(\rho, \sigma) = \frac{1}{2}\norm{\rho - \sigma}_{1}$ be the trace distance between the two.  We sometimes write $X_\reg{R}$ to indicate that $X$ acts on $\reg{R}$.  All un-labeled operators act on all registers that do not have an operator acting on them, and if an operator is associated with specific registers, we drop the register subscripts for brevity.

For two quantum channels $\mathcal{M}$ and $\mathcal{N}$, the diamond distance between them is
\begin{equation*}
    \norm{\mathcal{M} - \mathcal{N}}_{\diamond} = \max_{\ket{\psi} \in \reg{R}^{\otimes 2}} \norm{(\mathcal{M}\otimes \id_{\reg{R}})(\proj{\psi}) - (\mathcal{N}\otimes \id_{\reg{R}})(\proj{\psi})}_{1}\,,
\end{equation*}
where $\id_{\reg{R}}$ is the identity channel on $\reg{R}$.

\subsection{Useful Lemmas}

Here we write lemmas that will be important in the rest of the paper, but should be familiar to someone well-versed in quantum computation.
\begin{lemma}[Gentle measurement lemma~\cite{winter1999coding}]
    \label{lemma:gentle-measurement}
    Given a pure state $\rho$ and a projector $\Lambda$, let 
    \begin{equation*}
        \rho' = \frac{\Lambda \rho \Lambda}{\Tr(\Lambda \rho)}
    \end{equation*}
    be the post-measurement state.  Then $\td(\rho',\rho) \leq \sqrt{1 - \Tr(\Lambda \rho)}$.
\end{lemma}

We will also need a similar lemma, but where the projector acts on a register being traced out, and normalization is not applied. 
A similar lemma appears in \cite{ma2024pseudorandom}, we state the proof here for completeness.
\begin{lemma}
    \label{lem:un-normalized_gentle_measurement}
    Let $\Lambda_{\reg{B}}$ be a projector and $\rho_{\reg{AB}}$ be a quantum state, then the following holds
    \begin{equation*}
        \norm{\Tr_{\reg{B}}(\rho_{\reg{AB}}) - \Tr_{\reg{B}}(\Lambda_{\reg{B}} \rho_{\reg{AB}} \Lambda_{\reg{B}})}_1 = 1 - \Tr(\Lambda_{\reg{B}} \rho_{\reg{AB}})\,.
    \end{equation*}
\end{lemma}
\begin{proof}
    As $\reg{B}$ is being traced out, we can write the following.
    \begin{equation*}
        \Tr_{\reg{B}}(\rho_{\reg{AB}}) = \Tr_{\reg{B}}(\Lambda_{\reg{B}} \rho_{\reg{AB}} \Lambda_{\reg{B}}) + \Tr_{\reg{B}}((\id - \Lambda_{\reg{B}}) \rho_{\reg{AB}} (\id - \Lambda_{\reg{B}}))\,.
    \end{equation*}
    Then we can write the left hand side of the theorem statement as follows.
    \begin{align*}
        \norm{\Tr_{\reg{B}}(\rho_{\reg{AB}}) - \Tr_{\reg{B}}(\Lambda_{\reg{B}} \rho_{\reg{AB}} \Lambda_{\reg{B}})}_1 &= \norm{\Tr_{\reg{B}}((\id - \Lambda_{\reg{B}}) \rho_{\reg{AB}} (\id - \Lambda_{\reg{B}}))}_1\\
        &= \Tr((\id - \Lambda_{\reg{B}}) \rho_{\reg{AB}} (\id - \Lambda_{\reg{B}}))\\
        &= \Tr((\id - \Lambda_{\reg{B}}) \rho_{\reg{AB}})\\
        &= 1 - \Tr(\Lambda_{\reg{B}} \rho_{\reg{AB}})\,.
    \end{align*}
    Here the first line uses the previous observation.  Then we use the fact that $\norm{X}_1 = \Tr(X)$ whenever $X$ is positive-semidefinite.  Then we use the cyclic property of the trace and the fact that $\id - \Lambda$ is a projector (and thus squares to itself).  Finally, we use the fact that $\rho$ is trace $1$.
\end{proof}

\subsection{The Haar Measure and $t$-designs}

Here we define the Haar measure and quantum $t$-designs.
\begin{definition}[Haar measure and Haar random states]
    The Haar measure is the unique left- and right- invariant probability measure on the unitary group $U(d)$. A Haar random state is a state sampled by first sampling a unitary from the Haar measure and then applying it to $\ket{0}$ (or any fixed state). 
    We use the notation $\mu(d)$ to denote the Haar measure on $d \times d$ unitary matrices, and $\mathrm{Haar}(d)$ to denote the distribution over Haar random states.
    Similarly, we denote by $\mu_D(d)$ the unique left- and right- invariant probability measure on the diagonal subgroup $D(d)\subset U(d)$. 
    Moreover, the Haar measure induces a unique unitarily invariant probability measure on states in $\mathbb{C}^d$, which we will also denote by $\mu(d)$.
\end{definition}

An approximate $t$-design is a discrete distribution over unitaries that looks identical to $t$-calls to a unitary sampled from the Haar measure. 
\begin{definition}[$t$'th moment operator]
    Let $\nu$ be a probability distribution over unitaries, then the $t$'th moment operator is defined as 
    \begin{equation*}
        \mathcal{M}_{\nu}^{(t)}(\cdot) = \mathbb{E}_{U \sim \nu} (\cdot) \left[U^{\otimes t} (\cdot) U^{\dagger, \otimes t}\right]\,.
    \end{equation*}
\end{definition}

\begin{definition}[Approximate unitary and $t$-design]\label{def: approximate designs}
    An ensemble of $d \times d$ unitaries $\nu$ in is called an $\epsilon$-approximate $t$-design if 
    \begin{equation*}
        \norm{\mathcal{M}_{\nu}^{(t)} - \mathcal{M}_{\mu(d)}^{(t)}}_{\diamond} \leq \epsilon\,.
    \end{equation*}
    Similarly, an ensemble $\nu$ of states in $\mathbb{C}^d$ is a $\varepsilon$-approximate state $t$-design if 
    \begin{equation}
        \left|\left|\underset{\psi\sim\nu}{\mathbb{E}}(|\psi\rangle\langle\psi|)^{\otimes t}-\underset{\ket{\psi}\sim\mu(d)}{\mathbb{E}}(|\psi\rangle\langle\psi|)^{\otimes t}\right|\right|_{1}\leq \varepsilon.
    \end{equation}
\end{definition}
We will also need a notion of approximate designs on the diagonal subgroup. 
\begin{definition}[Approximate diagonal $t$-design]\label{def: diagonal designs}
    An ensemble of $d \times d$ diagonal unitaries $\nu$ in is called an $\epsilon$-approximate diagonal $t$-design if 
    \begin{equation*}
        \norm{\mathcal{M}_{\nu}^{(t)} - \mathcal{M}_{\mu_D(d)}^{(t)}}_{\diamond} \leq \epsilon\,.
    \end{equation*}
    Here $\mu_D$ is the Haar measure (unique left- and right- invariant measure) over diagonal unitary matrices.
\end{definition}

\subsection{Cryptographic Primitives}

Here we define a number of cryptographic primitives.

\paragraph{One-Way State Generator.}
First, we define a one-way state generator~\cite{morimae2022one}.  These are quantum equivalents of one-way functions, with classical input and quantum output. 
\begin{definition}[One-Way State Generator]\label{def:OWSG}
Let $n \in \N$ denote the security parameter.  A one-way state generatator is a tuple $(\mathsf{KeyGen},\mathsf{StateGen},\mathsf{Ver})$ consisting of the following alorithms:
\begin{itemize}
    \item $\mathsf{KeyGen}(1^n) \rightarrow \mathsf{k}$: on input $1^n$, it outputs a classical key $k$.

    \item $\mathsf{StateGen}(k) \rightarrow \ket{\phi_k}$: on input $k$, it outputs a quantum state $\ket{\phi_k}$.

    \item $\mathsf{Ver}(k,\sigma) \rightarrow \top/\bot$: on input a key $k$ and a state $\sigma$, it outputs $\top$ or $\bot$.
\end{itemize}
We require that the algorithms satisfy the following properties:
\begin{itemize}
    \item[(a)] \textbf{Correctness:} For any $n \in \N$, it holds that
    $$
    \Pr\left[\top \leftarrow \mathsf{Ver}(k,\ket{\phi_k}) \, : \, \substack{
    k \leftarrow \mathsf{KeyGen}(1^n)\\
\ket{\phi_k} \leftarrow \mathsf{StateGen}(k)
    }\right] \geq  1-\negl(n).
    $$

\item[(b)] \textbf{Security:} For any $n \in \N$, any $t=\poly(n)$, and for any QPT adversary $\mathcal{A}$,
$$
\Pr\left[\top \leftarrow \mathsf{Ver}(k',\ket{\phi_k}) \, : \, \substack{
    k \leftarrow \mathsf{KeyGen}(1^n)\\
\ket{\phi_k} \leftarrow \mathsf{StateGen}(k)\\
k' \leftarrow \mathcal{A}(\ket{\phi_k}^{\otimes t} )
}\right] 
\leq \negl(n) \, ,
$$
where $\mathsf{Ver}(k',\ket{\phi_k})$ denotes the algorithm which applies the projective measurement $$
\{\proj{\phi_{k'}}, I - \proj{\phi_{k'}}\}$$
onto $\ket{\phi_k}$ and outputs $\top$, if the measurement succeeds, and outputs $\bot$ otherwise. 
\end{itemize}
\end{definition}

\paragraph{Pseudorandom State Generator.}

Pseudorandom states are collections of states that are indistinguishable from a Haar random state to bounded adversaries.  These were first defined in \cite{ji2018pseudorandom}.
\begin{definition}[Pseudorandom State Generator]\label{def:prsg}
Let $n \in \N$ be the security parameter. A pseudorandom state generator (PRSG) is a tuple $(\mathsf{KeyGen},\mathsf{StateGen})$ of QPT algorithms:
    \begin{itemize}
\item $\mathsf{KeyGen}(1^n) \rightarrow k$: on input $1^n$, it outputs a classical key $k$.
 \item $\mathsf{StateGen}(k) \rightarrow \ket{\phi_k}$: on input $k$, it outputs a quantum state $\ket{\phi_k}$.

    \end{itemize}
We require that the algorithms satisfy the following security property:
For all $n \in \N$ and every polynomial function $t = \mathsf{poly}(n)$, and every QPT adversary $\mathcal{A}$, 
\begin{align*}
        \big|\Pr_{k \gets \mathsf{KeyGen}(1^n)}[\mathcal{A}(\ket{\phi_k}^{\otimes t})=1] - 
        \Pr_{\ket{\phi} \sim \mathrm{Haar}(2^n)}[\mathcal{A}(\ket{\phi}^{\otimes t})=1] \big| \leq \negl(n) 
        \end{align*}
\end{definition}

\section{Hamiltonian Phase State Assumption}
\label{sec:assumption}

In this section, we give a formal definition of our hardness assumption.
Recall that an $n$-qubit
Hamiltonian Phase State is of the form
\begin{align*}
\ket{\Phi_{\boldsymbol{\theta}}^{\vec A}} = \exp\left(\mathrm{i} \, \sum_{i=1}^m \theta_i \bigotimes_{j=1}^n \mathsf{Z}^{\vec A_{ij}}  \right) H^{\otimes n} \ket{0^n}
\end{align*}
where $\vec A \in \Z_2^{m \times n}$ is a binary matrix and $\boldsymbol{\theta} = (\theta_1, \dots, \theta_m)$ is a set of angles in the interval $[0,2\pi)$. To avoid matters of precision, we introduce a discretization parameter $q \in \N$ with $q=\poly(n)$ and partition the interval $[0,2\pi)$ into $q$ parts via the set
$$
\Theta_q := \left\{\frac{2\pi k}{q} \, : \, k \in \{0,1,\dots,q-1\}\right\}.
$$

\noindent We now introduce two variants of our hardness assumption.

\subsection{Search Variant}

Our first variant considers a search problem. Roughly speaking, it says that given many copies of a random Hamiltonian phase state, it is computationally difficult to reverse-engineer its architecture and its angles. Therefore, our assumption says that an ensemble of Hamiltonian Phase states forms a one-way state generator~\cite{morimae2022one}. We provide a formal definition of a one-way state generator in \Cref{def:OWSG}.

We now give a formal definition.

\begin{definition}[Search HPS]\label{def:search-HPS}
Let $n \in \N$ denote the security parameter, and let $m$ and $q$ be integers (possibly depending on $n$). Let $\chi$ be a distribution with support over matrices in $\Z_2^{m \times n}$.
Then, the (search) Hamiltonian Phase State assumption ($\HPS_{n,m,q,\chi}$) states that, for any number of copies $t=\poly(n)$ and for any efficient quantum algorithm $\mathcal{A}$,
$$
\Pr\left[1 \leftarrow \mathsf{Ver}(\vec A',\boldsymbol{\theta}', \ket{\Phi_{\boldsymbol{\theta}}^{\vec A}}) \, : \, \substack{
\vec A \sim \chi, \, \boldsymbol{\theta} \sim \Theta_q^{m}\\
(\vec A',\boldsymbol{\theta}') \leftarrow \mathcal{A}(\ket{\Phi_{\boldsymbol{\theta}}^{\vec A}}^{\otimes t})
}\right] \leq \negl(n) \, ,
$$
where $\mathsf{Ver}(\vec A',\boldsymbol{\theta}', \ket{\Phi_{\boldsymbol{\theta}}^{\vec A}})$ denotes the algorithm which applies the projective measurement $$
\{\proj{\Phi_{\boldsymbol{\theta}'}^{\vec A'}}, I - \proj{\Phi_{\boldsymbol{\theta}'}^{\vec A'}}\}$$
onto $\ket{\Phi_{\boldsymbol{\theta}}^{\vec A}}$ and outputs $1$, if the measurement succeeds, and outputs $0$ otherwise. We say that a quantum algorithm solves the (search) $\HPS_{n,m,q,\chi}$ problem if it runs in time $\poly(n,m,\log q)$ and succeeds with probability at least $1/\poly(n,m,\log q)$.
\end{definition}

An alternative but equivalent formulation of the security property is to say that it is computationally difficult to find a state $\ket{\Phi_{\boldsymbol{\theta}'}^{\vec A'}}$ which has non-vanishing fidelity with the input state, on average over the choice of architecture and set of angles.

%\alex{There's a precision issue here which we may want to address. Instead of outputting the real-valued angles, the algorithm should probably just output integers $k_i \in \Z_q$ corresponding to the angle $\frac{2\pi k_i}{q} \in [0,2\pi)$. Any suggestions for how to write this without the notation blowing up? Maybe this is best handled with a remark? We can be a bit sloppy with notation as long as we explain it}
%\Dom{Definitely. Like, how small does $q$ need to be in order to distinguish the two parameters of the assumption. }

\subsection{Decision Variant}

Our second variant considers a decision problem. Roughly speaking, it says that given many copies of a random Hamiltonian phase state, it is computationally difficult to distinguish it from many copies of a Haar random state. Therefore, our (decision) assumption says that an ensemble of Hamiltonian Phase states forms a so-called pseudorandom state generator~\cite{ji2018pseudorandom,morimae2022one}. We provide a formal definition of a pseudorandom state generator in \Cref{def:prsg}.

\begin{definition}[Decision HPS]\label{def:decision-HPS}
Let $n \in \N$ denote the security parameter, and let $m$ and $q$ be integers (possibly depending on $n$). Let $\chi$ be a distribution with support over matrices in $\Z_2^{m \times n}$.
Then, the (decision) Hamiltonian Phase State assumption ($\HPS_{n,m,q,\chi}$) states that, for any number of copies $t=\poly(n)$ and for any efficient quantum distinguisher $\mathcal{D}$,
$$
\left| \Pr\left[1 \leftarrow \mathcal{D}(\ket{\Phi_{\boldsymbol{\theta}}^{\vec A}}^{\otimes t}) \, : \, \substack{
\vec A \sim \chi\\ 
\boldsymbol{\theta} \sim \Theta_q^{m}
}\right] - \Pr\left[1 \leftarrow \mathcal{D}(\ket{\Psi}^{\otimes t}) \, : \, \substack{
\ket{\Psi} \sim \mathrm{Haar}(2^n)
}\right] \right| \leq \negl(n) \, ,
$$
We say that a quantum algorithm solves the (decision) $\HPS_{n,m,q,\chi}$ problem if it runs in time $\poly(n,m,\log q)$ and succeeds with probability at least $1/\poly(n,m,\log q)$.
\end{definition}

%!TEX root=./main.tex
\section{Evidence for Average-Case Hardness and Full Quantumness}
\label{sec:hardness and quantumness}

In this section, we give several pieces of evidence for the security of the $\HPS_{n,m,q,\chi}$ assumption, as well as evidence that it is a fully quantum assumption. 

First, in \cref{sec:worst-to-average}, we show two worst-to-average-case reductions for the $\HPS_{n,m,q,\chi}$ problem, and also discuss the limitations of those reductions. 
To this end, we first show that if the Hamiltonian architecture matrix $\vec A$ is publicly known, then there is a worst-to-average-case reduction for the angles $\theta \in \Theta_q$
. 
Second, we show that for $m=n$, and any fixed set of angles, there is a worst-to-average-case reduction over the architecture matrices $\vec A$, and discuss limitations of the reduction. 
We argue that the HPS assumption is related to the security of the McEliece cryptographic system. 

Then, we show that if $\chi$ is the uniform distribution of $m \times n$ binary matrices and $q >2t$, Hamiltonian Phase States  with $m \gtrsim n t^2$ random terms form approximate state $t$-designs in~\cref{section:designbounds}. 
This shows that given less than $\Omega(\sqrt{m/n})$ many copies, HPS are information-theoretically indistinguishable from Haar-random states.
It also implies that the Hamiltonian Phase States contain an exponentially large set of almost orthogonal states.
This implies that Hamiltonian phase states are fast mixing, giving additional evidence that the learning problem is computationally hard. 

Then, in \cref{sssec:solving hps}, we discuss algorithms for learning phase states with public and secret architecture matrices. 
In particular, we give a sample-optimal (but exponential-time) algorithm for solving the $\HPS$ problem using pretty good measurements \cite{barnum2000reversingquantumdynamicsnearoptimal, montanaro2019pretty} and a simple algorithm that uses classical shadows \cite{aaronson2018shadow,huang2020predicting}. 

Finally, we discuss why the HPS assumption is fully quantum. 
To this end, we give evidence against the possibility of building one-way functions from $\mathsf{HPS}$.

%!TEX root=./main.tex
\subsection{Worst-Case to Average-Case Reduction}
\label{sec:worst-to-average}

We begin providing evidence for the security of our assumption by showing how in different regimes learning the parameters of the HPS problem of a fixed (worst-case) instance can be reduced to learning a random instance.
Our evidence will treat the angles $\theta$ and the Hamiltonian architecture matrix $\vec A$ separately.
Specifically, we will show two types of worst-to-average-case reductions. 
First, we will show that  in a certain regime of $m,n$, given a copy of a HPS instance, a quantum algorithm can efficiently generate a random HPS with the same angles and architecture dimensions.
Second, we will fix the Hamiltonian architecture $\vec A$ and show that given a copy of a HPS instance and its architecture $\vec A$, a quantum algorithm can generate a random HPS with the same architecture but uniformly random angles. 
Our worst-to-average-case reductions are therefore similar to those for, say, the Learning with Errors (LWE) problem~\cite{regev2009lattices}, with different levels of public knowledge. 

\paragraph{Reduction for the architecture for $m\leq n $}
First, we observe that for any fixed choice of angles $\boldsymbol \theta$, the Hamiltonian architecture can be re-randomized if $m\leq n$ and $\chi$ is the uniform distribution over full-rank matrices $\mathcal R(m,n) \coloneqq \{ \vec A \in \mathbb Z_2^{m \times n} \,|\, \mathrm{rank}(\vec A) = \min(m,n)\}$.
Notice that the restriction to Hamiltonian architectures with full rank is not too significant, since the probability that a uniformly random $\Z_2^{m\times n}$ matrix has full rank with probability\footnote{See \url{https://math.mit.edu/~dav/genlin.pdf}, and this \href{https://math.stackexchange.com/questions/71288/probability-of-a-random-n-times-n-matrix-over-mathbb-f-2-being-nonsingular}{Stackexchange post} for a proof.} $\prod_{k=1}^{\min(m,n)} (1- 2^{-k})\ge 0.288$ \cite{oeisA048651}. 
The basic idea of the reduction is to apply a circuit composed of uniformly random \textsf{CNOT} gates to the given HPS instance.
In the parameter regime we consider, this will have the effect of completely scrambling the Hamiltonian architecture to a uniformly random one with the same choices of $m,n$ and subject to the full-rank constraint.

\begin{lemma}[Worst-to-average-case reduction for the architecture]
\label{lem:architecture reduction}
Suppose there exists an algorithm $\mathcal{A}$ that runs in time $T$ and solves the
(search) $\mathsf{HPS}_{n,m,q,\chi}$ problem with probability $\epsilon$ in the average case, where $\chi$ is the uniform distribution over $\mathcal R(m,n)$ and $m \leq n$. 
Then, there exists an algorithm which runs in time $T \cdot  \poly(n)$ and  inverts Hamiltonian phase states $\ket{\Phi_{\boldsymbol{\theta}}^{\vec C}}^{\otimes t}$ with probability $\epsilon$ for a worst-case choice of architecture $\vec C \in \mathcal R(m,n)$, uniformly random angles $\boldsymbol{\theta}$,  and for any number of copies $t=\poly(n)$.
Here, $\mathcal R(m,n) = \{ \vec A \in \mathbb Z_2^{m \times n} | \mathrm{rank} (\vec A) = \min(m,n)\}$ is the set of full-rank binary $m \times n$ matrices.
\end{lemma}

\begin{proof}
Consider the reduction $\mathcal{B}$ which, on input $\ket{\Phi_{\boldsymbol{\theta}}^{\vec C}}^{\otimes t}$, does the following:
\begin{enumerate}
    \item $\mathcal{B}$ samples a uniformly random invertible matrix $\vec R \sim \mathrm{GL}(n,\Z_2)$.

\item $\mathcal{B}$ runs the average-case solver $\mathcal{A}$ on the input
$$
(U_{\vec R}\ket{\Phi_{\boldsymbol{\theta} }^{\vec C}})^{\otimes t}.
$$
where $U_{\vec R}$ is the $n$-qubit unitary transformation given by $U_{\vec R}: \ket{x} \mapsto \ket{\vec R^{-1} \cdot x}$, for $x \in \bit^n$. Finally, $\mathcal{B}$ outputs whatever $\mathcal{A}$ outputs.
\end{enumerate}

Note that $U_{\vec R}$ is a quantum circuit composed just of \textsf{CNOT} gates and therefore efficiently implementable. 
Because the average-case solver $\mathcal{A}$ runs in time $T$, it follows that the reduction $\mathcal{B}$ runs in time $T \cdot  \poly(n)$. 

Next, we show that $\mathcal{B}$ also succeeds with probability $\epsilon$.
By assumption, the worst-case instance $\ket{\Phi_{\boldsymbol{\theta}}^{\vec C}}^{\otimes t}$ consists of structured phase states states
$$
\ket{\Phi_{\boldsymbol{\theta}}^{\vec C}} = \exp\left(\mathrm{i} \, \sum_{i=1}^m \theta_i \bigotimes_{j=1}^n \mathsf{Z}^{\vec C_{ij}}  \right) H^{\otimes n} \ket{0^n},
$$
where $\vec C \in \mathcal R(m,n)$ and $\boldsymbol{\theta}$ is a tuple of random angles
$\boldsymbol{\theta} = (\theta_1, \dots, \theta_m) \in \Theta_q^m$.
To complete the proof, it suffices to show that $U_{\vec R}\ket{\psi_{ \boldsymbol \theta }^{\vec C}})^{\otimes t}$ is distributed exactly as in the $\mathsf{HPS}_{n,n,q,\chi}$ problem, where $\chi$ is the uniform distribution over $\mathcal R(m,n)$. First, we make the following key observation: it follows from unitarity of $U_{\vec R}$ that
$$
U_{\vec R}\ket{\Phi_{\boldsymbol \theta }^{\vec C}} = \left(U_{\vec R} \exp\left(\mathrm{i} \, \sum_{i=1}^m \theta_i \bigotimes_{j=1}^n \mathsf{Z}^{\vec C_{ij}}  \right) U_{\vec R}^\dag \right) U_{\vec R} H^{\otimes n} \ket{0^n}.
$$
Because $U_{\vec R}$ is an invertible matrix, it leaves the state $H^{\otimes n} \ket{0^n}$ invariant, and thus we have $U_{\vec R} H^{\otimes n} \ket{0^n} = H^{\otimes n} \ket{0^n}$.
Next, we study the action of $U_{\vec R}$ onto tensor products of Pauli operators. We find that for any index $i \in [n]$:
\begin{align*}
U_{\vec R} \left(\bigotimes_{j=1}^n \mathsf{Z}^{\vec C_{ij}} \right)U_{\vec R} ^\dag &= \sum_{x \in \bit^n} \bra{x} U_{\vec R} \left(\bigotimes_{j=1}^n \mathsf{Z}^{\vec C_{ij}} \right)U_{\vec R} ^\dag \ket{x} \cdot \proj{x}\\
&= \sum_{x \in \bit^n} \bra{\vec R x} \left(\bigotimes_{j=1}^n \mathsf{Z}^{\vec C_{ij}} \right) \ket{\vec R x} \cdot \proj{x}\\
&= \sum_{x \in \bit^n} (-1)^{\sum_{j=1}^n \vec C_{ij} (\vec R x)_{j}}  \proj{x}\\
&= \sum_{x \in \bit^n} (-1)^{\sum_{j=1}^n (\vec C \cdot \vec R)_{ij} x_{j}}  \proj{x} \\
&= \sum_{x \in \bit^n} \bra{x} \left(\bigotimes_{j=1}^n \mathsf{Z}^{(\vec C \cdot \vec R)_{ij}} \right) \ket{x}  \cdot \proj{x}\\
&= \bigotimes_{j=1}^n \mathsf{Z}^{(\vec C \cdot \vec R)_{ij}}.
\end{align*}
Because $U_{\vec R}$ is acting on a matrix exponential of a diagonal matrix, it follows that
\begin{align*}
U_{\vec R} \,\exp\left(\mathrm{i} \, \sum_{i=1}^m \theta_i\bigotimes_{j=1}^n \mathsf{Z}^{\vec C_{ij}}  \right) U_{\vec R}^\dag &= \exp\left(\mathrm{i} \, \sum_{i=1}^m \theta_i \, U_{\vec R} \left(\bigotimes_{j=1}^n \mathsf{Z}^{\vec C_{ij}} \right) U_{\vec R}^\dag \right)\\
&= \exp\left(\mathrm{i} \, \sum_{i=1}^m  \theta_i\bigotimes_{j=1}^n \mathsf{Z}^{(\vec C \cdot \vec R)_{ij}}  \right).
\end{align*}
Finally, we observe that for $m=n$, $\mathcal R(m,n) = \mathrm{GL}(n,\Z_2)$, which is a group. 
Because $\vec C \in \mathrm{GL}(n,\Z_2)$ it follows that $\vec C \cdot \vec R$ is uniformly distributed whenever $\vec R \sim \mathrm{GL}(n,\Z_2)$. Putting everything together, it follows that $U_{\vec R}\ket{\psi_{\boldsymbol \theta }^{\vec C}})^{\otimes t}$ is distributed precisely as in the $\mathsf{HPS}_{n,n,q,\chi}$ problem, and thus $\mathcal{B}$ succeeds with probability $\epsilon$.
The claim for $m \leq n$ follows from the fact that in that case $\vec C$ is a submatrix of a $\mathrm{GL}(n,\Z_2)$ matrix.
\end{proof}

\cref{lem:architecture reduction} shows a worst-to-average-case reduction in the regime $m \le n$ using only a \textsf{CNOT} circuit for randomization. 
We can also obtain a meaningful worst-case to average-case reduction when $m > n$ using additional Hamiltonian terms on a larger system, and a \textsf{CNOT} circuit on the larger system. However, we currently do not know how to analyze the resulting distribution and therefore state it as a conjecture:
\begin{conjecture}
\label{conj:uniformity}
    There is a regime of $m >n$ and choices of $k,l = \poly(n)$ such that the following is true.
    The uniform distribution with $C \sim \Z_2^{(m+ k) \times (n+ l)}$ is $1/\poly(n)$ close in total variation distance to the distribution over matrices of the form \begin{align}
    \label{eq:hidden ensemble}
    \vec M = \Pi \cdot \begin{pmatrix}
        \vec A& \vec 0\\ \vec B & \vec C
    \end{pmatrix}\cdot \vec R,
    \end{align} where $\vec A \in  \Z_2^{m \times n}$ is a fixed matrix, and all other matrices are uniformly random, i.e., $(\vec B|\vec C) \sim  \Z_2^{m \times (n+l)}$, $\vec R \sim \mathrm{GL}(n+l, \Z_2)$, $\Pi$ is a random row-permutation operator which we identify with an element in $\mathrm S_{m+k}$, where $\mathrm S_N$ denotes the symmetric group over $N$ elements.
\end{conjecture}

The Hamiltonian $\vec M$ characterizes a HPS which is obtained from a phase state $\ket{\Phi_{\boldsymbol \theta}^{\vec A}}$ via the following re-randomization:
\begin{align}
U_{\vec R}\exp\left (i\sum_{i=1}^k \varphi_i \sum_{j=1}^{n+l} Z_j^{(\vec B|\vec C)_{ij}}\right) \ket{\Phi_{\boldsymbol \theta}^{\vec A}}\otimes \ket +^{l} =  \ket{\Phi_{(\boldsymbol \theta, \boldsymbol \varphi)}^{\vec M}},
 \end{align}
using additional angles $\boldsymbol \varphi \in \Theta_q^{k}$.
Note that the permutation $\Pi$ of the rows of $\vec M$ does not affect the final state, but is a freedom we have in its description since all Hamiltonian terms commute.
Interestingly, these are precisely the freedoms exploited in the proposal for verified quantum advantage by Shepherd and Bremner~\cite{shepherd_temporally_2009}, and studied in more detail in Refs.~\cite{kahanamoku-meyer_forging_2023,bremner_iqp_2023,gross_secret_2023}.

\paragraph{Worst-case hardness and relation with the McEliece cryptosystem}
\cref{conj:uniformity} makes the connection to a particular classical cryptosystem, the McEliece cryptosystem, explicit. 
Using the relation to this cryptosystem, we argue for the worst-case hardness of the $\mathsf{HPS}$ problem. 

The McEliece cryptosystem is a code-based candidates for post-quantum cryptography~\cite{mceliece1978public}.  The scheme is based on Goppa codes:
\begin{definition}[Goppa code]
    Let $g(z) = \sum_{i} g_i z^{i} \in \mathbb{F}_{q^m}[z]$ be a polynomial and $L = \{\alpha_1, \ldots, \alpha_n\} \subseteq \mathbb{F}_{q^m}$ be values that do not evaluate to $0$, i.e. $g(\alpha_i) \neq 0$. 
    Then the \emph{Goppa code} with parameters $g(z)$ and $L$, $\Gamma(L, g(z))$ has the codewords
    \begin{equation*}
        \left\{c = (c_1, \ldots, c_n) \in \mathbb{F}_{q^m}^n : \sum_{i = 1}^{n} \frac{c_i}{z - \alpha_i} \equiv 0 \bmod g(z)\right\}\,.
    \end{equation*}
    Note that the sum in the condition describes a function of $z$, not an element of $\mathbb{F}_{q^m}$.
\end{definition}
The definition of the Goppa code will not be important, but it is important to note that instantiations of the McEliece cryptosystem with some other kinds of codes have been shown to be insecure, while the instantiation with Goppa codes is still believed to be secure against quantum adversaries.  
The cryptosystem is usually defined in terms of the following public-key encryption scheme.
\begin{enumerate}
    \item Alice samples a $[n, k, 2t+1]$-Goppa code, specified by its $k \times n$ sized generator matrix $G$, a random $k \times k$ non-singular matrix $R$, and a permutation $\Pi \in S_n$. 
    \item Alice releases $\hat{G} = R \cdot G \cdot \Pi$ and the rate $t$ as the public key, and keeps the tuple $(R, G, \Pi)$ as her private key.
    \item To encrypt a message $m$ of $k$ bits, Bob computes a $n$-bit error vector $e$ with Hamming weight $t$ and computes $c = m \cdot \hat{G} + e$.
\end{enumerate}

A related problem to the security of the primitive is the so called ``Goppa code distinguishing'' problem, which asks if the public key $\hat G = R \cdot G\cdot  \Pi$ can be distinguished from a uniformly random matrix.
In certain regimes, specifically the high rate regime, it is known that there is an efficient distinguisher~\cite{faugere2013distinguisher}, however for the parameter range used in practice the problem is believed to be hard.

Consider the following variant of the $\mathsf{HPS}$ problem
\begin{definition}[Goppa code phase learning]
\label{def:goppa learning}
    Let $G$ be a $[n, k, t]$-Goppa code, and consider the state $\ket{\Phi^{G^{\intercal}}_{\theta}}$.  
    The search variant of the Goppa code phase learning problem is the problem of learning the phases $\theta$ given a desciption of $G$ and samples of $\ket{\Phi^{G^{\intercal}}_\theta}$.

    The decision variant is to distinguish the states $\ket{\Phi^{G^{\intercal}}_{\theta}}$ from a Haar random state, given samples of the state and the architecture matrix.
\end{definition}

We can relate this problem to the search and decision variants of $\mathsf{HPS}$ using the security of the McEleice cryptosystem.

\begin{lemma}
\label{lem:goppa relation}
    If the Goppa code distinguishing problem is hard for quantum adversaries, then the Goppa code phase learning problem is as hard as the $\mathsf{HPS}$ problem.
\end{lemma}
\begin{proof}
    Given an adversary for the $\mathsf{HPS}$ problem, which works for a uniformly random architecture $\vec A$, we can sample a random permutation and non-singular matrix, and apply the architecture re-randomization from the previous section to get $\hat{G}^{\intercal} = \Pi G^{\intercal} R$. 
    We can further re-randomize the phase by applying another IQP circuit with the same architecture.
    By the hardness of the Goppa code distinguishing problem, this architecture is indistinguishable from a uniformly random one for the $\mathsf{HPS}$ solver, so the solver will output the correct phases with high probability.   
\end{proof}

Since \cref{lem:goppa relation} holds for all Goppa codes, the Goppa code phase learning problem (\cref{def:goppa learning}) is a worst-case problem over a restricted family of architectures that come from Goppa codes.
It shows that the $\mathsf{HPS}_{m,n,q,\chi}$ problem is at least as hard as the Goppa code distinguishing problem, thereby showing a worst-case reduction for it.
Note however, that the parameter regime of the worst-case instances may not be the same regime in which we can show a worst-to-average-case reduction.

\paragraph{Limitations of the architecture reduction.} 
Let us now discuss some limitations of the worst-to-average-case reductions we have given for the Hamiltonian architectures with fixed angles.
Specifically, we provide evidence that randomizing with \textsf{CNOT} circuits and additional Hamiltonian terms is not sufficient for a worst-to-average-case reduction whenever $m \gg n$, but that this might still be possible in the regime in which $m = n + O(\log n)$. 
We note however that while the \emph{information-theoretic} worst-to-average-case reduction might not work in this regime, we still believe that the \emph{computational} hardness persists, and therefore \cref{conj:uniformity} to be true.

To give a lower bound on the total-variation distance between the uniform distribution and the distribution with a fixed Hamiltonian architecture~$\vec A$, it is sufficient to describe a distinguishing test and give a lower bound on its success probability. 
Our test will distinguish two matrices $M_1, M_2$ drawn from the hidden ensemble $\mu$ with architecture matrices chosen as in \cref{eq:hidden ensemble} of \cref{conj:uniformity} (with $l=0$) versus the uniform ensemble $\mathcal U$ over $(m+k) \times n$ binary matrices, respectively. 
To this end, let us write $\vec A = \vec Z \vec E$ in row echelon form $\vec E$ via an invertible matrix $ \vec Z \in \mathrm{GL}(m,\Z_2)$.
We ask the question: what is the probability that there exist an invertible matrix $\vec Z \in \mathrm {GL}(m, \mathbb Z_2)$ and a matrix $\vec D\in \mb Z_2^{m \times n}$ in reduced row-echelon form with the property that
\begin{align}
    \exists \vec P_1, \vec P_2 \in S(m+k): (\vec Z \oplus \id_{k}) \vec P_1 \vec M_1 = (\vec Z \oplus \id_{k}) \vec P_2 \vec M_2 = \begin{pmatrix}
        \vec E\\
        \vdots
    \end{pmatrix}?
\end{align}
In particular, $\vec E$ has $m-n$ all-$0$ rows which distinguish it from a uniformly random matrix even up to row permutations, and multiplication by a $\mathrm{GL}(n,\mathbb Z_2)$ matrix from the right. 
For the $\mu$ ensemble, we know that this probability is equal to $1$ by construction. 
For the uniform ensemble, we can get a rough estimate on it as follows 
\begin{itemize}
    \item There are $\binom{k+m}{m}$ subsets of $m$ rows of $\vec M_1$, and $m!$ permutations of those rows. Hence, there are this many matrices $\vec Z_i$ which bring the permutations of those rows into the form $\vec E$.
    Let us think of all of the $\vec Z_i$ as being  different from each other.
    \item For each $\vec Z_i$, we can apply it to an arbitrary permutation of $m$ rows of $\vec M_2$ and check the probability of getting $\vec E$ (of which there are $m!\binom{k+m}{m}$).
    This probability is given roughly by $2^{-m \cdot n + \Theta(n^2)}$ if we model each choice of $m$ rows as iid.\ uniformly random, and the factor of $2^{\Theta(n^2)}$ comes from the fact that we only care about equality up to a $\mathrm{GL}(n,\mathbb Z_2)$ transformation on the right (of which there are roughly $2^{\Theta(n^2)}$).  
    Modelling the submatrices as uniform iid.\ matrices is of course strictly speaking false. 
    However, we expect it to be approximately true. To see why observe that if we choose a fixed subset of $m$ rows of $\vec M$, denoted as $\vec A = \vec Z\vec E$ and exchange only one of the rows of $\vec A$ for another row of the big matrix $\vec M$, yielding a matrix $\vec A^*$, then  $\vec Z^{-1}\vec A^*$ will be far from~$\vec E$ and look essentially random.
\end{itemize}

In total we roughly get a `collision probability' of 
\begin{align}
    \frac{(k+m)!^2}{k!^2} \cdot 2^{-m \cdot n + \Theta(n^2)}
    \sim 2^{ 2m \log k - n (m-\alpha n)} = 2^{2dc n \log n - (d-\alpha)n^2} \in \exp(-O(n^2))
\end{align}
where we chose $m = d n $ and $k = n^c$, and we fix $\Theta(n^2) = \alpha n^2$. 
This suggests that we can make the collision probability on the order of $1$ (i.e., a constant TVD away from the hidden ensemble $\mu$) if we choose $m-n = e \log n $ for some constant $e$, yielding 
\begin{align}
    2^{m \log k - n (m-n)} =  2^{(2c-e) n \log n + 2ec \log n}.
\end{align}
Hence, we can choose $c$ in $k = n^c$ large enough that $2c -e \geq 0$.
This suggests that the worst-to-average-case reduction using random additional Hamiltonian terms and randomization by a \textsf{CNOT} circuit fails for $k \in \omega(\log n)$, but might still work for $k \in O(\log n)$.
\\

The worst-to-average-case reduction over the architectures $\boldsymbol{A}$ does not require any information about the angles.
In~\Cref{sssec:solving hps} we will show that HPS can be broken if the architecture $\boldsymbol{A}$ is known precisely in the regime $m\leq n$ in which the above worst-to-average reduction works. 
It can, however, be desirable to make part of the key publicly available.
An example where making part of the key public is helpful for designing cryptographic primitives is the Learning with Errors (LWE) problem~\cite{regev2009lattices}.
The LWE problem is to recover an unknown vector $\vec s\in \mathbb{Z}_q^n$ from ``noisy'' scalar products with public but random vectors $\vec a\in \mathbb{Z}_q^n$.
Clearly, the LWE problem could be made harder by keeping the random vectors $\vec a$ secret, but keeping them public is crucial for many applications such as public-key encryption.
In the following, we study a variant of the $\mathsf{HPS}$ problem, where the architectures are randomly drawn but publicly available. 
We discover that even with public architecture, we can reduce the problem of learning the angles of a worst case to learning those of a typical instance. 

\paragraph{Reduction for the angles.}
To show a worst-to-average case reduction for task of learning the angles with a public architecture, we show that given an (worst-case) $\mathsf{HPS}_{n,n,q,\chi}$ instance $\ket{\Phi_{\boldsymbol \theta}^{\vec A}}$ with $\vec A \in \Z_q^{m \times n}$ in the clear, the angles can be completely randomized by a simple quantum algorithm.
The key idea is just to one-time-pad the angles by applying uniformly random rotations around the public Hamiltonian terms to the given HPS $\ket{\Phi_{\boldsymbol \theta}^{\vec A}}$.
\begin{lemma}[Worst-to-average-case reduction for the angles]
\label{lem:angle reduction}
    Suppose there is a quantum algorithm $\mathcal A$ which for every choice of $\vec A \in \mathbb Z_2^{m \times n}$, given a HPS $\ket{\Phi_{\boldsymbol \theta}^{\vec A}}$  with $\boldsymbol \theta$ chosen uniformly at random from $\Theta_q^m$, succeeds in returning $(\boldsymbol \theta, \vec A)$ with probability at least $1/\poly(m,n,\log q)$ over the choice of $\boldsymbol \theta$. 
    Then there is a quantum algorithm $\mathcal A'$ which succeeds at this task for \emph{every} choice of $\boldsymbol{\theta}$ with probability $1/\poly(m,n,\log q)$. 
\end{lemma}
\begin{proof}
    To show the lemma, we observe that given any fixed HPS $\ket{\Phi_{\boldsymbol \theta}^{\vec A}}$, we can randomize it by choosing $\boldsymbol \varphi \sim \Theta_q^m$  uniformly at random and apply $\mathcal A$ to the re-randomized  state  
    \begin{align}
        \exp\left( i \sum_{j =1}^m \varphi_j \prod_{k=1}^m  Z_k^{\vec A_{jk}}\right) \ket{\Phi_{\boldsymbol \theta}^{\vec A}} \equiv \ket{\Phi_{\boldsymbol \theta + \boldsymbol \varphi}^{\vec A}}.
    \end{align}
    Repeating this procedure at most $\poly(n,m,\log q)$ times will recover $\boldsymbol \theta + \boldsymbol \varphi$ and $\vec A$, and thus~$(\boldsymbol \theta, \vec A)$.
    This implies that an average-case algorithm can solve worst-case instances and thus proves the reduction.
\end{proof}

\paragraph{Learning with public architecture}
The worst-to-average-case reduction works for any parameters of $m,n$, but of course, for the hardness of $\HPS$ it is crucial that the learning problem is also worst-case hard. 
In particular, one might wonder whether the $\HPS$ assumption can be relaxed to a publicly known architecture, and unknown angles. 
We delineate the parameter regime in which we can expect this in the following.
To this end, we show that for $m \le n$ the angles can in fact be efficiently learned given a public, full-rank Hamiltonian architecture. 
Nonetheless, we believe such a stronger assumption to still be true whenever $m \gg n$. 

\begin{lemma}[Learning with public Hamiltonian architecture]
\label{lem:learning public architecture}

Let $n \in \mathbb{N}$ be the security parameter, and $m = \poly(n)$ and $q = \poly(n)$.
Consider the family of $\mathsf{HPS}_{n,m,q,\chi}$ instances $\mathcal E_{\vec A} = \{\ket{\Phi_{\boldsymbol \theta}^{\vec A}}\}_{\boldsymbol \theta \in \Theta_q^m}$ with publicly known $\vec A \in \mathcal R(m,n)$. 
Then, there exists an efficient quantum algorithm which inverts any state in the ensemble $\mathcal E_{\vec A}$ whenever $m \le n$. 
\end{lemma}
\begin{proof}
    Let $\vec R \in \mathrm{GL}(n,\Z_2)$ be an invertible matrix such that if we write $\vec A = \sum_{i=1}^m \ket{e_i}\bra{a_i}$, i.e., the rows of $\vec A$ are denoted by $\ket{a_i} \in \Z_2^n$, and the standard vectors by $\ket{e_i} \in \Z_2^m$, then 
    $\bra{a_i} \vec R = \bra{e_i} $. 
    Then $\vec A \vec R = (\id_m | 0_{(n-m) \times (n-m)})$. 

    This means that we can apply a \textsf{CNOT} circuit $U_{\vec R}$ to $\ket{\Phi_{\boldsymbol \theta}^{\vec A}}$, yielding 
    \begin{align}
        U_{\vec R}\ket{\Phi_{\boldsymbol \theta}^{\vec A}} &= \exp\left(i \sum_{j=1}^m \theta_j\prod_{k=1}^n Z_k^{(\vec A \vec R)_{jk}}\right)\ket +^{\otimes n} \\
        & =  \exp\left(i \sum_{j=1}^m \theta_j\prod_{k=1}^n Z_j\right)\ket +^{\otimes n} = \bigotimes_{j=1}^m\left( e^{i \theta_j Z_j} \ket + \right)\otimes \ket +^{\otimes (n-m)} ,
    \end{align}
    which is a product state. 
    The angles $\theta_j$ can now be learned to precision $1/\epsilon$ using $O(1/\epsilon^2)$ many copies of $ U_{\vec R}\ket{\Phi_{\boldsymbol \theta}^{\vec A}}$ via a measurement in the $X$-basis.
\end{proof}

The consequence of \cref{lem:learning public architecture} is that at least in the parameter regime in which we can show a worst-to-average case reduction over the architecture, it is necessary for the $\HPS$ assumption to be true that the Hamiltonian architecture is hidden. 
We note that---like for the worst-to-average-case reduction over the architectures---we expect that the learning algorithm can be extended to $m = n + O(\log n)$.

%!TEX root=./main.tex
\subsection{Hamiltonian Phase States Form Approximate State Designs}\label{section:designbounds}

In this subsection we show that the states in the HPS ensemble form approximate state designs if $m\geq Cn$ for a constant $C>0$.
It will be convenient to view HPS as a random walk of depth $m$ on the diagonal group. 
We will therefore slightly adjust the notation.
Consider the following probability distribution $\nu$ on the diagonal subgroup of $SU(2^n)$: Draw a uniformly random bitstring $\boldsymbol{A}_{1}\in\{0,1\}^n$ and a uniformly random angle $\theta\in [0,2\pi)$ and apply $e^{\mathrm{i}\theta \bigotimes_{j=1}^n Z^{\boldsymbol{A}_{1j}}}$. 
We can draw $m$ such diagonal unitaries independently and multiply them. The resulting probability measure is denoted by $\nu^{*m}$.

We will first show that $e^{\mathrm{i}\sum_{i=1}^m\theta_i\bigotimes_{j=1}^n Z^{\boldsymbol{A}_{ij}}}$ is an approximate $t$-design on the diagonal group.
% Recall that the diamond norm of a superoperator $\Phi$ is defined via:
% \begin{equation}
% ||\Phi||_{\diamond}=\max_{X,||X||_1}||\Phi\otimes \mathbb{1}(X)||_{1},
% \end{equation}
% where $||\bullet||_1$ denotes the trace norm. 
% In particular, $||\bullet||_{1\to 1}\leq ||\bullet||_{\diamond}$, where $||\bullet||_{1\to1}$ denotes the induces trace norm. 
% We can define approximate unitary designs:
% \begin{definition}
%     A probability measure $\nu$ on $SU(d)$ is an (additive) $\varepsilon$-approximate unitary $t$-design if
%     \begin{equation}
%          \left|\left|\Phi^{(t)}_{\nu}-\Phi^{(t)}_{\mu_H}\right|\right|_{\diamond}\leq \varepsilon
%     \end{equation}
% \end{definition}
% Similarly, we define approximate designs on the diagonal subgroup:
% \begin{definition}
%     A probability measure $\nu$ on the diagonal subgroup $D\subset SU(2^n)$ is a $\varepsilon$-approximate diagonal $t$-design if 
%     \begin{equation}
%         \left|\left|\Phi^{(t)}_{\nu}-\Phi^{(t)}_{\mu_D}\right|\right|_{\diamond}\leq \varepsilon, 
%     \end{equation}
%     where $\Phi^{(t)}_{\nu}$ denotes the $t$-fold twirl:
%     \begin{equation}
%         \Phi^{(t)}_{\nu}(A)\coloneqq \mathbb{E}_{U\sim \nu} U^{\otimes t} A (U^{\dagger})^{\otimes t}.
%     \end{equation}
% \end{definition}
More precisely, we prove the following theorem:
\begin{theorem}\label{thm:diagonal designs}
    For $m\geq 2t(2nt+\log(1/\varepsilon))$ the random unitary $e^{\mathrm{i}\theta_i \sum_{i=1}^m\bigotimes_j Z^{\boldsymbol{A}_{ij}}}$ with random $A_{ij}$ and $\theta_i$ is a $\varepsilon$-approximate diagonal $t$-design in the sense of~\Cref{def: diagonal designs}.
    Moreover, the same bound holds if $\theta_i$ is drawn uniformly from $\{2\pi k/q\}_{k=1}^q$, where $q$ is an integer satisfying $q>2t$.
\end{theorem}
We provide a proof of Theorem~\ref{thm:diagonal designs} in \Cref{app:design proofs}.
The proof of Theorem~\ref{thm:diagonal designs} is remarkably simple in comparison to the derivations of similar results for random quantum circuits~\cite{chen2024incompressibility,brandao2016local,haferkamp2022random}.
Additionally, the constants in Theorem~\ref{thm:diagonal designs} are unusually small: In stark contrast the constants in these results are north of $10^{13}$.
A similar result was obtained in Ref.~\cite{haah2024efficient} for the related random Pauli rotations $e^{\mathrm{i}\theta P}$ for a random $\theta\in(0,2\pi]$ and a random Pauli string $P$.

Theorem~\ref{thm:diagonal designs} almost directly implies the following corollary:
\begin{corollary}\label{cor:designs}
    For $m\geq 2t(2nt+\log(1/\varepsilon))$ the state ensemble defined by $|\Phi_{\boldsymbol{\theta}}^{\boldsymbol{A}}\rangle=U|+^n\rangle$ for U drawn from $\nu^{*m}$ (or $\nu_q^{*m}$ for $q>2t$) is a $\varepsilon+O(t^2/2^n)$-approximate state t-design.
\end{corollary}
As a consequence no algorithm with access to $t$ copies can distinguish the states $|\phi_{\boldsymbol{\theta}}^{\boldsymbol{A}}\rangle$ from Haar random. 
In particular, this rules out a large class of natural attacks which make use of a small number of samples. 
Prominent examples in classical cryptanalysis are linear attacks ($2$-wise independence rules this out), and differential attacks ($t$-wise independence rules out $\log_2(t)$ differential attacks).  
Moreover, the fact that HPS with sufficiently many terms can generate arbitrary state $t$ designs makes it seem unlikely even that there is a distinguishing algorithm using just a few more than $t$ samples. This would mean that there is a sharp transition in the complexity of distinguishing HPS states from uniform. 
Thus, the $t$-design property gives evidence for the security of the $\textsf{HPS}$ assumption.

\begin{proof}[Proof of Corollary~\ref{cor:designs}]
    $U|+^n\rangle=\sum_{x\in \{0,1\}^n} e^{\mathrm{i}\phi_x}|x\rangle$ for a uniformly random diagonal $U$ form a $O(t^2/2^n)$-approximate state design~\cite{brakerski2019pseudo}.
    The statement then follows directly from the triangle inequality:
    \begin{align*}
       & \left|\left|\underset{U\sim\nu^{*m}}{\mathbb{E}} (U|+^n\rangle\langle+^n|U^{\dagger})^{\otimes t}-\underset{U\sim\mu(d)}{\mathbb{E}} (U|+^n\rangle\langle+^n|U^{\dagger})^{\otimes t}\right|\right|_{1}\\
        &\leq \left|\left|\underset{U\sim\nu^{*m}}{\mathbb{E}} (U|+^n\rangle\langle+^n|U^{\dagger})^{\otimes t}-\underset{U\sim\mu_D(d)}{\mathbb{E}} (U|+^n\rangle\langle+^n|U^{\dagger})^{\otimes t}\right|\right|_{1}\\
        &+\left|\left|\underset{U\sim\mu(d)}{\mathbb{E}} (U|+^n\rangle\langle+^n|U^{\dagger})^{\otimes t}-\underset{U\sim\mu_D(d)}{\mathbb{E}} (U|+^n\rangle\langle+^n|U^{\dagger})^{\otimes t}\right|\right|_{1}\\
        &\leq \varepsilon+O(t^2/2^{n}).
    \end{align*}
\end{proof}

As a consequence we can also show that HPS contains many almost orthogonal states, yielding additional evidence for the $\textsf{HPS}$ assumption:
\begin{corollary}\label{lemma: many orthogonal states}
    Let $m=100nt$, $\delta = 1-2^{-n/8}$ and $t\leq 2^{n/2}$. 
    For any fixed state $|\psi\rangle$, we have with probability $1-2^{-\Omega(nt)}$ over the matrix $\boldsymbol{A}$ that 
    \begin{equation}
\mathrm{Pr}_{\boldsymbol{\theta}}\left[|\langle \psi|\exp\left(\sum_{i=1}^m\mathrm{i}\theta_i \bigotimes_{j=1}^{n}Z^{\boldsymbol{A}_{ij}}\right)|+^n\rangle|^2\geq 1-\delta\right]\leq 2^{-\Omega(nt)}.
    \end{equation}
\end{corollary}

We defer the proofs of  
Theorem~\ref{thm:diagonal designs} and Corollary~\ref{lemma: many orthogonal states} to \cref{app:design proofs}.

\subsection{Algorithms for Learning Hamiltonian Phase States} 
\label{sssec:solving hps}

Recall that our (search) $\HPS$ assumption can be thought of as a state discrimination task. The goal is to recover the architecture $\vec A \in \Z_2^{m \times n}$ and the set of angles $\boldsymbol{\theta} \in \Theta_q^m$ given many copies of a random Hamiltonian phase state from the ensemble
$$
\mathcal{E} =\left\{\ket{\Phi_{\boldsymbol{\theta}}^{\vec A}} = \exp\left(\mathrm{i} \, \sum_{i=1}^m \theta_i \bigotimes_{j=1}^n \mathsf{Z}^{\vec A_{ij}}  \right)  \ket{+^n} \right\}_{\vec A \in \Z_2^{m \times n}, \,\boldsymbol{\theta} = (\theta_1, \dots, \theta_m) \in \Theta_q^m}.
$$
In this section, we consider various learning algorithms for the (search) HPS probem. We observe that the HPS problem does in fact have polynomial quantum sample complexity, and can thus be solved information-theoretically. However, as we also observe, all known learning algorithms have exponential time complexity, which suggests that the HSP problem cannot be solved efficiently. 

We distinguish between the \emph{private-key} and \emph{public-key} setting: the former is essentially the learning task from \Cref{def:search-HPS}, whereas in the latter we further assume that the learner also has access to the architecture matrix $\vec A \in \Z_2^{m \times n}$. 
We provide evidence that the learning tasks remains hard even if we reveal additional information about $\vec A \in \Z_2^{m \times n}$ and the goal is simply to guess the angles $\boldsymbol{\theta}$.

\paragraph{Sample complexity of HPS and hypothesis selection}
While we believe that HPS is a computationally hard problem, it can be solved information-theoretically with only polynomially many samples. 
In full generality, the problem of finding a fixed state $\rho_j$ among many hypothesis states $\rho_1,\ldots,\rho_M$ is called quantum hypothesis testing. 
Currently, the best known general algorithm is threshold search as described in~\cite[Theorem~1.5]{buadescu2021improved} requires $n\log^2(M)$ copies improving over the bound from Ref.~\cite{aaronson2018shadow}.
For the HPS problem this implies an upper bound on the sample complexity of $O(n\log^2(q^m2^{nm}))=O(n^3m^2\log(q))$.
As the fidelities for pure states are PSD observables of rank $1$, we can also use the shadow tomography protocol of Ref.~\cite{huang2020predicting}.
Given a secret state $|\Phi^{\boldsymbol{A}}_{\boldsymbol{\theta}}\rangle$ allows us to estimate the fidelities of all the $M=q^{m}2^{nm}$ phase states up to an error of $\varepsilon$ from $O(\log(M)/\varepsilon^2)=O(mn\log(q)/\varepsilon^2)$ samples.
Then, a solver can simply list all estimated fidelities and pick the state with the largest overlap up to an error of $\varepsilon$.

We expect these bounds to be tight in the regime where $m\leq O(n^{\log(q)})$.
For $m\to \infty$ better bounds are available at least for $q=2^d$.
In this case, the HPS instance generated b unitaries in the $d$th level of the Clifford hierarchy and it was proven in Ref.~\cite[Theorem~15]{arunachalam2023optimal} that for any state of the form 
\begin{equation}
\exp\left(\mathrm{i}\sum_{y\in\{0,1\}^n} a_y \bigotimes_{j=1}^nZ^{y_j}\right)|+^n\rangle
\end{equation}
with $a_y\in \mathbb{Z}$ a circuit description can be learned with $O(n^d)$ copies using only measurements in the standard basis.

\paragraph{Learning algorithms for HPS with a public architecture.}

In the special case when the architecture is public, our $\HPS$ assumption does in fact admit an optimal\footnote{Here, we mean an algorithm that achieves the optimal success probability for a given number of copies.} but nevertheless exponential-time learning algorithm.

We consider the following state discrimination task, where the goal is to recover the set of angles $\boldsymbol{\theta}$ given many copies from the ensemble
$$
\mathcal{E}_{\vec A} =\left\{\ket{\Phi_{\boldsymbol{\theta}}^{\vec A}} = \exp\left(\mathrm{i} \, \sum_{i=1}^m \theta_i \bigotimes_{j=1}^n \mathsf{Z}^{\vec A_{ij}}  \right)  \ket{+^n} \right\}_{\boldsymbol{\theta} = (\theta_1, \dots, \theta_m) \in \Theta_q^m}
$$
where the matrix $\vec A \in \Z_2^{m \times n}$ is a random but fixed \emph{architecture} which is known to the learner.
This fits exactly into the framework of the pretty good measurement (PGM)~\cite{barnum2000reversingquantumdynamicsnearoptimal,montanaro2019pretty}.
The ensemble $\mathcal{E}$ now turns out to be \emph{geometrically uniform} because it can be written as 
$\mathcal{E}_{\vec A} =\left\{  U_{\boldsymbol{\theta}}^{\vec A} \ket{+^n} \right\}_{\boldsymbol{\theta} = (\theta_1, \dots, \theta_m)}$
where $\{U_{\boldsymbol{\theta}}^{\vec A}\}_{\boldsymbol{\theta}}$ is an Abelian group of matrices. Eldar and Forney~\cite{eldar2000quantumdetectionsquarerootmeasurement} showed that the PGM is optimal for all geometrically uniform ensembles, which implies that it is also optimal for our variant of the HPS problem. Nevertheless, despite the optimality, the best known algorithm for implementing pretty good measurements has exponential-time complexity in the size of the ensemble~{\cite{Gily_n_2022}. Consequently, we believe that the HPS problem remains computationally intractable, even if the architecture is public.

\subsection{Discussion of Full Quantumness}
\label{ssec:fully quantumness of hps}

We now provide some evidence that our HPS assumption is \emph{plausibly} fully quantum.
To begin, note that it is quite challenging (in general) to argue that any concrete assumption (one that names a specific problem as not being solvable in polynomial time) does not imply one-way functions.  
Indeed, if one-way functions exist by other means, then every concrete assumption implies a one-way function independent of whether it is true.  

One way that researchers get around this is to take some ``idealized'' version of the assumption, and show that relative to some oracle, the idealized version exists while one-way functions do not.  
The oracle separations shown in \cite{kretschmer2021quantum,kretschmer2023quantum} already prove this for us.  Namely, if we assume that our circuits actually output Haar random states, then with an additional $\mathsf{PSPACE}$ oracle, $\mathsf{BQP} = \mathsf{QMA}$, but our idealized Haar random states are genuinely Haar random.  
However, using this as evidence that our assumption is fully-quantum is barely better than simply stating it to be so, since our argument essentially reduced to ``nothing that claims to output Haar random-looking states can imply one-way functions''.
However, the fact that our ensembles are actually $t$-designs for a very high $t$ should make one-way functions very hard to build from these states in practice.  
Indeed, examining the proof of \cite{kretschmer2021quantum}, we can see that if we instead take a family of $t$-designs instead of Haar random unitaries, any $\mathsf{BQP}$ algorithm that calls the oracle fewer than $t$ times will have a similar concentration as the Haar measure, and thus will be simulatable by the $\mathsf{PSPACE}$ oracle (without access to the $t$-design).
This implies that any construction of one-way functions that directly uses our assumption must make many (more than $t = \sqrt{m/n}$) calls to our oracle, implying that any such construction cannot be simple.

\section{Applications}\label{sec:applications}

\subsection{One-Way State Generators}

The search variant of $\mathsf{HPS}$ allows for a practical and efficient implementation of one-way state generators (\cref{def:OWSG}). The following observation is immediate.

\begin{theorem}[One-Way State Generators from HPS]
Let $n \in \N$ be the security parameter, and let $m=\poly(n)$ and $q=\poly(n)$. Then, under the (search) $\mathsf{HPS}_{n,m,q,\chi}$ assumption, where $\chi$ is the uniform distribution, there exist one-way state generators.
\end{theorem}

\subsection{Pseudorandom State Generators}

In a similar vein, the decision variant of $\mathsf{HPS}$ gives rise to an efficient implementaion of pseudo-random states (\cref{def:prsg}). The following observation is immediate.

\begin{theorem}[Pseudorandom State Generators from HPS]
Let $n \in \N$ be the security parameter, and let $m=\poly(n)$ and $q=\poly(n)$. Then, under the (decision) $\mathsf{HPS}_{n,m,q,\chi}$ assumption, where $\chi$ is the uniform distribution, there exist pseudorandom state generators.
\end{theorem}

\subsection{Quantum Trapdoor Functions}\label{sec:QTF}
As a stepping stone towards building public-key cryptography from the HPS assumption, we show how to construct quantum trapdoor functions~\cite{coladangelo2023quantumtrapdoorfunctionsclassical}. These are one-way state generators that have a trapdoor that can be used to invert them easily. 
\begin{definition}[Quantum Trapdoor Function] Let $n \in \N$ denote the security parameter.\\ A quantum trapdoor function is a tuple of QPT algorithms $(\mathsf{GenTrap},\mathsf{GenEval}, \mathsf{Eval}, \mathsf{Invert})$:
\begin{itemize}
    \item $\mathsf{GenTrap}(1^n) \rightarrow \mathsf{td}$: on input $1^n$, it outputs a classical trapdoor $\mathsf{td}$.

    \item $\mathsf{GenEval}(\mathsf{td}) \rightarrow \ket{\xi_{\mathsf{td}}}$: on input $\mathsf{td}$, it outputs a quantum evaluation state $\ket{\xi_{\mathsf{td}}}$.

    \item $\mathsf{Eval}(\ket{\xi_{\mathsf{td}}},x) \rightarrow \ket{\psi_x}$: on input $\ket{\xi_{\mathsf{td}}}$ and $x \in \bit^n$, it outputs a state $\ket{\psi_x}$.

    \item $\mathsf{Invert}(\mathsf{td}, \ket{\psi_x}) \rightarrow x'$: on input $\mathsf{td}$ and $\ket{\psi_x}$, it outputs a string $x' \in \bit^n$.
\end{itemize}
We require that the algorithms satisfy the following properties:
\begin{itemize}
    \item[(a)] \textbf{Hardness of inversion:} For any $n \in \N$, for any QPT algorithm $\mathcal{A}$ and for any number of copies $m=\poly(n)$, it holds that
    $$
    \Pr\left[x \leftarrow \mathcal{A}(\ket{\psi_x},\ket{\xi_{\mathsf{td}}}^{\otimes m}) \, : \, \substack{
    x \sim \bit^n\\
\mathsf{td} \leftarrow \mathsf{GenTrap}(1^n)\\
\ket{\xi_{\mathsf{td}}} \leftarrow \mathsf{GenEval}(\mathsf{td})\\
\ket{\psi_x} \leftarrow \mathsf{Eval}(\ket{\xi_{\mathsf{td}}},x)
    }\right] \leq \negl(n).
    $$

\item[(b)] \textbf{Trapdoor inversion:} For any $n \in \N$ and for any $x \in \bit^n$, it holds that
$$
    \Pr\left[x \leftarrow \mathsf{Invert}(\mathsf{td}, \ket{\psi_x}) \, : \, \substack{
\mathsf{td} \leftarrow \mathsf{GenTrap}(1^n)\\
\ket{\xi_{\mathsf{td}}} \leftarrow \mathsf{GenEval}(\mathsf{td})\\
\ket{\psi_x} \leftarrow \mathsf{Eval}(\ket{\xi_{\mathsf{td}}},x)
    }\right] =1.
$$
\end{itemize}

\end{definition}

\paragraph{Quantum trapdoor functions from HPS.}

Let us now construct QTFs from HPS using the following construction.

\begin{construction}\label{cons:qtf}
Let $n \in \N$ be the security parameter and let $m=\poly(n)$ and $q=\poly(n)$. Then, we define the QTF construction $(\mathsf{GenTrap},\mathsf{GenEval}, \mathsf{Eval}, \mathsf{Invert})$ as follows:
\begin{itemize}
    \item $\mathsf{GenTrap}(1^n) \rightarrow \mathsf{td}$: on input $1^n$, it outputs a trapdoor $\mathsf{td} = (\boldsymbol{\theta},\vec A)$, where
    $$
    \boldsymbol{\theta} \sim \Theta_q^m \quad\quad\text{and}\quad\quad \vec A \sim  \Z_2^{m \times n}.
    $$

    \item $\mathsf{GenEval}(\mathsf{td}) \rightarrow \ket{\xi_{\mathsf{td}}}$: on input $\mathsf{td}$, it outputs the quantum evaluation state $$
\ket{\xi_{\mathsf{td}}} = \exp\left(\mathrm{i} \, \sum_{i=1}^m \theta_i \bigotimes_{j=1}^n \mathsf{Z}^{\vec A_{ij}}  \right)  \ket{+^n}.
    $$

    \item $\mathsf{Eval}(\ket{\xi_{\mathsf{td}}},x) \rightarrow \ket{\psi_x}$: on input $\ket{\xi_{\mathsf{td}}}$ and $x \in \bit^n$, it outputs the state
$$
\ket{\psi_x} = (\mathsf{Z}^{x_1} \otimes  \dots \otimes \mathsf{Z}^{x_n})\ket{\xi_{\mathsf{td}}}.
$$

    \item $\mathsf{Invert}(\mathsf{td}, \ket{\psi_x}) \rightarrow x'$: on input $\mathsf{td}$ and $\ket{\psi_x}$, it 
    applies $H^{\otimes n} U_{\boldsymbol{\theta}}^{\vec A, \dag}$ to $\ket{\psi_x}$, where
    $$
    U_{\boldsymbol{\theta}}^{\vec A} =\exp\left(\mathrm{i} \, \sum_{i=1}^m \theta_i \bigotimes_{j=1}^n \mathsf{Z}^{\vec A_{ij}}  \right). 
    $$
   Then, it measures all qubits in the computational basis and outputs the result $x'$.
\end{itemize}
\end{construction}

\begin{theorem}
Let $n \in \N$ be the security parameter, and let $m=\poly(n)$ and $q=\poly(n)$. Then, under the (decisional) $\mathsf{HPS}_{n,m,q,\chi}$ assumption, where $\chi$ is the uniform distribution, \Cref{cons:qtf} yields a quantum trapdoor function. 
\end{theorem}
\begin{proof} We need to verify two properties:
\begin{itemize}
\item[(a)] \emph{Hardness of inversion:} First, we invoke the decisional $\mathsf{HPS}_{n,m,q,\chi}$ assumption,
\begin{align*}
 &\Pr\left[x \leftarrow \mathcal{A}(\ket{\psi_x},\ket{\xi_{\mathsf{td}}}^{\otimes m}) \, : \, \substack{
    x \sim \bit^n\\
\mathsf{td} \leftarrow \mathsf{GenTrap}(1^n)\\
\ket{\xi_{\mathsf{td}}} \leftarrow \mathsf{GenEval}(\mathsf{td})\\
\ket{\psi_x} \leftarrow \mathsf{Eval}(\ket{\xi_{\mathsf{td}}},x)
    }\right]\\
    &\leq \Pr\left[x \leftarrow \mathcal{A}(\mathsf{Z}^x\ket{\psi},\ket{\psi}^{\otimes m}) \, : \, \substack{
    x \sim \bit^n\\
\ket{\psi} \leftarrow \mathrm{Haar}(2^n)
    }\right] + \negl(n).
\end{align*}
The claim then follows from \cite[Lemma 5]{coladangelo2023quantumtrapdoorfunctionsclassical}.

\item[(b)] \emph{Trapdoor inversion:} This follows from the fact that $\mathsf{Z}^x$ commutes with the diagonal part of the Hamiltonian Phase state. In particular, using that 
$$
\mathsf{Z}^x H^{\otimes n}  =  H^{\otimes n} \mathsf{X}^x, \quad \forall x \in \bit^n,
$$
we find that
\begin{align*}
\mathsf{Z}^{x}U_{\boldsymbol{\theta}}^{\vec A} \ket{+^n}   &=
U_{\boldsymbol{\theta}}^{\vec A}\mathsf{Z}^{x}\ket{+^n}\\
&=U_{\boldsymbol{\theta}}^{\vec A} \mathsf{Z}^{x}H^{\otimes n}\ket{0^n}\\
&=
U_{\boldsymbol{\theta}}^{\vec A} H^{\otimes n}\mathsf{X}^{x}\ket{0^n}\\
&=
U_{\boldsymbol{\theta}}^{\vec A} H^{\otimes n}\ket{x}. 
\end{align*}
\end{itemize}
Therefore, inversion takes place with probability $1$.
\end{proof}

\subsection{Public-Key Encryption with Quantum Keys}

Due to the work of \cite{coladangelo2023quantumtrapdoorfunctionsclassical}, quantum trapdoor functions imply a public-key encrpytion scheme, with quantum public keys.  

\begin{lemma}[Theorem 5 from \cite{coladangelo2023quantumtrapdoorfunctionsclassical}]
    If quantum trapdoor functions exist, then there is a public-key encryption scheme with quantum public keys.
\end{lemma}

Note, that the work of Coladangelo~\cite{coladangelo2023quantumtrapdoorfunctionsclassical} required the use of (post-quantum) one-way functions.
Using our construction of quantum trapdoor functions from HPS (which is a potentially even weaker assumption than the existence of one-way functions), we immediately get that public-key encryption scheme with quantum public keys exist.

\begin{corollary}
    Let $n \in \mathbb{N}$ be the security parameter, and $m = \poly(n)$ and $q = \poly(n)$.  Then under the (decisional) $\mathsf{HPS}_{n, m, q, \chi}$ assumption, where $\chi$ is the uniform distribution, there is a construction of public-key encryption with quantum public keys.
\end{corollary}
\begin{proof}
    The construction results from inserting \cref{cons:qtf} from HPS into the construction of public-key encryption with quantum public keys from \cite{coladangelo2023quantumtrapdoorfunctionsclassical}.
\end{proof}

\subsection{Quantum Pseudoentanglement}
\label{sec:pseudoentanglement}

In this section, we show that we use ensembles of Hamiltonian Phase states to get a highly efficient construction of quantum pseudoentanglement, a notion that was introduced in the work of Aaronson et al.~\cite{aaronson_et_al:LIPIcs.ITCS.2024.2}.

\begin{definition}[Quantum pseudoentanglement] 
\label{def:pseudo_entangled_states}
Let $n \in \N$ be the security parameter. A $(f(n),g(n))$-pseudoentangled state ensemble (PES) consists of a pair of $n$-qubit state ensembles $\{\ket{\psi_k}\}_{k \in \algo K_n}$ and $\{\ket{\phi}_k\}_{k \in \algo K_n}$ with key space $\algo K_n$ such that:
\begin{itemize}
    \item (Efficiency:) There exists a pair of uniform $\poly(n)$-sized quantum circuits $C_\psi, C_\phi$ which, on input $k \in \algo K_n$, prepare the states $\ket{\psi_k}$ and $\ket{\phi_k}$, respectively.

    \item (Entanglement gap:) With high probability (i.e., at least $1-1/\poly(n)$) over the choice of $k \sim \algo K_n$, the entanglement entropy between the first $\frac{n}{2}$ and last $\frac{n}{2}$ qubits of $\ket{\psi_k}$ (respectively, $\ket{\phi_k}$) is in the order of $\Theta(f(n))$ (respectively, $\Theta(g(n))$).

    \item (Computational indistinguishability:) For any integer $n \in \N$ and any number of copies $t=\poly(n)$, the following mixed states are computationally indistinguishable,
    $$
    \rho = \underset{k \sim \algo K_n}{\mathbb{E}}\left[\proj{\psi_k}^{\otimes t(n)} \right] \quad \approx_c \quad  \sigma = \underset{k \sim \algo K_n}{\mathbb{E}}\left[\proj{\phi_k}^{\otimes t(n)} \right].
    $$
    In other words, for any efficient quantum distinguisher $\algo A$ which outputs a single-bit,
    $$
            \left|\Pr_{k \sim \algo K_n(1^n)}[\mathcal{A}(\ket{\psi_k}^{\otimes t(n)})=1] - 
        \Pr_{k \sim \algo K_n}[\mathcal{A}(\ket{\phi_k}^{\otimes t(n)})=1] \right| \leq \negl(n).
    $$
\end{itemize}
    
\end{definition}

\paragraph{Quantum Pseudoentanglement from the HPS assumption.}

Let $n \in \N$ be the security parameter. We can construct a $(k,n)$-pseudoentangled state ensemble (PES) from the HPS assumption as follows:
\begin{itemize}
    \item (high-entanglement state) choose a parameter $m \geq n$, and let $\{\ket{\Phi_{\boldsymbol{\theta}}^{\vec A}}\}_{\vec A \in \Z_2^{m \times n}, \,\boldsymbol{\theta}\in \Theta_q^m}$ be the Hamiltonian phase state family with
    $$
\ket{\Phi_{\boldsymbol{\theta}}^{\vec A}} = \exp\left(\mathrm{i} \, \sum_{i=1}^m \theta_i \bigotimes_{j=1}^n \mathsf{Z}^{\vec A_{ij}}  \right) H^{\otimes n} \ket{0^n}.
$$

\item (low-entanglement state) choose a parameter $k \leq n$, and let $\{\ket{\Phi_{\boldsymbol{\omega}}^{\vec B}}\}_{\vec B \in \Z_2^{k \times n}, \,\boldsymbol{\omega}\in \Theta_q^k}$ be the Hamiltonian phase state family with
    $$
\ket{\Phi_{\boldsymbol{\omega}}^{\vec B}} = \exp\left(\mathrm{i} \, \sum_{i=1}^k \omega_i \bigotimes_{j=1}^n \mathsf{Z}^{\vec B_{ij}}  \right) H^{\otimes n} \ket{0^n}.
$$
\end{itemize}

The goal of this section is to prove the following theorem.

\begin{theorem}[Quantum pseudoentanglement from HPS]\label{thm:pseudoentanglement}

Let $n \in \mathbb{N}$. Let $k \leq n \leq m$ with $k=\omega(\log(n))$ and $q=\poly(n)$ with $q>4$. Assuming the hardness of the (decisional) HPS assumption $\HPS_{n,\mu,q,\chi}$, for $\mu \in\{k,m\}$ and where $\chi$ is the uniform distribution, the pair  $\{\ket{\Phi_{\boldsymbol{\theta}}^{\vec A}}\}_{\vec A \in \Z_2^{m \times n}, \,\boldsymbol{\theta}\in \Theta_q^m}$ and $\{\ket{\Phi_{\boldsymbol{\omega}}^{\vec B}}\}_{\vec B \in \Z_2^{k \times n}, \,\boldsymbol{\omega}\in \Theta_q^k}$ form a $(n,k)$-pseudoentangled state ensemble. 
\end{theorem}

Before we give a proof of the statement, let us first introduce some relevant notation. For brevity, we define the diagonal part of of a Hamiltonian Phase state $\ket{\Phi_{\boldsymbol{\theta}}^{\vec A}}$ as
$$
D_{\boldsymbol{\theta}}^{\vec A} =\exp\left(\mathrm{i} \, \sum_{i=1}^m \theta_i \bigotimes_{j=1}^n \mathsf{Z}^{\vec A_{ij}}  \right)
$$
Next, we define the $\frac{n}{2} \times \frac{n}{2}$ partition matrix $\boldsymbol{\Lambda}_{\boldsymbol{\theta}}^{\vec A}$ such that, for $k,l \in \bit^{\frac{n}{2}}$,
$$
\boldsymbol{\Lambda}_{\boldsymbol{\theta},(k,l)}^{\vec A} := \bra{k,l}D_{\boldsymbol{\theta}}^{\vec A} \ket{k,l}.
$$
To get a handle on the entanglement entropy of the state, we have to analyze its reduced state across a particular cut $(A:B)$, say the first $n/2$ and the last $n/2$ many qubits.
\begin{align*}
\rho_A &= 2^{-n} \left(\sum_{i \in \bit^{\frac{n}{2}}} \sum_{j \in \bit^{\frac{n}{2}}} \sum_{k \in \bit^{\frac{n}{2}}} \sum_{l \in \bit^{\frac{n}{2}}} \boldsymbol{\Lambda}_{\boldsymbol{\theta},(i,j)}^{\vec A} \overline{\boldsymbol{\Lambda}}_{\boldsymbol{\theta},(k,j)}^{\vec A}\mathrm{Tr}_B[\ket{i,j}\bra{k,l}] \right)\\
&= 2^{-n} \left(\sum_{i \in \bit^{\frac{n}{2}}} \sum_{j \in \bit^{\frac{n}{2}}} \sum_{k \in \bit^{\frac{n}{2}}} \sum_{l \in \bit^{\frac{n}{2}}} \boldsymbol{\Lambda}_{\boldsymbol{\theta},(i,j)}^{\vec A} \overline{\boldsymbol{\Lambda}}_{\boldsymbol{\theta},(k,j)}^{\vec A} \ket{i}\bra{k}\right) \\
&= 2^{-n} \left(\sum_{i \in \bit^{\frac{n}{2}}} \sum_{j \in \bit^{\frac{n}{2}}} \sum_{k \in \bit^{\frac{n}{2}}} \sum_{l \in \bit^{\frac{n}{2}}} \boldsymbol{\Lambda}_{\boldsymbol{\theta},(i,j)}^{\vec A}\overline{\boldsymbol{\Lambda}}_{\boldsymbol{\theta},(k,j)}^{\vec A} \ket{i} \braket{j|j}\bra{k}\right)\\
&= 2^{-n} \left(\sum_{i \in \bit^{\frac{n}{2}}} \sum_{j \in \bit^{\frac{n}{2}}} \boldsymbol{\Lambda}_{\boldsymbol{\theta},(i,j)}^{\vec A}\ket{i} \bra{j}\right) \left(\sum_{j \in \bit^{\frac{n}{2}}} \sum_{k \in \bit^{\frac{n}{2}}} \overline{\boldsymbol{\Lambda}}_{\boldsymbol{\theta},(j,k)}^{\vec A}\ket{j} \bra{k}\right) \\
&= 2^{-n} \boldsymbol{\Lambda}_{\boldsymbol{\theta}}^{\vec A} (\boldsymbol{\Lambda}_{\boldsymbol{\theta}}^{\vec A})^\dag.
\end{align*}
From the work of Aaronson et al.~\cite{aaronson_et_al:LIPIcs.ITCS.2024.2}, we know that the entanglement entropy $S(\rho)$ is bounded below and above by the following quantities
$$
-\log\left(\| 2^{-n}\boldsymbol{\Lambda}_{\boldsymbol{\theta}}^{\vec A} (\boldsymbol{\Lambda}_{\boldsymbol{\theta}}^{\vec A})^\dag \|_F\right) \leq S(\rho) \leq \log\left(\mathrm{rank}(\boldsymbol{\Lambda}_{\boldsymbol{\theta}}^{\vec A} (\boldsymbol{\Lambda}_{\boldsymbol{\theta}}^{\vec A})^\dag)\right)
$$
% \alex{Question: can we tune the entanglement entropy via certain architectures or angles?}

We now give a proof of the aforementioned theorem.

\begin{proof}[Proof of Theorem~\ref{thm:pseudoentanglement}]

The efficiency of the state ensembles is immediate. For pseudorandomness, it suffices to invoke the (decisional) HPS assumption $\HPS_{n,\mu,q,\chi}$, for $\mu \in\{k,m\}$ and the uniform distribution $\chi$ the uniform distribution, to argue that the Hamiltonian Phase state pair  $\{\ket{\Phi_{\boldsymbol{\theta}}^{\vec A}}\}_{\vec A \in \Z_2^{m \times n}, \,\boldsymbol{\theta}\in \Theta_q^m}$ and $\{\ket{\Phi_{\boldsymbol{\omega}}^{\vec B}}\}_{\vec B \in \Z_2^{k \times n}, \,\boldsymbol{\omega}\in \Theta_q^k}$ is computationally indistinguishable. This is essentially a consequence of the triangle inequality, since we each of the states is itself indistinguishable from a Haar random state.

Next, we attempt to get a handle on the average entanglement entropy for a general Hamiltonian Phase state. For convenience, we analyze the case for a general $m \in \N$, which we can later choose to be either $m\geq n$ or $m \leq n$ with $m = \omega(\log n)$. To get a lower bound on the entanglement entropy, we need to bound the expectation
\begin{align}
\begin{split}\label{eq:frobeniusnorm}
\underset{\substack{\vec A \sim \Z_2^{m \times n}\\
\boldsymbol{\theta} \sim \Theta_q^m}}{\mathbb{E}} & \left[\| 2^{-n}\boldsymbol{\Lambda}_{\boldsymbol{\theta}}^{\vec A} (\boldsymbol{\Lambda}_{\boldsymbol{\theta}}^{\vec A})^\dag \|^2_F\right] \\
&=2^{-2n} \sum_{i\in \bit^{\frac{n}{2}}}  \sum_{j\in \bit^{\frac{n}{2}}} \underset{\substack{\vec A \sim \Z_2^{m \times n}\\
\boldsymbol{\theta} \sim \Theta_q^m}}{\mathbb{E}}\left[ \left(\sum_{k\in \bit^{\frac{n}{2}}} \boldsymbol{\Lambda}_{\boldsymbol{\theta},(i,k)}^{\vec A} \overline{\boldsymbol{\Lambda}}_{\boldsymbol{\theta},(j,k)}^{\vec A} \right)^2\right] \\
&=2^{-2n} \sum_{i\in \bit^{\frac{n}{2}}}  \sum_{j\in \bit^{\frac{n}{2}}}\sum_{k,l\in\{0,1\}^{\frac{n}{2}}}\underset{\substack{\vec A \sim \Z_2^{m \times n}\\
\boldsymbol{\theta} \sim \Theta_q^m}}{\mathbb{E}}\boldsymbol{\Lambda}_{\boldsymbol{\theta},(i,k)}^{\vec A} \overline{\boldsymbol{\Lambda}}_{\boldsymbol{\theta},(j,k)}^{\vec A} \boldsymbol{\Lambda}_{\boldsymbol{\theta},(i,l)}^{\vec A} \overline{\boldsymbol{\Lambda}}_{\boldsymbol{\theta},(j,l)}^{\vec A}. 
\end{split}
\end{align}

Consider the expression 
\begin{align}
\begin{split}\label{eq:secondmoment}
    \underset{\substack{\vec A \sim \Z_2^{m \times n}\\
\boldsymbol{\theta} \sim \Theta_q^m}}{\mathbb{E}}&\boldsymbol{\Lambda}_{\boldsymbol{\theta},(i,k)}^{\vec A} \overline{\boldsymbol{\Lambda}}_{\boldsymbol{\theta},(j,k)}^{\vec A} \boldsymbol{\Lambda}_{\boldsymbol{\theta},(i,l)}^{\vec A} \overline{\boldsymbol{\Lambda}}_{\boldsymbol{\theta},(j,l)}^{\vec A}\\
&= \underset{\substack{\vec A \sim \Z_2^{m \times n}\\
\boldsymbol{\theta} \sim \Theta_q^m}}{\mathbb{E}} \exp\left[\mathrm{i}\sum_{r=1}^m\theta_r((-1)^{\langle i,k|A^r\rangle}+(-1)^{\langle j,k|A^r\rangle} -(-1)^{\langle i,l|A^r\rangle}-(-1)^{\langle j,l|A^r\rangle})\right]\\
&=\left(\underbrace{\underset{\substack{\vec A \sim \Z_2^{1 \times n}\\
\boldsymbol{\theta} \sim \Theta_q^1}}{\mathbb{E}} \exp\left[\mathrm{i}\theta((-1)^{\langle i,k|A\rangle}+(-1)^{\langle j,k|A\rangle} -(-1)^{\langle i,l|A\rangle}-(-1)^{\langle j,l|A\rangle})\right]}_{=:e_{i,j,k,l}}\right)^m.
\end{split}
\end{align}
Clearly, if $\{(i,k),(j,k)\}=\{(i,l),(j,l)\}$ then $e_{i,j,k,l}=1$.
We can upper bound the number of all  tuples with this property by $2^{3n/2}\times 2!$, which leads to a contribution of $2\times 2^{-n/2}$ in Eq.~\eqref{eq:frobeniusnorm}.
Assume that $q> 4$. 
Then, for all other tuples $(i,j,k,l)$ we notice that $e_{i,j,k,l}$ is an eigenvalue $<1$ of the moment operator $\mathbb{E}U^{\otimes 2}\otimes \overline{U}^{\otimes 2}$ for $U=e^{\mathrm{i}\bigotimes_{j=1}^n Z^{\boldsymbol{A}_{1j}}}$.
Then, we can invoke Lemma~\ref{lemma:gapbound} with $t=2$ to obtain the bound $e_{i,j,k,l}\leq 1-\frac{1}{2t}=\frac34$.
Plugging this back into Eq.~\eqref{eq:frobeniusnorm} yields the simple upper bound
\begin{equation}
 \underset{\substack{\vec A \sim \Z_2^{m \times n}\\
\boldsymbol{\theta} \sim \Theta_q^n}}{\mathbb{E}}  \left[\| 2^{-n}\boldsymbol{\Lambda}_{\boldsymbol{\theta}}^{\vec A} (\boldsymbol{\Lambda}_{\boldsymbol{\theta}}^{\vec A})^\dag \|^2_F\right] \leq 2\times 2^{-\frac{n}{2}}+\left(\frac34\right)^m.
\end{equation}
In particular, for $m\geq n$ we find with Markov's inequality that 
\begin{equation}
    S(\rho)\geq -\log\left( \| 2^{-n}\boldsymbol{\Lambda}_{\boldsymbol{\theta}}^{\vec A} (\boldsymbol{\Lambda}_{\boldsymbol{\theta}}^{\vec A})^\dag \|_F\right)=\Omega(n)
\end{equation}
with probability $1-2^{-\Omega(n)}$.

To obtain a matching upper bound on the entanglement entropy observe that each unitary $e^{\mathrm{i}\phi \bigotimes_{j=1}^n Z^{y_j}}$ can be decomposed into a single qubit $Z$ rotation and at most $2n$ \textsf{CNOT} and \textsf{SWAP} gates such that at most $2$ \textsf{CNOT} or \textsf{SWAP} gates act across the cut $(A : B)$.
In particular, \textsf{CNOT} has an operator Schmidt rank of $2$ and SWAP has an operator Schmidt rank of $4$. 
Therefore, any robation of the form $e^{\mathrm{i}\phi \bigotimes_{j=1}^n Z^{y_j}}$ can increase the Schmidt rank of a state across the cut $(A : B)$ by at most $4\times 4=16$.
Overall, we find for the entanglement entropy 
\begin{equation}
  -\frac12\log(2\times 2^{-n}+(3/4)^m)\leq  S(\rho_A)\leq \log(\mathrm{rank}(\rho_A))\leq \log(16)m.
\end{equation}
The theorem then follows from choosing the values $m=n$ and $m=\omega(\log(n))$.
\end{proof}
\subsection{Pseudorandom Unitaries}
\label{sec:unitary}

In this section, we show that a variant of our $\mathsf{HPS}$ assumption might yield \emph{pseudorandom unitaries}.  \ifhaspru We first show that a unitary that repeatedly applies (uniformly random) phase oracles and Hadamards is indistinguishable from a truly Haar random unitary.

\fi
Let us first give a formal definiton of pseudorandom unitaries, as proposed by
Ji, Liu, and Song~\cite{ji2018pseudorandom}. Broadly speaking, these are efficient ensembles of unitaries which emulate Haar random unitaries to computationally bounded adversaries. 

\begin{definition}[Pseudorandom unitary]\label{def:PRU} Let $n \in \N$ be the security parameter. An infinite sequence $\mathcal{U} = \{\mathcal{U}_n\}_{n \in \N}$ of $n$-qubit unitary ensembles $\mathcal{U}_n = \{U_k\}_{k \in \mathcal{K}}$ is a called a family of pseudorandom unitaries if it satisfies the following conditions.
\begin{itemize}
    \item (\textbf{Efficient computation}) There exists an efficient quantum algorithm $\mathcal{Q}$ such that for all keys $k \in \mathcal{K}$, where $\mathcal{K}$ denotes the key space, and any $\ket{\psi} \in (\mathbb{C}^2)^{\ot n}$, it holds that
    $$
    \mathcal{Q}(k,\ket{\psi}) = U_k \ket{\psi}\,.
    $$

    \item (\textbf{Pseudorandomness}) The unitary $U_k$, for a randomly chosen key $k \sim \algo K$, is computationally indistinguishable from a Haar random unitary $U \sim \mu(2^n)$. In other words, for any efficient quantum algorithm $\algo A$, it holds that $$
    \vline\, \underset{k \sim \algo K}{\Pr}[\algo A^{U_k}(1^n)=1] - \underset{U \sim \mu(2^n)}{\Pr}[\algo A^{U}(1^n) =1]  \,\vline \,\leq \, \negl(n)\,.
    $$ 
\end{itemize}
Note that, whenever we write $\mathcal{U}_n = \{U_k\}_{k \in \mathcal{K}}$, it is implicit that the key space $\mathcal{K}$ depends on the security parameter $n \in \N$, and that the length of each key $k \in \mathcal{K}$ is polynomial in $n$.
\end{definition}

\ifhaspru
\subsubsection{Security of the $\mathsf{FHFHFC}$ pseudo-random unitary}

A recent line of work has shown that a ``$\mathsf{PFC}$'' ensemble~\cite{metger2024simpleconstructionslineardepthtdesigns}, where $\mathsf{P}$ and $\mathsf{F}$ are a pseudo-random permutation and pseudo-random function respectively (and $\mathsf{C}$ is a uniformly random Clifford) is a pseudo-random unitary~\cite{ma2024pseudorandom}.  However, as these require one-way functions, there are no suitable post-quantum instantiations of pseudo-random unitaries.  An older folklore construction of pseudo-random unitaries, first proposed by \cite{ji2018pseudorandom}, is to apply alternating pseudo-random phase unitaries and Hadamards.  In this section, we prove that a version of this construction based on the HPS assumption is secure, and propose an instantiation of these pseudo-random unitaries using $\mathsf{HPS}$ oracles. 

Let $\mathsf{F}$ a uniformly random diagonal matrix.  We can think of $\mathsf{F}$ being an ``ideal HPS'' phase oracle, which acts as follows for a tuple of angles $\theta = (\theta_z)_{z}$:
\begin{equation*}
    \mathsf{F}_{\theta} = \sum_{x} \exp\left(i \sum_{z} \theta_z (-1)^{\braket{x, z}}\right)\proj{x}\,.
\end{equation*}
Here, each angle $\theta_z$ will be randomly drawn from a discrete set $\{2\pi k/q\}_{k=0}^{q-1}$.
In the following, we will choose $q=2^n$.
We note that any exponentially large $q$ suffices.
We make two observations before continuing with the proof.  First, we call this an ideal HPS because it involves sampling an exponential number of real phases, instead of only a polynomial number of them.  Second, we note that these matrices are identically distributed to uniformly random diagonal matrices with complex phases, but we phrase it as an ideal HPS instance to motivate our conjectured pseudo-random unitary construction from HPS.  

\begin{remark}
As a matter of notation, we will write $\theta$, or $\theta^{(1)}$, to mean the tuple of angles $(\theta_1, \ldots, \theta_{2^n})$ (or with the appropriate superscript).  The association between the list of angles associated with a $\theta$ will be clear from context.
\end{remark}

We claim that the following unitary is pseudo-random, for three uniformly random diagonal unitaries with complex phases, $\mathsf{F}_1$, $\mathsf{F}_2$ and $\mathsf{F}_3$:
\begin{equation*}
    \mathsf{F}_3 \cdot H^{\otimes n} \cdot \mathsf{F}_2 \cdot H^{\otimes n} \cdot \mathsf{F}_1 \cdot \mathsf{C}\,.
\end{equation*}
We call this the $\mathsf{FHFHFC}$-ensemble, and by the previous observation, this is identical to the proposed construction from \cite{ji2018pseudorandom}.  The following is the main theorem of this section.

\begin{theorem}
\label{thm:fhfhfc_psuedorandom}
        Let $n \in \mathbb{N}$ be the number of qubits as input to a quantum adversary.  Then for all $t$-query adversaries $\mathcal{A}_t^{(\cdot)}$ making queries to a $n$ qubit oracle, the following holds
        \begin{equation*}
            \td\left(\mathbb{E}_{U\sim \mu(2^{n})}\left[\mathcal{A}_t^{U}(\ket{0})\right],\mathbb{E}_{\mathcal{O} \sim \mathsf{FHFHFC}}\left[\mathcal{A}_t^{\mathcal{O}}(\ket{0})\right] \right) \leq O\left(\sqrt{\frac{t^3}{2^n}}\right)\,.
        \end{equation*}
\end{theorem}

In order to prove that these are pseudo-random, we rely on the results of \cite{ma2024pseudorandom}, which showed that a certain isometry, called the path-recording isometry, is indinstinguishable from a Haar random unitary for bounded-query adversaries.
We state the definition of the path-recording isometry here.
\begin{definition}[Path recording isometry]
    Define the path recording isometry $\mathsf{PR}$, which acts on two registers, an input $\mathsf{A}$, and database $\mathsf{XY}$ containing $t$ input-output pairs, as follows:
    \begin{equation*}
       \mathsf{PR} \ket{x}_{\reg{A}} \ket{R}_{\reg{XY}} \mapsto \frac{1}{\sqrt{2^n - t}}\sum_{y \not\in \mathrm{Im}(R)} \ket{y}_{\reg{A}} \ket{R \cup (x, y)}_{\reg{XY}}\,. 
    \end{equation*}
\end{definition}

Before we begin with the proof, we comment on the construction itself: Why take three rounds of $\mathsf{F}$?  Intuitively, applying a single random function $\mathsf{F}$ records the input to the function, as noted in \cite{zhandry2019record}.  
Therefore, we can see that each of the three random phase oracles has a distinct ``job'': The first one records the input, and the last one records the output.  The middle one records the relationship between the input and output, as we will see later.  Where our proof differs from the proof of \cite{ma2024pseudorandom} is that the middle unitary is not able to record all \emph{bijective} relations, but instead can only record so called ``collision free'' relations.  We prove that, with high probability, a bijection with random output is collision free, which allows us to perform path recording with slightly more error than incurred in \cite{ma2024pseudorandom}.

As the first step of the proof, we will imagine a purified version of the $\mathsf{FHFHFC}$ ensemble, where the phases are stored in a register that is not visible to the adversary.
We will show that the $\mathsf{FHFHF}$ oracle exhibits behavior similar to the path recording isometry, up to an isometry that acts on the purifying register containing the phases.

We first argue that we can purify the action of $\mathsf{FHFHF}$ into another oracle we denote by $\mathsf{FHFHFO}$.
\begin{definition}
    $\mathsf{FHFHFO}$ is a unitary oracle that acts on two registers as follows:
    \begin{equation*}
        \mathsf{FHFHFO} \ket{x} \ket{\theta^{(1)}, \theta^{(2)}, \theta^{(3)}} = \mathsf{F}_{\theta^{(3)}}\cdot H^{\otimes n} \cdot\mathsf{F}_{\theta^{(2)}} \cdot H^{\otimes n} \cdot \mathsf{F}_{\theta^{(1)}} \ket{x} \ket{\theta^{(1)}, \theta^{(2)}, \theta^{(3)}}\,. 
    \end{equation*}
\end{definition}
\begin{lemma}[Purification technique for $\mathsf{FHFHF}$]
    \label{lem:purification_fhf}
    The view of any adversary querying $\mathsf{FHFHF}$ is equivalent to the view of the adversary querying $\mathsf{FHFHFO}$ when the purifying register is initialized in the following state:
    \begin{equation*}
        \frac{1}{2^{3n}}\sum_{\substack{\theta_{1}^{(1)}, \ldots, \theta_{2^n}^{(1)}\\ \theta_{1}^{(2)}, \ldots, \theta_{2^n}^{(2)}\\ \theta_{1}^{(3)}, \ldots, \theta_{2^n}^{(3)}}}\ket{\theta^{(1)}, \theta^{(2)}, \theta^{(3)}}\,.
    \end{equation*}
\end{lemma}
\begin{proof}
    Since the view of the adversary does not contain the purifying register, when we trace them out we will get a uniform mixture of states corresponding to $\mathsf{FHFHF}$ for uniformly random angles $\theta^{(1)}$, $\theta^{(2)}$, and $\theta^{(3)}$.  
    This corresponds exactly to choosing $\theta^{(1)}$, $\theta^{(2)}$, and $\theta^{(3)}$ uniformly at random.
\end{proof}

We now define a family of orthogonal states that correspond to a database $R$.
We some new notation here to make the remaining equations less cumbersome.  
For a set $S$ of $n$-bit strings and $n$-bit string $z$, define $S_z$ to be:
\begin{equation*}
    S_z = \sum_{x \in S} (-1)^{\braket{x, z}}\,.
\end{equation*}
For a relation $R$, define the following sets derived from $R$:
\begin{align*}
    \mathrm{Dom}(R) &= \{x : (x, y) \in R\}\\
    \mathrm{Im}(R) &= \{y : (x, y) \in R\}\,.
\end{align*}

As a first step, we define so-called $\mathsf{FHFHF}$-set states, as follows:
\begin{definition}[$\mathsf{FHFHF}$-set states]
    Given a set $S$, define the set $\ket{\mathsf{fhfhf}_S}$ as follows
    \begin{equation*}
        \frac{1}{2^n}\sum_{\theta_1, \ldots, \theta_{2^{n}} \in [2^n]} \exp\left(i \sum_{z} \theta_z S_z\right)\ket{\theta}\,.
    \end{equation*}
\end{definition}

The following lemma will show that these set states are orthogonal whenever the sets are different.

We show that these states form an orthogonal basis. In order to prove this, we first prove the following property, which comes from Fourier analysis of Boolean functions.
\begin{lemma}\label{lemma:BooleanFunctionsHaveDifferentFourierTransforms}
    Let $X$ and $Y$ be two multi-sets of $t < 2^{n-1}$ elements that are not equal, then there exists a string $z$ such that
    \begin{equation*}
        \sum_{x \in X}(-1)^{\braket{z, x}} \neq \sum_{y \in Y}(-1)^{\braket{z, y}}\,.
    \end{equation*}
\end{lemma}
\begin{proof}
    The value $\sum_{x \in X} (-1)^{\braket{z, x}}$ is exactly the $z$'th Fourier coefficient of a function that takes value $m(x)$ in $x$, where $m(x)$ is the multiplicity of $x$ in the set $X$.  This also holds for the sum over $Y$.  
    Since two functions have the same Fourier coefficients if and only if they are the same function, if no $z$ exists, then we get a contradiction with the fact that $
    X$ and $Y$ are different.
    Thus, there must exist a $z$ for which these values differ.
\end{proof}
We note that this proof stops working when $t$ is allowed to be $2^{n-1}$, as two sets whose union is $\{0, 1\}^{n}$ will have identical Fourier coefficients. 

\begin{lemma}[Orthogonality of set states]
    \label{lem:set_states_orthogonal}
    Set $S$, $S'$ be two different multi-sets of $n$-bit strings, then the set states $\ket{\mathsf{fhfhf}_S}$ and $\ket{\mathsf{fhfhf}_{S'}}$ are orthogonal.
\end{lemma}
\begin{proof}
    Let $z^*$ be the bit string on which $S_z$ and $S'_z$ differ (given by \cref{lemma:BooleanFunctionsHaveDifferentFourierTransforms}).  
    Then we have the following
    \begin{align*}
       \braket{\mathsf{fhfhf}_S| \mathsf{fhfhf}_{S'}} &= \frac{1}{2^{2n}} \sum_{\theta_1, \ldots, \theta_{2^n}} \left(i\sum_{z} \theta_z (S_z - S'_z)\right)\\
       &= \frac{1}{2^{2n}} \sum_{\theta_{z^*}} \exp\left(i\theta_{z^*} S_{z^*} - S'_{z^*})\right) \cdot  \sum_{\theta_i | i\neq z^*} \exp\left(i \sum_{z \neq z^*} \theta_z (S_z - S'_z)\right)\\
       &= 0\,. 
    \end{align*}
    Here we use the fact that the following equality:
    \begin{equation*}
        \mathbb{E}_{\theta} e^{i\theta x} = \begin{cases}
            1\quad x = 0\\
            0\quad \text{otherwise}
        \end{cases}.
    \end{equation*}
     Moreover, it is easy to check tht $\mathbb{E}_{k} e^{\mathrm{i}2\pi k/q}=\mathbb{E}_\theta e^{i\theta x}$ for an integer $x$ with $|x|<q$.
     Since $S_{z^*} - S'_{z^*}$ is not $0$, this term evalutes to $0$ and therefore the entire product is $0$.
\end{proof}

Now we define the $\mathsf{FHFHF}$-relation state, which will be how the $\mathsf{FHFHF}$ oracle simulates the path recording oracle.  
\begin{definition}[[$\mathsf{FHFHF}$-relation states]
    For $0 \leq t \leq 2^n$, and a relation $R$, define the $\mathsf{FHFHF}$-relation states to be as follows
    \begin{multline*}
        \ket{\mathsf{fhfhf}_{R}} = \sqrt{\frac{1}{2^{nt}}}\sum_{w_1, \ldots, w_t \in [2^n]^t}\left(\prod_{i}(-1)^{\braket{w_i, y_i \oplus x_i}}\right)\\\frac{1}{2^{3n}}\sum_{\substack{\theta_{1}^{(1)}, \ldots, \theta_{2^n}^{(1)}\\ \theta_{1}^{(2)}, \ldots, \theta_{2^n}^{(2)}\\ \theta_{1}^{(3)}, \ldots, \theta_{2^n}^{(3)}}}\exp\left(i\sum_{z, i} \theta^{(1)}_z(-1)^{\braket{z, x_i}} + \theta^{(2)}_z(-1)^{\braket{z, w_i}} + (-1)^{\braket{z, y_i}} \theta_z^{(3)}\right)\\\ket{\theta^{(1)}} \otimes \ket{\theta^{(2)}} \otimes \ket{\theta^{(3)}}\,.
    \end{multline*}
\end{definition}

\begin{remark}
    \label{rem:relations_as_sets}
    We note that by inspection, we can re-write the state $\ket{\mathsf{fhfhf}_R}$ as follows:
    \begin{equation*}
        \ket{\mathsf{fhfhf}_R} = \ket{\mathsf{fhfhf}_{\mathrm{Dom}(R)}} \left( \sqrt{\frac{1}{2^{nt}}} \sum_{w_1, \ldots, w_t \in [2^n]^{t}} \prod_{i = 1}^{t} \left((-1)^{\braket{w_i, x_i \oplus y_i}}\right)\ket{\mathsf{fhfhf}_{\{w_1, \ldots, w_t\}}} \right) \ket{\mathsf{fhfhf}_{\mathrm{Im}(R)}}\,.
    \end{equation*}
\end{remark}
Now we deviate from the analysis of the $\mathsf{PFC}$-ensemble in two important ways.  Recall that in \cite{ma2024pseudorandom}, the authors show that (a) the relation states are orthogonal to each other, and (b) the $\mathsf{PFC}$-ensemble appends an input-output pair to the relation states.  We will find it easier to show that a slightly different set of states are an orthonormal basis.  Consider the following ``distinct [$\mathsf{FHFHF}$-relation  state''.
\begin{definition}[Distinct [$\mathsf{FHFHF}$-relation states]
    For $0 \leq t \leq 2^n$, and a relation $R$, define the distinct $\mathsf{FHFHF}$-relation state as follows:
    \begin{equation*}
        \ket{\mathsf{fhfhf}^{\mathrm{dist}}_R} = \ket{\mathsf{fhfhf}_{\mathrm{Dom}(R)}} \left( \sqrt{\frac{(2^{n}-t)!}{2^{n}!}} \sum_{w_1, \ldots, w_t \in [2^n]^{t}_{\mathrm{dist}}} \prod_{i = 1}^{t} \left((-1)^{\braket{w_i, x_i \oplus y_i}}\right)\ket{\mathsf{fhfhf}_{\{w_1, \ldots, w_t\}}} \right) \ket{\mathsf{fhfhf}_{\mathrm{Im}(R)}}\,.
    \end{equation*}
\end{definition}
The difference between the previous and distinct relation state is that the internal sum in the modified state only sums over distinct choices of $w_1, \ldots, w_t$.    
A second difference between the $\mathsf{FHFHF}$ and $\mathsf{PFC}$ ensembles is that the $\mathsf{FHFHFC}$-relation states will not work on bijections, but on a slightly more restrictive set of relations we call ``collision-free'' relations, defined below.
\begin{definition}[Collision-free relation]
    Let $x_1, \ldots x_t$ and $y_1, \ldots, y_t$ be a collection of $n$-bit string.  Then we say the pair is called ``collision-free'' if the following holds. For all $x_i$, $x_j$ (with $i$ potentially equal to $j$), and $y_i$ and $y_k$ (again with $k$ potentially equal to $i$), the quantities $x_i \oplus x_j \oplus y_i \oplus y_k$ are distinct from each other.
    Given a set of $x_1, \ldots, x_t$, let $\mathsf{CF}_{\vec{x}}$ be the set of collision free $y_1, \ldots, y_t$.
\end{definition}

Similar to bijective relations, a randomly sampled set of outputs is collision-free with overwhelming probability, which the following lemma shows. 
\begin{lemma}[Collision probability bound]
    \label{lem:collision_probability_bound}
    Fix distinct $x_1, \ldots, x_t$, then the probability over uniformly chosen, distinct $y_1, \ldots, y_t$ that the pair of sets is not collision free is at most $t^8/2^n$.
\end{lemma}
\begin{proof}
    A choice of $y_1, \ldots, y_t$ forms a collision-free collection if every pair $y_a \oplus y_b \oplus y_c \oplus y_d$ is not equal to some $x_i \oplus x_j \oplus x_k \oplus x_\ell$.  This can be seen by setting two different pairs equal to each other and re-arranging terms so that the $x$'s and $y$'s are on the same side.  Since there are at most $t^4$ choices for $x_i \oplus x_j \oplus x_k \oplus x_{\ell}$, every quadruple of $y_a \oplus y_b \oplus y_c \oplus y_d$ has probability $t^4 / 2^n$ of being equal to one.  Taking a union bound over all $t^4$ choices of $a$, $b$, $c$, and $d$ yields the desired bound.
\end{proof}

We now prove that these states form an orthonormal basis whenever the relation is collision-free.  First we prove the following lemma.
\begin{lemma}
    \label{lem:permutation_sum}
    Let $\sigma \in S_{t}$ be a permutation with a single cycle for $t > 1$, and let $(z_{i})_{i = 1}^{t}$ be a collection of $t$-many $n$-bit strings such that not all $z_{i}$ are the same.  Then the following holds   
    \begin{equation*}
        \sum_{w_1, \ldots, w_t \in [2^n]^{t}}\prod_{i = 1}^{t}(-1)^{\braket{w_{i} \oplus w_{\sigma(i)}, z_i}} = 0\,.
    \end{equation*}
\end{lemma}
\begin{proof}
    We can assume without loss of generality that $\sigma$ maps $i$ to $i+1$.  Then we can re-write the expression in the lemma statement as follows.
    \begin{align*}
        \sum_{w_1, \ldots, w_t \in [2^n]^{t}}\prod_{i = 1}^{t}(-1)^{\braket{w_{i} \oplus w_{\sigma(i)}, z_i}} &= \sum_{w_1, \ldots, w_t \in [2^n]^{t}}\prod_{i = 1}^{t}(-1)^{\braket{w_{i}, z_i \oplus z_{i-1}}}\\
        &= \prod_{i = 1}^{t}\left(\sum_{w_i \in [2^n]^{t}}(-1)^{\braket{w_{i}, z_i \oplus z_{i-1}}} \right)\,.
    \end{align*}
    Now we note that if any $z_{i} \oplus z_{i-1}$ is not $0$, the inner sum will be $0$, and as long as not all of the $z_i$ are exactly the same, this will be the case. 
\end{proof}

Note that whenever all $z_i$ are identical, we can use the same formula from the previous lemma to compute that the sum will equal $2^{nt}$.
Therefore, for an arbitrary set of $z$'s, we can write the expression as $2^{nt} \mathds{1}(z_1 = \ldots = z_t)$.

We will use the following lemma about the average inner product over distinct strings.  
\begin{lemma}[Inner product of distinct vectors]
\label{lem:inner_product_distinct}
    Let $v_1, \ldots, v_t$ be distinct $n$-bit strings (and when $t = 1$, $v_1$ is not $0^n$).  Then the following holds: 
    \begin{equation*}
        \sum_{w_1, \ldots, w_t \in [2^n]^{t}_{\mathrm{dist}}} \prod_{i = 1}^{t} (-1)^{\braket{w_i, v_i}} = 0\,.
    \end{equation*}
\end{lemma}
\begin{proof}
    We proceed by induction.  In the base case, if we have a single $v_1$ that is not $0^n$, then summing over all $w_1$ yields $0$.

    Now, given more than $1$ string $v_1, \ldots, v_t$, consider the following:
    \begin{align*}
        \sum_{w_1, \ldots, w_t \in [2^n]^{t}_{\mathrm{dist}}} \prod_{i = 1}^{t} (-1)^{\braket{w_i, v_i}} &= \sum_{w_1, \ldots, w_t \in [2^n]^{t}_{\mathrm{dist}}} \prod_{i = 1}^{t} (-1)^{\braket{w_i, v_i \oplus v_1}} (-1)^{\sum_{i} \braket{v_i, v_1}}\\
        &= (-1)^{\sum_{i} \braket{v_i, v_1}} \sum_{w_2, \ldots, w_t \in [2^{n}]^{t-1}_{\mathrm{dist}}} \sum_{w_1 \in [2^{n}] \backslash \{w_2, \ldots, w_t\}} \prod_{i = 2}^{t} (-1)^{\braket{w_i, v_i \oplus v_1}}\\
        &= (-1)^{\sum_{i} \braket{v_i, v_1}}(2^n - t - 1)\sum_{w_2, \ldots, w_t \in [2^{n}]^{t-1}_{\mathrm{dist}}} \prod_{i = 2}^{t} (-1)^{\braket{w_i, v_i \oplus v_1}}\\
        &= 0\,.
    \end{align*}
    In the first line, we use the fact that the inner product is bi-linear in each of its inputs.  Then we pull out the factor of $(-1)^{\sum_{i}\braket{v_i, v_1}}$, since it does not depend on $w_1, \ldots, w_t$.  Then, since $v_1 \oplus v_1$ is $0$, summing over $w_1$ does not change the term being summed, so it becomes a multiplicative factor of $2^{n} - t - 1$.  Then, we use the inductive hypothesis with $w_2, \ldots, w_t$.  Notice that every $v_i \oplus v_1$ is distinct from each other, and they will never all be $0^n$ when we get to the base case.
\end{proof}

A corollary of \Cref{lem:inner_product_distinct} is that when the sum in \Cref{lem:permutation_sum} is taken to be over distinct strings, the result is still $0$ as long as $z_i$ are distinct. Now we can compute, the normalization of the relation states that we described.

\begin{lemma}
    \label{lem:normalized_relation_state}
    Every for every collision-free relation $R$, the distinct [$\mathsf{FHFHF}$-relation state is normalized. 
\end{lemma}
\begin{proof}
    First note that the set states corresponding to the domain and image will have inner product $1$, so we can ignore those and focus on the middle set state with the multi-set $\{w_1, \ldots, w_t\}$. 
    Also note that we can write a permutation $\sigma = (c_1, \ldots, c_m)$ as a collection of disjoint cycles.  We will use the notation $\mathrm{cyc}(\sigma)$ to denote the collection of cycles, and $\mathrm{cyc}(\sigma, j)$ to denote the $j$'th cycle.  The ordering of that $\mathrm{cyc}$ outputs the cycles does not matter so long as it is consistent.
    Then we can compute the inner product as follows.
    \begin{align*}
       \braket{\mathsf{fhfhf}_R^{\mathrm{dist}} | \mathsf{fhfhf}_R^{\mathrm{dist}}} &= \frac{(2^{n}-t)!}{2^{n}!} \sum_{\substack{w_1, \ldots, w_t \in [2^n]^t_{\mathrm{dist}}\\ v_1, \ldots, v_t \in [2^n]^t_{\mathrm{dist}}}} \prod_{i = 1}^{t} (-1)^{\braket{w_i, x_i \oplus y_i} + \braket{v_i, x_i \oplus y_i}}\mathds{1}(\{w_1, \ldots, w_t\} = \{v_1, \ldots, v_t\})\\
       &= \frac{(2^{n}-t)!}{2^{n}!} \sum_{w_1, \ldots, w_t \in [2^n]^t_{\mathrm{dist}}} \sum_{\sigma \in S_{t}} \prod_{i = 1}^{t} (-1)^{\braket{w_{i} \oplus w_{\sigma(i)}, x_i \oplus y_i}}\\ 
       &= \frac{(2^{n}-t)!}{2^{n}!} \sum_{\sigma \in S_{t}} \sum_{w_1, \ldots, w_{t} \in [2^n]^t} \prod_{i = 1}^{t}(-1)^{\braket{w_{p(i)} \oplus w_{p(\sigma(i))}, x_i \oplus y_i}}\\
       &= \frac{(2^{n}-t)!}{2^{n}!} \sum_{\sigma \in S_{t}} \sum_{w_1, \ldots, w_{t} \in [2^n]^t} \prod_{j = 1}^{|\mathrm{cyc}(\sigma)|} \left(\prod_{i \in \mathrm{cycles}(\sigma)_{j}}(-1)^{\braket{w_{i} \oplus w_{\sigma(i)}, x_i \oplus y_i}}\right)\\
       &= \frac{(2^{n}-t)!}{2^{n}!} \sum_{\sigma \in S_{t}}  \prod_{j = 1}^{|\mathrm{cyc}(\sigma)|} \left(\sum_{\substack{w_{\mathrm{cyc}(\sigma, j)_1}, \ldots,\\w_{\mathrm{cyc}(\sigma, j)_{|\mathrm{cyc}(\sigma, j)|}}} \in [2^n]^{|\mathrm{cyc}(\sigma, j)|}}\prod_{i \in \mathrm{cyc}(\sigma, j)}(-1)^{\braket{w_{i} \oplus w_{\sigma(i)}, x_i \oplus y_i}}\right)\\
       &= \frac{(2^{n}-t)!}{2^{n}!}\frac{2^{n}!}{(2^{n}-t)!} = 1\,.
    \end{align*}
    Here the first line is using the fact that two multi-sets of distinct elements are the same if and only if their elements are the same, up to a permutation of the indices.  Then we re-arrange terms and break the permutation $\sigma$ into its disjoint cycles.
    Once we have an inner product over cycles, we can switch the order of the second sum and first product, and then apply \Cref{lem:permutation_sum} and \Cref{lem:inner_product_distinct}, using the fact that $R$ is collision-free, so all $x_i \oplus y_i$ are distinct.
    This means that for all non-identity permutations, the term in the sum is $0$, and we are left with only the contribution of the identity permutation, which is $2^{nt}$.
\end{proof}

Note that following the same proof (but without needing to apply \Cref{lem:inner_product_distinct}) shows that the non-distinct $\mathsf{FHFHF}$ relation states are also normalized quantum states, whenever the relation $R$ is collision-free.  Now we show that the \emph{distinct} relation states form an orthonormal basis when the relations are collision-free. 

\begin{lemma}[Orthogonality of distinct relation states]
    Let $R$ and $R'$ be two different \emph{collision-free} relations of size at most $2^{n-1}$ elements.  Then $\ket{\mathsf{fhfhf}^{\mathrm{dist}}_R}$ and $\ket{\mathsf{fhfhf}^{\mathrm{dist}}_{R'}}$ are orthogonal.
\end{lemma}
\begin{proof}
    From \cref{lem:set_states_orthogonal}, together with \cref{rem:relations_as_sets}, if the image or domain of the two relations is a different multi-set, then the states will be orthogonal since the set states storing the image and domain will differ.
    Therefore, we only have to consider pairs of collision free relations that have the same inputs and outputs.  However, by assumption these are collision free.
    Say that these two relations are related by a permutation $\pi$ on the $y_1, \ldots, y_t$ for every permutation $\sigma$, then for all $\sigma$, there is some index $i_{\sigma}$ satisfying $x_{i} \oplus y_i \neq x_{\sigma^{-1}(i)} \oplus y_{(\pi(\sigma^{-1}(i))}$, because of the collision-free property.

    Then the we compute the inner product between these two relation states as follows.
    \begin{align*}
        & \braket{\mathsf{fhfhfh}^{\mathrm{dist}}_R | \mathsf{fhfhfh}^{\mathrm{dist}}_{R'}}\\ &\hspace{10mm}=
       \frac{(2^n - t)!}{2^n!}\sum_{\substack{w_1, \ldots, w_t \in [2^n]^t_{\mathrm{dist}}\\ w'_1, \ldots, w'_t  \in [2^n]^t_{\mathrm{dist}}}} \prod_{i = 1}^{t} (-1)^{\braket{w_i, (x_i \oplus y_i)}+ \braket{w_i', (x_i \oplus y_{\pi(i)})}} \braket{\mathsf{fhfhf}_{\{w_1, \ldots, w_t\}}|\mathsf{fhfhf}_{\{w'_1, \ldots, w'_t\}}}\\
        &\hspace{10mm}= \frac{(2^n - t)!}{2^n!} \sum_{\substack{w_1, \ldots, w_t  \in [2^n]^t_{\mathrm{dist}} \\ w_1', \ldots, w_t'  \in [2^n]^t_{\mathrm{dist}}\\ \{w_1, \ldots, w_t\} = \{w_1', \ldots, w_t'\}}} \prod_{i} (-1)^{\braket{w_i, (x_i \oplus y_i)} + \braket{w'_i, (x_i \oplus y_{\pi(i)})}}\\
        &\hspace{10mm}= \frac{(2^n - t)!}{2^n!} \sum_{\sigma \in S_{t}} \sum_{w_1, \ldots, w_t  \in [2^n]^t_{\mathrm{dist}}} \prod_{i = 1}^{t} (-1)^{\braket{w_i, (x_i \oplus y_i)} + \braket{w_{\sigma(i)}, (x_i \oplus y_{\pi(i)})}}\\
        &\hspace{10mm}= \frac{(2^n - t)!}{2^n!} \sum_{\sigma \in S_{t}} \sum_{w_1, \ldots, w_t  \in [2^n]^t_{\mathrm{dist}}} \prod_{i = 1}^{t}  (-1)^{\braket{w_i, (x_i \oplus y_i)} + \braket{w_{i}, (x_{\sigma^{-1}(i)} \oplus y_{\pi(\sigma^{-1}(i))})}}\\ 
        &\hspace{10mm}= \frac{(2^n - t)!}{2^n!} \sum_{\sigma \in S_{t}} \sum_{w_1, \ldots, w_t  \in [2^n]^t_{\mathrm{dist}}} \prod_{i = 1}^{t}  (-1)^{\braket{w_i, (x_i \oplus y_i \oplus x_{\sigma^{-1}(i)} \oplus y_{\pi(\sigma^{-1}(i))})}}\\ 
        &\hspace{10mm}= 0\,.
    \end{align*}
    Here in the first line, we expand out the definition of the relation states, using the fact that the mutli-set of inputs and outputs is the same.  Then we note that the braket cancels all terms for which the multi-set of $w$ is not the same as the multi-set of $w'$, and distinct tuple with the same multi-set is related by a permutation of elements.  Then we re-label the indices of the sum so that instead of $w_{\sigma(i)}$ appearing along-side $x_i$, $w_i$ appears alongside $x_{\sigma^{-1}(i)}$.  Then, since the inner product is bi-linear in both terms, we can combine the two inner products.  Finally, we apply \Cref{lem:inner_product_distinct}, together with the fact that by the collision-free property, for every choice of $\sigma$, every $x_i \oplus y_i \oplus x_{\sigma^{-1}(i)} \oplus y_{\pi(\sigma^{-1}(i))}$ is distinct. This implies that for every $\sigma$, the inner sum evaluates to $0$, giving us $0$ overall.
\end{proof}

The next step in the proof is to show that the distinct relation states are close to the ordinary relation states:
\begin{lemma}
\label{lem:relation_states_close}
    For every relation $R$ with $t$ elements, the following holds 
    \begin{equation*}
        \td(\proj{\mathsf{fhfhf}_R} , \proj{\mathsf{fhfhf}^{\mathrm{dist}}_R}) \leq \sqrt{\frac{t^2}{2^n}}\,.
    \end{equation*}
\end{lemma}
\begin{proof}
    The proof will proceed by observing that the distinct relation states are the post-measurement state after measuring a projector onto set states corresponding to distinct sets.  Formally, define the following
    \begin{equation*}
        \Pi^{t}_{\mathrm{dist}} = \sum_{w_1, \ldots, w_i \in [2^n]^t_{\mathrm{dist}}} \proj{\mathsf{fhfhf}_{\{w_1, \ldots, w_t\}}}\,.
    \end{equation*}
    Note that since the set states are orthogonal from \Cref{lem:set_states_orthogonal}, this is a projector.  Then we see the following
    \begin{multline*}
        \id \otimes \Pi^{t}_{\mathrm{dist}} \otimes \id \ket{\mathsf{fhfhf}_{R}} \\= \ket{\mathsf{fhfhf}_{\mathrm{Dom}(R)}} \left( \sqrt{\frac{1}{2^{nt}}} \sum_{w_1, \ldots, w_t \in [2^n]^{t}_{\mathrm{dist}}} \prod_{i = 1}^{t} \left((-1)^{\braket{w_i, x_i \oplus y_i}}\right)\ket{\mathsf{fhfhf}_{\{w_1, \ldots, w_t\}}} \right) \ket{\mathsf{fhfhf}_{\mathrm{Im}(R)}}
    \end{multline*}
    Here we just use the orthogonality of the set states, so the terms in the sum where $w_1, \ldots, w_t$ are not distinct become $0$.  This state is exactly $\ket{\mathsf{fhfhf}_{R}^{\mathrm{dist}}}$, un-normalized.  
    Since we know that if the $2^{nt}$ was replaced with the appropriate constant, the state would become the normalized distinct relation state, we can compute the probability of measuring the distinct outcome as follows.
    \begin{align*}
        &\Tr\left[(\id \otimes \Pi^t_{\mathrm{dist}} \otimes \id) \proj{\mathsf{fhfhf}_{R}}\right] \\
        &\hspace{10mm}= \Tr\left[(\id \otimes \Pi^t_{\mathrm{dist}} \otimes \id)\proj{\mathsf{fhfhf}_{R}} (\id \otimes \Pi^t_{\mathrm{dist}} \otimes \id)\right] \\
        &\hspace{10mm}= \frac{1}{2^{nt}} \sum_{\substack{w_1, \ldots, w_t \in [2^n]^t_{\mathrm{dist}}\\ v_1, \ldots, v_t \in [2^n]^t_{\mathrm{dist}}}} \prod_{i = 1}^{t}(-1)^{\braket{w_i, x_i \oplus y_i} + \braket{v_i, x_i \oplus y_i}} \mathds{1}(\{w_1, \ldots, w_t\} = \{v_1, \ldots, v_t\}))\\
        &\hspace{10mm}= \frac{1}{2^{nt}} \sum_{w_1, \ldots, w_t \in [2^n]^t_{\mathrm{dist}}} \sum_{\sigma \in S_t} \prod_{i = 1}^{t} (-1)^{\braket{w_i \oplus w_{\sigma(i)}, x_i \oplus y_i}}\\
        &\hspace{10mm}= \frac{1}{2^{nt}} \frac{2^{n}!}{(2^{n}-t)!}\\
        &\hspace{10mm} \geq 1 - \frac{t(t-1)}{2^{n}}\,.
    \end{align*}
    Here we notice that the third line is identical to the expression we had in \Cref{lem:normalized_relation_state}, and the final line is exactly the probability of measuring distinct random strings.
    To complete the lemma, we apply the gentle measurement lemma.
\end{proof}

Now, we have shown that there is an orthogonal set of states that correspond to relations.  The next step is to show that the $\mathsf{FHFHF}$ oracle, when run on a state $\ket{x} \ket{\mathsf{fhfhf}_{R}}$, actually acts to append $(x, y)$ to the relation state.  We note that in the following lemma, we show that it creates a \emph{non-distinct} relation state.
\begin{lemma}
    \label{lem:fhfhf_action}
    For $t = \poly(n)$, relation $R = \{(x_1, y_1), \ldots, (x_t, y_t)\}$, and $x \in \{0, 1\}^{n}$, the following holds
    \begin{equation*}
        \mathsf{FHFHFO}\ket{x}\ket{\mathsf{fhfhf}_R} = \frac{1}{2^{n/2}}\sum_{y} \ket{y} \ket{\mathsf{fhfhf}_{R \cup \{(x, y)\}}}\,.
    \end{equation*}
\end{lemma}
\begin{proof}
    The proof proceeds by evaluating the action of $\mathsf{FHFHF}$ on the input state.
    Let $\alpha_{R}^{\Theta_1, \Theta_2, \Theta_3}$ be the coefficient that appears in the integral in the definition of $\ket{\mathsf{fhfhf}_R}$.
    We first apply the first function, yielding the following
    \begin{equation*}
        \ket{x} \frac{1}{2^{3n}}\sum_{\substack{\theta_{1}^{(1)}, \ldots, \theta_{2^n}^{(1)}\\ \theta_{1}^{(2)}, \ldots, \theta_{2^n}^{(2)}\\ \theta_{1}^{(3)}, \ldots, \theta_{2^n}^{(3)}}}\alpha_{R}^{\theta^{(1)}, \theta^{(2)}, \theta^{(3)}}\exp\left(i\sum_{z}\theta_z^{(1)} (-1)^{\braket{x, z}}\right)\ket{\theta^{(1)}, \theta^{(2)}, \theta^{(3)}}\,. 
    \end{equation*}
    Then we apply the Hadamard, which yields the following state (without normalization)
    \begin{equation*}
       \frac{1}{2^{n/2}}\sum_{y \in \{0, 1\}^{n}}(-1)^{\braket{x, w}}\ket{w}\frac{1}{2^{3n}}\sum_{\substack{\theta_{1}^{(1)}, \ldots, \theta_{2^n}^{(1)}\\ \theta_{1}^{(2)}, \ldots, \theta_{2^n}^{(2)}\\ \theta_{1}^{(3)}, \ldots, \theta_{2^n}^{(3)}}} \alpha_{R}^{\theta^{(1)}, \theta^{(2)}, \theta^{(3)}}\exp\left(i\sum_{z}\theta_z^{(1)} (-1)^{\braket{x, z}}\right)\ket{\theta^{(1)}, \theta^{(2)}, \theta^{(3)}}\,. 
    \end{equation*}
    Finally, we apply the second ideal HPS to the state, yielding the following
    \begin{equation*}
        \frac{1}{2^{n/2}}\sum_{w \in \{0, 1\}^{n}}(-1)^{\braket{x, w}}\ket{w}\sum_{\substack{\theta_{1}^{(1)}, \ldots, \theta_{2^n}^{(1)}\\ \theta_{1}^{(2)}, \ldots, \theta_{2^n}^{(2)}\\ \theta_{1}^{(3)}, \ldots, \theta_{2^n}^{(3)}}} \alpha_{R}^{\theta^{(1)}, \theta^{(2)}, \theta^{(3)}}\exp\left(i\sum_{z}\theta_z^{(1)} (-1)^{\braket{x, z}} + \theta_z^{(2)} (-1)^{\braket{w, z}}\right)\ket{\theta^{(1)}, \theta^{(2)}, \theta^{(3)}}\,. 
    \end{equation*}
    Finally, we apply another Hadamard and a final ideal HPS unitary to get the following state.
    \begin{multline*}
        \frac{1}{2^{n}}\sum_{w, y \in \{0, 1\}^{n}}(-1)^{\braket{y, x\oplus w}}\ket{y}\frac{1}{2^{3n}}\sum_{\substack{\theta_{1}^{(1)}, \ldots, \theta_{2^n}^{(1)}\\ \theta_{1}^{(2)}, \ldots, \theta_{2^n}^{(2)}\\ \theta_{1}^{(3)}, \ldots, \theta_{2^n}^{(3)}}} \alpha_{R}^{\theta^{(1)}, \theta^{(2)}, \theta^{(3)}} \cdot\\\exp\left(i\sum_{z}\theta_z^{(1)} (-1)^{\braket{x, z}} + \theta_z^{(2)} (-1)^{\braket{w, z}} + \theta_z^{(3)} (-1)^{\braket{y, z}}\right)\ket{\theta^{(1)}, \theta^{(2)}, \theta^{(3)}}\,. 
    \end{multline*}
    Expanding out the definition of $\alpha_{R}^{\theta^{(1)}, \theta^{(2)}, \theta^{(3)}}$, pushing the sum over $w$ into the purifying register, and combining the product of exponentials yields the desired state.  
\end{proof}

Now, we have shown that the distinct relation states for the $\mathsf{FHFHF}$ oracle are orthogonal whenever the relation is collision free, and that the action of the $\mathsf{FHFHF}$ oracle is to output a random $y$ and store $(x, y)$ in the relation state.  
Now, we make the connection to the path recording oracle formal by defining two isometries that will act on the purifying register, $\mathsf{Compress}$ and $\mathsf{Collide}$.  The compress isometry acts as follows:
\begin{equation*}
    \mathsf{Compress} = \sum_{R\ \mathrm{collision\ free}} \ket{R}\!\!\bra{\mathsf{fhfhf}^{\mathrm{dist}}_R}\,.
\end{equation*}
The collide isometry, which restricts us to collision-free relations, will act as follows.
\begin{equation*}
    \mathsf{Collide} = \sum_{R\ \mathrm{collision\ free}} \proj{\mathsf{fhfhf}^{\mathrm{dist}}_R}\,.
\end{equation*}
Let $\Pi^{(t)}$ be the projection onto bijection \emph{relation states}, as defined in \cite{ma2024pseudorandom}, and $\Lambda^{(t)}$ be the projection onto \emph{collision-free} relation states.  Then we have the following lemma, which states that compressing collision-free $\mathsf{FHFHF}$ relation states yields something very close to the effect of the path-recording oracle from \cite{ma2024pseudorandom}. 
\begin{lemma}
\label{lem:compressing_bijective_relations}
    For all adversaries $\mathcal{A}_t^{\mathcal{O}}$ that make $t$ queries, let $\ket{\mathcal{A}_t^{\mathcal{O}}}$ be the state of the adversary after making $t$ queries to $\mathcal{O}$.
    Then for all $n$-qubit unitaries $G$, the following holds.
\begin{multline*}
    \td\left(\mathsf{Compress} \cdot \mathsf{Collide}\cdot \proj{\mathcal{A}_t^{\mathsf{FHFHFO}\cdot G}} \cdot \mathsf{Collide} \cdot \mathsf{Compress},  \Lambda^{(t)}\cdot \proj{\mathcal{A}^{\mathsf{PR} \cdot G}_{t}} \cdot \Lambda^{(t)}\right) \\\leq O\left(\sqrt{\frac{t^8}{2^n}}\right)\,.
\end{multline*}
\end{lemma}
\begin{proof}
    First, note that collision-free relations are a subset of bijective relations, and the projectors $\Pi^{(t)}$ and $\Lambda^{(t)}$ commute, which $\Lambda^{(t)} = \Lambda^{(t)} \Pi^{(t)}$.
    Similar to the proof from \cite{ma2024pseudorandom}, we can explicitly write out the states $\ket{\mathcal{A}_t^{\mathcal{O}}}$ in either case.
    For a set $S$, let $S^t$ denote the $t$-fold cartestian product of the set, and $S^t_{\mathrm{dist}}$ be the distinct subset of $S^t$ (i.e. all sets of $t$ elements that are distinct from each other).
    First, we note the following
    \begin{equation*}
        \ket{\mathcal{A}_t^{\mathsf{PR}\cdot G}} = \sqrt{\frac{(2^n-t)!}{2^n!}} 
        \sum_{\substack{x_1, \ldots, x_t \in (\{0, 1\}^{n})^t \\ y_1, \ldots, y_t \in (\{0, 1\}^n)^{t}_{\mathrm{dist}}}} 
        \left(\prod_{i = 1}^t 
        \left(
            \ket{y_i}\!\!\bra{x_i} \cdot G \cdot A_i
        \right)\ket{0}
        \right) \otimes \ket{\{(x_i, y_i)\}_{i = 1}^t}\,.
    \end{equation*}
    Performing the projection onto bijective relations yields the following state:
    \begin{equation*}
        \Pi^{(t)} \cdot \ket{\mathcal{A}_t^{\mathsf{PR}\cdot G}} = \sqrt{\frac{(2^n-t)!}{2^n!}} 
        \sum_{\substack{x_1, \ldots, x_t \in (\{0, 1\}^{n})^t_{\mathrm{dist}} \\ y_1, \ldots, y_t \in (\{0, 1\}^n)^{t}_{\mathrm{dist}}}} 
        \left(\prod_{i = 1}^t 
        \left(
            \ket{y_i}\!\!\bra{x_i} \cdot G \cdot A_i
        \right)\ket{0}
        \right) \otimes \ket{\{(x_i, y_i)\}_{i = 1}^t}\,.
    \end{equation*}
    Further performing the projection onto collision free relations yields the following:
    \begin{equation*}
        \Pi^{(t)} \cdot \ket{\mathcal{A}_t^{\mathsf{PR}\cdot G}} = \sqrt{\frac{(2^n-t)!}{2^n!}} 
        \sum_{\substack{x_1, \ldots, x_t \in (\{0, 1\}^{n})^t_{\mathrm{dist}} \\ y_1, \ldots, y_t \in \mathsf{CF}_{x_1, \ldots, x_t}}} 
        \left(\prod_{i = 1}^t 
        \left(
            \ket{y_i}\!\!\bra{x_i} \cdot G \cdot A_i
        \right)\ket{0}
        \right) \otimes \ket{\{(x_i, y_i)\}_{i = 1}^t}\,.
    \end{equation*}
    Now we turn to the $\mathsf{FHFHF}$ oracle.  From \Cref{lem:fhfhf_action}, we have the following.
    \begin{equation*}
         \ket{\mathcal{A}_t^{\mathsf{FHFHFO}\cdot G}} = \sqrt{\frac{1}{2^{nt}}} 
        \sum_{\substack{x_1, \ldots, x_t \in (\{0, 1\}^{n})^t \\ y_1, \ldots, y_t \in (\{0, 1\}^n)^{t}}} 
        \left(\prod_{i = 1}^t 
        \left(
            \ket{y_i}\!\!\bra{x_i} \cdot G \cdot A_i
        \right)\ket{0}
        \right) \otimes \ket{\mathsf{fhfhf}_{\{(x_i, y_i)\}_{i = 1}^t}}\,.
    \end{equation*}
    From the \Cref{lem:relation_states_close}, this state is within trace distance $O(\sqrt{t^2 / 2^n})$ of the following state
    \begin{equation*}
         \ket{\mathcal{A}_t^{\mathsf{FHFHFO}\cdot G}} = \sqrt{\frac{1}{2^{nt}}} 
        \sum_{\substack{x_1, \ldots, x_t \in (\{0, 1\}^{n})^t \\ y_1, \ldots, y_t \in (\{0, 1\}^n)^{t}}} 
        \left(\prod_{i = 1}^t 
        \left(
            \ket{y_i}\!\!\bra{x_i} \cdot G \cdot A_i
        \right)\ket{0}
        \right) \otimes \ket{\mathsf{fhfhf}^{\mathrm{dist}}_{\{(x_i, y_i)\}_{i = 1}^t}}\,.
    \end{equation*}
    Applying $\mathsf{Collide}$ to this state, we get exactly the following state.
    \begin{multline*}
        \mathsf{Collide} \cdot \ket{\mathcal{A}_t^{\mathsf{FHFHFO}\cdot G}} = \sqrt{\frac{1}{2^{nt}}} 
        \sum_{\substack{x_1, \ldots, x_t \in (\{0, 1\}^{n})^t_{\mathrm{dist}} \\ y_1, \ldots, y_t \in \mathsf{CF}_{x_1, \ldots, x_t}}} 
        \left(\prod_{i = 1}^t 
        \left(
            \ket{y_i}\!\!\bra{x_i} \cdot G \cdot A_i
        \right)\ket{0}
        \right) \otimes \ket{\mathsf{fhfhf}^{\mathrm{dist}}_{\{(x_i, y_i)\}_{i = 1}^t}}\\
        + \mathsf{Collide} \cdot \sqrt{\frac{1}{2^{nt}}} 
        \sum_{\substack{x_1, \ldots, x_t \in (\{0, 1\}^{n})^t_{\mathrm{dist}} \\ y_1, \ldots, y_t \not\in \mathsf{CF}_{x_1, \ldots, x_t}}} 
        \left(\prod_{i = 1}^t 
        \left(
            \ket{y_i}\!\!\bra{x_i} \cdot G \cdot A_i
        \right)\ket{0}
        \right) \otimes \ket{\mathsf{fhfhf}^{\mathrm{dist}}_{\{(x_i, y_i)\}_{i = 1}^t}}\,.
    \end{multline*}
    Since we showed that collision free relation states are orthogonal to each other, compress will map the first term to the correct relation states, and shrink the trace of the second term. However, since the probability that a $y$ is not collision free relative to $x$ is at most $t^8 / 2^{n}$, the trace of the second term is at most $\sqrt{t^8 / 2^n}$ from the gentle measurement lemma, and the normalization difference in the first term contributes at most a $t^2 / 2^n$ trace distance error, as noted before.  Applying the triangle inequality on all of the steps, we get the desired error bound.
\end{proof}
In the previous lemma, we showed that for collision-free relation states, the $\mathsf{FHFHF}$ can be thought of as implementing the path recording oracle, but \cite{ma2024pseudorandom} only proves that the path recording oracle works on bijective relations. Therefore, we need a small lemma to show that the projecting the path recording isometry onto bijective and collision-free relations yields similar states.
\begin{lemma}
    For all $t$-query adversaries $\mathcal{A}_t^{\mathcal{O}}$ and unitaries $G$, the following bound holds:
    \begin{equation*}
    \td\left(\Lambda^{(t)}\cdot\proj{\mathcal{A}^{\mathsf{PR}\cdot G}} \cdot \Lambda^{(t)}, \Pi^{(t)} \cdot \proj{\mathcal{A}^{\mathsf{PR}\cdot G}}\cdot \Pi^{(t)}\right) \leq O\left(\sqrt{\frac{t^8}{2^{n}}}\right)\,.
    \end{equation*}
\end{lemma}
\begin{proof}
    Since the probability that a uniformly random $y_1, \ldots, y_t$ is collision free is at least $1 - t^8 / 2^{n}$ when measuring a uniformly random distinct $y_1, \ldots, y_t$, applying the gentle measurement measurement lemma yields the desired trace distance bound. 
\end{proof}

The final important step is to take $G$ to be a $2$-design, which forces the relation states to be close to the subspace of bijective relations.  Formally, we have the following lemma, which follows closely the proof of \cite{ma2024pseudorandom}.
\begin{lemma}[Trace distance of projecting onto bijective relation states]
\label{lem:projecting_bijective_relations}
    \begin{multline*}
        \td\left(\mathbb{E}_{\mathsf{C}} \left[\proj{\mathcal{A}^{\mathsf{FHFHFO}\cdot \mathsf{C}}}\right], \mathsf{Collide} \cdot \mathbb{E}_{\mathsf{C}} \left[\proj{\mathcal{A}^{\mathsf{FHFHFO}\cdot \mathsf{C}}}\right]\cdot \mathsf{Collide}\right) \\\leq O\left(\sqrt{\frac{t^8}{2^{n}}}\right)\,.
    \end{multline*}
\end{lemma}
\begin{proof}
    Following the proof in \cite{ma2024pseudorandom}, we observe that $\mathsf{Compress}^{\dagger} \cdot \mathsf{Compress} \cdot \mathsf{Collide} = \mathsf{Collide}$.  
    Then applying \Cref{lem:un-normalized_gentle_measurement}, we have the following bound on the trace distance:
    \begin{align*}
        &1 - \Tr\left(\mathsf{Collide}\cdot\mathbb{E}_{\mathsf{C}} \left[\proj{\mathcal{A}^{\mathsf{FHFO}\cdot \mathsf{C}}}\right]\cdot\mathsf{Collide}\right)\\
        &\hspace{15mm} = 1 - \Tr\left(\mathsf{Compress} \cdot \mathsf{Collide}\cdot\mathbb{E}_{\mathsf{C}} \left[\proj{\mathcal{A}^{\mathsf{FHFHFO}\cdot \mathsf{C}}}\right]\cdot\mathsf{Collide} \cdot \mathsf{Compress}^{\dagger}\right)\\
        &\hspace{15mm} \leq 1 - \Tr\left(\Lambda^{(t)}\cdot\mathbb{E}_{\mathsf{C}} \left[\proj{\mathcal{A}^{\mathsf{PR}\cdot \mathsf{C}}}\right]\cdot \Lambda^{(t)}\right) + O\left(\sqrt{\frac{t^8}{2^n}}\right)\\
        &\hspace{15mm} \leq 1 - \Tr\left(\Pi^{(t)}\cdot\mathbb{E}_{\mathsf{C}} \left[\proj{\mathcal{A}^{\mathsf{PR}\cdot \mathsf{C}}}\right]\cdot \Pi^{(t)}\right) + O\left(\sqrt{\frac{t^8}{2^n}}\right)\\
        &\hspace{15mm} \leq \frac{t(t-1)}{2^n} + O\left(\sqrt{\frac{t^8}{2^n}}\right)\,.
    \end{align*}
    Here the second line uses \cref{lem:compressing_bijective_relations}, the second to last line uses the previous lemma, which states that the projection onto collision free relations is close to the projection onto bijections, and the final line uses Corollary 4.1 from \cite{ma2024pseudorandom}.
    This completes the proof of the lemma.
\end{proof}

Now, we put everything together and prove \Cref{thm:fhfhfc_psuedorandom}, using the fact that \cite{ma2024pseudorandom} has already shown that the path recording oracle itself is both right-invariant and indistinguishable from a Haar random unitary.

\begin{proof}[Proof of \Cref{thm:fhfhfc_psuedorandom}]
    We consider the following sequence of hybrids
    \begin{enumerate}
        \item The state of the adversary querying the $\mathsf{FHFHFC}$ ensemble.
        \item The adversary's view of the purification of the $\mathsf{FHFHFC}$ ensemble.
        \item The adversary's view of purification of the $\mathsf{FHFHFC}$ ensemble, after projecting onto collision free relation states.
        \item The adversary's view of the purification of the $\mathsf{PR} \cdot \mathsf{C}$ oracle, after projection onto collision free relations.
        \item The adversary's view of the purification of the $\mathsf{PR} \cdot \mathsf{C}$ oracle.
        \item The adversary's view of the $\mathsf{PR}$ oracle.
        \item The state of the adversary querying a Haar random unitary.
    \end{enumerate}
    $(1)$ is equivalent to $(2)$ by the purification technique for the $\mathsf{FHF}$ oracle (\Cref{lem:purification_fhf}).  $(2)$ is close to $(3)$ by \cref{lem:projecting_bijective_relations}.  $(3)$ is close to $(4)$ by \cref{lem:compressing_bijective_relations}.  $(4)$ is close to $(5)$ by Corollary $4.1$ from \cite{ma2024pseudorandom}, and $(5)$ is close to $(6)$ by Lemma $4.3$ from \cite{ma2024pseudorandom}.  Applying the triangle inequality between all of these hybrids completes the proof. 
\end{proof}

We note that we expect the proof in this section also applies to the case when $\mathsf{F}$ is a random binary phase oracle.  In particular, the only difference is that the set state corresponding to binary phase states can only store distinct sets (since elements cancel $\bmod 2$).  Since we show orthogonality for the distinct relation states, and the closeness to the distinct relation states holds by a counting argument independent of the set states themselves, all of the results should hold in this case.  

\begin{conjecture}
    We conjecture that it can be shown that $\mathsf{FHFHF}$ itself (without the random Clifford operation) for uniformly random phases is a strong (i.e. inverse secure) pseudo-random unitary as well.  As noted before, proving this likely requires understanding when $((x_i, y_i))_{i}$ is uniquely determined by the set of $x$'s, $y$'s, and $x \oplus y$'s.  
\end{conjecture}

\subsubsection{Removing the random Clifford unitary}
In Ref.~\cite{ji2018pseudorandom} it was conjectured that a variant of the $\mathsf{FHFHF}$ ensemble is a pseudo-random unitaries.
Here, we observe that the Clifford unitary in $\mathsf{FHFHFC}$ can be replaced by any approximate unitary $2$-design. 
More precisely, the above proof works unchanged for $C$ a relative error approximate $2$-design $\nu$ in the sense that
\begin{equation}\label{eq: relative design}
    (1-\varepsilon)\underset{U\sim \nu}{\mathbb{E}} U^{\otimes t} (.)(U^{\dagger})^{\otimes t}\preceq \underset{U\sim \nu}{\mathbb{E}} U^{\otimes t} (.)(U^{\dagger})^{\otimes t}\preceq (1+\varepsilon)\underset{U\sim \nu}{\mathbb{E}}U^{\otimes t} (.)(U^{\dagger})^{\otimes t},
\end{equation}
where $A\preceq B$ if $B-A$ is a completely positive map.
Let $\nu_{\mathsf{FHF}}$ denote the probability distribution defined by drawing two iid random diagonal untiaries $U_{1/2}=\sum_{x\in\{0,1\}^n} e^{\mathrm{i}\varphi_{x,1/2}}|x\rangle\langle x|$ for uniformly random angles $\varphi_{x,1/2}\in \{2\pi k/2^n\}_{k=0}^{2^n-1}$ and implementing  $U_1 H^{\otimes n} U_2$.
\begin{theorem}[\cite{nakata2017unitary}]
    We have 
    \begin{equation}\label{eq: additive design}
        ||\mathbb{E}_{U\sim \nu_{\mathsf{FHF}}^{* l}}U^{\otimes 2} (.) (U^{\dagger})^{\otimes 2}-\mathbb{E}_{U\sim \mu_H}U^{\otimes 2} (.) (U^{\dagger})^{\otimes 2}||_{\diamond}\leq 4\times 2^{-nl},
    \end{equation}
    where $||(.)||_{\diamond}$ denotes the diamond norm, or channel indistinguishability defined as 
    \begin{equation}
        ||A||_{\infty}:= \max_{\rho, ||\rho||_1=1}||(A\otimes \mathbbm{1}_d)\rho||_{1}.
    \end{equation}
\end{theorem}
This is readily combined with the following standard lemma from Ref.~\cite{brandao2016local}:
\begin{lemma}
    If $\nu$ is an additive $\varepsilon$-approximate unitary $t$-design in the sense of \Cref{def: approximate designs} then $\nu$ is a relative $\varepsilon\times 2^{2nt}$-approximate $t$-design.
\end{lemma}
It thus suffices to choose $l=3$ to obtain a relative error design from iterations of $\mathsf{F}$ and $\mathsf{H}$.
In particular, we can replace $\mathsf{C}$ in $\mathsf{FHFHFC}$ by a unitary drawn from $\nu_{\mathsf{FHF}}^{*3}$.
Notice that the $\mathsf{F}$ taken in the proof is identically distributed to $D$ (i.e. it is a uniformly random diagonal matrix).
We arrive at the following statement
\begin{theorem}
    The ensemble $\mathsf{FHFHFHFHFHF}$ is indistinguishable from a Haar random unitary by any agent with access to $\mathrm{poly}(n)$ adaptive queries.
\end{theorem}
In particular, the diagonal unitaries in the repetated $\mathsf{FH}$ can be replaced by cryptographic assumptions such as quantum secure one-way functions or $\mathsf{HPS}$ as discussed in the next section.

\subsubsection{Instantiation under $\mathsf{HPS}$}
\fi
Note that the version of $\mathsf{HPS}$ learning proposed in \cref{sec:background} involves distinguishing an ensemble of states from Haar random states.  As of this paper, we do not know how to build pseudo-random unitaries from pseudo-random states, which seems closely related to building pseudo-random unitaries from the distinguishing variant of $\mathsf{HPS}$.  We instead propose that building pseudo-random unitaries by layering $\mathsf{HPS}$ circuits with Hadamard gates yields a pseudo-random unitary.
\begin{conjecture}[$\mathsf{HPS}$ pseudorandom unitaries]
\label{conj:HPS_unitary}
    Let $\Vec{A} = \{\Vec{A}_\ell \sim \mathbb{Z}_{2}^{m \times n}\}_{\ell = 1}^{6}$, and let $\Vec{k} = \{\Vec{k}_\ell = (\theta_1, \ldots, \theta_m)\}_{\ell = 1}^{6}$ be chosen uniformly at random (according to some discretization), with $m \geq n + \log(n)$.  We conjecture that the following gives rise to a pseudorandom unitary:
    \begin{equation*}
        \mathsf{U}^{\Vec{A}, \Vec{k}} = \prod_{i = 1}^{6} \left( H^{\otimes n} \exp\left(i \sum_{i = 1}^{m} (\theta_{\ell})_i \bigotimes_{j = 1}^{n} Z^{(\vec{A}_{\ell})_{ij}}\right)\right)\,.
    \end{equation*}
\end{conjecture}
Of course, we note that this conjecture is an assumption about the distinguishing advantage of any computationally bounded adversary.  
Our evidence that this conjecture is true comes from the conjectured security of the repeated diagonal-Hadamard pseudo-random unitary\ifhaspru that we proved in the previous section.
In that section, we took ``ideal'' HPS instances (i.e. completely random diagonal unitary) for all the diagonal gates, which can be viewed as running an exponentially deep HPS circuit.
Thus, our conjecture essentially amounts to the fact that stopping this evolution by the HPS unitary after some polynomial-amount of time yields a pseudorandom unitary instead of a truly Haar random one\fi.
It is an interesting question to interrogate the security of this construction, as beyond the security of the idealized scheme, we have no evidence for or against the security of this pseudorandom unitary.
\section{Discussion and Future Work}

Here we discuss connections between the $\mathsf{HPS}$ assumption and other assumptions in post-quantum cryptography and fully quantum cryptography.  We also provide directions for potential future work.

\paragraph{Applications of HPS in the real world.} In the work of \cite{bouland2019computational}, the authors consider pseudo-random states formed by starting from a maximally entangled state, and applying time evolution by a fixed operator (called $H_{\mathrm{CFT}}$) that represents the black hole.  In order to justify their assumption, they model the black hole as a Haar random unitary, and construct pseudo-random states by alternating random Pauli gates with applications of the Haar random unitary.  
The $\mathsf{HPS}$ assumption, while having no formal connection to the physics of black holes, has a similar flavor to the random-circuit toy models of black hole evolution.  In particular, if we imagine that the CFT Hamiltonian looks like the $\mathsf{X}$-program Hamiltonian, then the $\mathsf{HPS}$ assumption becomes akin to distinguishing the states after time-evolution of a black hole from a known state.  Furthermore, if we take the view that black hole evolution appends a uniformly random $\theta_i$, the $\mathsf{HPS}$ Hamiltonian allows us to both perform black hole time evolution, and interweave Pauli $X$ gates, yielding a similar construction to the black-box $\mathsf{PRS}$ from \cite{bouland2019computational}.  We believe that the HPS assumption could serve as a toy model for black hole scrambling dynamics, and would be interested in seeing a small scale implementation on a real quantum computer.

\paragraph{Building more cryptography.}
In this work, we built a number of cryptographic primitives from the HPS assumption.  These included all primitives contained in ``MicroCrypt'' (which are implied by one-way functions but do not imply one-way functions in a black box way), and even some primitives (like public key cryptography with quantum public keys) which are not generally considered MicroCrypt primitives.  
We believe that the concrete assumption we propose should allow for much more interesting cryptography, in the same way that cryptographers have extended the \emph{Learning with Errors}~\cite{regev2009lattices} problem to build increasingly powerful and complex tools like functional encryption, indistinguishability and obfuscation, and zero-knowledge proofs.  We hope that similar extensions of our results will yield efficient implementations of more useful cryptographic primitives on near-term quantum computers as well.  We point to recent work describing possible constructions of indistinguishability obfuscation (iO) from classical reversible circuits as a possible future application of the HPS assumption: Can HPS be used in a similar way to build a quantum analog of random output iO? 

\paragraph{Establishing the security of HPS.}

This paper proposes a concrete problem that we consider hard to solve.  While there have been attempts to solve related problems, and these problems have evaded polynomial time solutions so far, it would be foolish to take the word of a single paper on the computational hardness of any problem.  In order for both the authors and the community at large to build confidence in the security of the HPS assumption, effort must go into attempting to break the assumption, or providing more evidence of its hardness.  This paper (the authors believe) provides ample evidence towards the security of HPS.  We provide an worst-to-average case reduction, an attack in certain parameter regimes, and an analysis of run-time of state of the art algorithms for attacking HPS.  We believe that our work makes the HPS assumption worthy of future study and an interesting algorithmic problem in its own right.  We hope our work prompts researchers to design attacks and attempt to break the security of our assumptions.

\printbibliography

\begin{appendix}
%!TEX root=./main.tex
\section{Proofs for the Design Property of Hamiltonian Phase States}
\label{app:design proofs}

% \paragraph*{Design property from spectral gaps}
% We finally proceed to prove Theorem~\ref{thm:diagonal designs}.
We will prove Theorem~\ref{thm:diagonal designs} via a bound on the spectral gap of the random walk $\nu$ on the diagonal group $D(2^n)$. 
Instead of bounding the diamond norm directly, we consider the essential norm (see the discussion in Ref.~\cite{chen2024incompressibility}) in the diagonal group:
\begin{equation}
    g_D(\nu,t)\coloneqq\left|\left| \mathbb{E}_{U\sim \nu}U^{\otimes t}\otimes \overline{U}^{\otimes t}-\mathbb{E}_{U\sim \mu_D}U^{\otimes t}\otimes \overline{U}^{\otimes t}\right|\right|_{\infty}~.
\end{equation}
This equals the operator norm of $\mathbb{E}_{U\sim \nu}U^{\otimes t}\otimes \overline{U}^{\otimes t}$ restricted to orthocomplement of the vectors invariant under $U^{\otimes t}\otimes \overline{U}^{\otimes t}$ for $U\in D(2^n)$.
Therefore, $g_D$ inherits the submultiplicativity of the operator norm:
\begin{equation}
    g_D(\nu^{*m},t)\leq g_D(\nu,t)^m.
\end{equation}
Our strategy is therefore to establish a bound on $g_D(\nu,t)$ and then amplify to apply the following standard lemma:
\begin{lemma}
A probability measure $\nu$ on $D(2^n)$ is a $g_D(\nu,t)\times 2^{nt}$-approximate diagonal $t$-design.
\end{lemma}
\begin{proof}
  $g_D(\nu,t)$ equals the induced $2$-norm of $\mathcal{M}^{(t)}_{\nu}-\mathcal{M}^{(t)}_{\mu_H}$. 
  The claim now follows from the inequality $||(.)||_{\diamond}\leq 2^{d}||(.)||_{\infty}$ (see e.g.~\cite{low2010pseudo}), where $(.)$ acts on the space of linear operators on $\mathbb{C}^d$. 
  In our case $\mathcal{M}^{(t)}_{\nu}-\mathcal{M}^{(t)}_{\mu_D}$ acts on $\mathbb{C}^{2^{nt}\times 2^{nt}}\stackrel{\sim}{=} \mathbb{C}^{2^{2nt}}$.
\end{proof}

 In the following, we denote by $\nu$ the probability measure over $D(2^n)$ defined by $\exp(\mathrm{i}\theta \bigotimes_{j=1}^nZ^{\boldsymbol{A}_{1j}})$, with $y$ drawn uniformly from $\{0,1\}$ and $\theta$ drawn uniformly from $(0,2\pi]$. 
 In the remainder of this section we will use the simplified notation $a$ for the vector $a=\boldsymbol{A}_{1i}$.
Moreover, we denote by $\nu_q$ the measure that instead draws $\theta$ uniformly from the discrete set $\{2\pi l/q\}_{l=0}^{q-1}$.
Theorem~\ref{thm:diagonal designs} will follow from the following lemma:
\begin{lemma}\label{lemma:gapbound}
  For all $t\geq 1$ and all $n\geq 1$ we have $ g_D(\nu,t)\leq 1-\frac{1}{2t}$.
  Moreover, $g_D(\nu_q, t)\geq 1-\frac{1}{2t}$ if $q>2t$.
\end{lemma}
\begin{proof}
    We begin by expanding the moment operators in the computational basis:
    \begin{align}
    \begin{split}
        \underset{U\sim \nu}{\mathbb{E}} U^{\otimes t}\otimes \overline{U}^{\otimes t}&=\underset{\theta\in (0,2\pi]}{\mathbb{E}}\underset{a\in \{0,1\}^n }{\mathbb{E}}\sum_{\textbf{x},\textbf{x'}\in\{0,1\}^{nt}} e^{\mathrm{i}\theta\sum_{j=1}^t(-1)^{\langle a, x_j\rangle} -(-1)^{\langle a, x'_j\rangle }}|\textbf{x},\textbf{x}'\rangle\langle \textbf{x},\textbf{x}'|.\\
         \underset{U\sim \mu_D}{\mathbb{E}} U^{\otimes t}\otimes \overline{U}^{\otimes t}&=\sum_{\textbf{x},\textbf{x}'\in\{0,1\}^{nt}}\underset{\phi_x\in(0,2\pi]}{\mathbb{E}}e^{\mathrm{i}\sum_{j=1}^t\phi_{x_j}-\phi_{x'_j}}|\textbf{x},\textbf{x}'\rangle\langle\textbf{x},\textbf{x}'|~.
        \end{split}
 \end{align}
 $\mathbb{E}_{\boldsymbol{x},\boldsymbol{x}'\in\{0,1\}^n}\mathbb{E}_{\phi_x}e^{\mathrm{i}\sum_{j=1}^t\phi_{x_j}-\phi_{x'_j}}=0$ if and only if $\textbf{x}$ and $\textbf{x}'$ are related by a permutation of the $t$ bitstrings. 
 If two bitstrings are related by such a permutation we denote it by $\textbf{x}\sim \textbf{x}'$.
 Therefore, we find 
 \begin{align}
     g_D(\nu,t)&=\max_{\textbf{x},\textbf{x}';\textbf{x}\not \sim \textbf{x}'}\underset{a\in \{0,1\}^n}{\mathbb{E}}\underset{\theta\in (0,2\pi]} {\mathbb{E}}e^{\mathrm{i}\theta\sum_{j=1}^t(-1)^{\langle a|x_j\rangle} -(-1)^{\langle a|x'_j\rangle }}.
 \end{align}
 Notice that $\mathbb{E}_{\theta\in (0,2\pi]} e^{\mathrm{i}\theta \sum_{j=1}^t(-1)^{\langle a, x_j\rangle} -(-1)^{\langle a, x'_j\rangle }}=0$ if and only if $\sum_{j=1}^t(-1)^{\langle a, x_j\rangle} -(-1)^{\langle a, x'_j\rangle }\neq 0$.
 Therefore
 \begin{equation}
     g_D(\nu,t)=\max_{\textbf{x},\textbf{x}';\textbf{x}\not \sim \textbf{x}'} \mathrm{Pr}_{a}\left[\sum_{j=1}^t(-1)^{\langle a, x_j\rangle} -(-1)^{\langle a, x'_j\rangle}=0\right].
 \end{equation}
 To bound this probability we use Fourier analysis of Boolean functions. 
 More precisely, we define the function $f^{\textbf{x},\textbf{x}'}:\{0,1\}\to \mathbb{R}$ as
 \begin{equation}
     f^{\textbf{x},\textbf{x}'}(a)\coloneqq\sum_{j=1}^t(-1)^{\langle a, x_j\rangle} -(-1)^{\langle a, x'_j\rangle }.
 \end{equation}
 Recall that the Fourier transform of a Boolean function $f$ is defined as
 \begin{equation}
     \hat{f}(a)= \underset{x\in\{0,1\}}{\mathbb{E}}f(x) (-1)^{\langle x|a\rangle}.
 \end{equation}
 We can easily verify that 
 \begin{equation}
      \hat{f}^{\textbf{x},\textbf{x}'}= \sum_{j=1}^t \delta_{x_j}-\delta_{x'_j},
 \end{equation}
 where $\delta_x(z)=\begin{cases} 0~~\text{if}~~x\neq z\\ 1~~\text{if}~~x=z\end{cases}$.

 If a string appears in both $\textbf{x}$ and $\textbf{x}'$ we can remove it without changing $f^{\textbf{x},\textbf{x}'}$.
 Therefore, we can remove strings until we obtain two tuples of length $0<r\leq t$ that have no common strings.
 Now, via Paseval's identity, we find the following bound:
 % on the support $\mathrm{supp}(f^{\textbf{x},\textbf{x}'})$.~\jonas{to be continued}.
\begin{align}
    \begin{split}
        (2r)^2 \left|\left\{a,\sum_{j=1}^r(-1)^{\langle a, x_j\rangle} -(-1)^{\langle a, x'_j\rangle }\neq 0\right\}\right|&\geq \sum_{a} \left(\sum_{j=1}^r(-1)^{\langle a, x_j\rangle} -(-1)^{\langle a, x'_j\rangle }\right)^2\\
        &=2^n\underset{a}{\mathbb{E}}\left(f^{\textbf{x},\textbf{x}'}(a)\right)^2\\
        &= 2^{n} \sum_{z} \left(\hat{f}^{\textbf{x},\textbf{x}'}(z)\right)^2\\
        &\geq 2^{n}2r. 
    \end{split}
\end{align}
Therefore, we find 
\begin{equation}
    g_D(\nu,t)=\max_{\textbf{x}\not \sim \textbf{x}'} 1-2^{-n}\left|\left\{a\in\{0,1\},\sum_{j=1}^t(-1)^{\langle a, x_j\rangle} -(-1)^{\langle a, x'_j\rangle }\neq 0\right\}\right|\leq 1-\frac{1}{2r}\leq 1-\frac{1}{2t},
\end{equation}
which completes the proof.

For the bound on $g_D(\nu_q,t)$ we observe just as above that we have 
  \begin{equation}
     g_D(\nu,t)=\max_{\textbf{x},\textbf{x}';\textbf{x}\not \sim \textbf{x}'} \mathrm{Pr}_{a}\left[\sum_{j=1}^t(-1)^{\langle a, x_j\rangle} -(-1)^{\langle a, x'_j\rangle }=0 \mod q \right].
 \end{equation}
 But as $\left|\sum_{j=1}^t(-1)^{\langle a|x_j\rangle} -(-1)^{\langle a|x'_j\rangle }\right|\leq 2t$, and $q> 2t$ we have $\sum_{j=1}^t(-1)^{\langle a, x_j\rangle} -(-1)^{\langle a, x'_j\rangle}=0 \mod q$ if and only if $\sum_{j=1}^t(-1)^{\langle a, x_j\rangle} -(-1)^{\langle a, x'_j\rangle }=0$.
 The rest of the argument thus goes through without change.
\end{proof}

\begin{proof}[Proof of Corollary~\ref{lemma: many orthogonal states}]
We use the notation $|\psi\rangle^{\otimes t,t}\coloneqq|\psi\rangle^{\otimes t}\otimes \overline{|\psi\rangle}^{\otimes t}$.
For any state $|\psi\rangle$:
\begin{align}\label{eq:momentbound}
\begin{split}
    \mathbb{E}_{U\sim \mu_D}|\langle\psi|U|+^n\rangle|^{2t}&=\langle\psi|^{\otimes t,t}\mathbb{E}_{U\sim \mu_D}U^{\otimes t}\otimes \overline{U}^{\otimes t}|+^n\rangle^{\otimes 2t}\\
    &\leq \sqrt{\langle +^n|^{\otimes t,t} (\mathbb{E}_{U\sim \mu_D}U^{\otimes t}\otimes \overline{U}^{\otimes t})^{\dagger}\mathbb{E}_{U\sim \mu_D}U^{\otimes t}\otimes \overline{U}^{\otimes t}|+\rangle^{\otimes t,t}}\\
     &\leq \sqrt{\mathbb{E}_{U\sim \mu_D}|\langle +^n|U|+^n\rangle|^2}.
    \end{split}
\end{align}
We compute
\begin{align}
    \begin{split}
        \mathbb{E}_{U\sim \mu_D}|\langle +^n|U|+^n\rangle|^2 = \frac{1}{2^{2nt}}\sum_{x_1,\ldots,x_t, x'_1,\ldots,x'_t} \mathbb{E}_{\phi_x,x\in\{0,1\}}e^{\mathrm{i}\sum_{j=1}^t \phi_{x_j}-\phi_{x'_j}}
    \end{split}
\end{align}
As long as $\sum_{j=1}^t\phi_{x_j}-\phi_{x'_j}=0$, for all $\phi$, we have $$\mathbb{E}_{\phi_j}e^{\mathrm{i}\sum_{j=1}^t\phi_{x_j}-\phi_{x'_j}}=0.$$
But $\sum_{j=1}^t\phi_{x_j}-\phi_{x'_j}=0$ if and only if $(x_1,\ldots,x_t)$ and $(x'_1,\ldots x'_t)$ are related by a permutation in $S_t$.
Clearly, for every fixed tuple $(x_1,\ldots, x_t)$ there are at most $t!$ permutations.
Consequently, we find 
\begin{equation}
     \mathbb{E}_{U\sim \mu_D}|\langle +^n|U|+^n\rangle|^2\leq \frac{t!}{2^{nt}}
\end{equation}
and therefore, plugging into Eq.~\eqref{eq:momentbound},
\begin{equation}
    \mathbb{E}_{U\sim \mu_D}|\langle\psi|U|+^n\rangle|^{2t}\leq \sqrt{\frac{t!}{2^{nt}}}.
\end{equation}
By Theorem~\ref{thm:diagonal designs} we have that $\nu^{*m}$ is a $\varepsilon$-approximate diagonal $t$-design with $\varepsilon\leq \sqrt{t!/2^{nt}}$ and thus
\begin{equation}
    \mathbb{E}_{U\sim \nu^{*d}} |\langle\psi|U|+^n\rangle|^{2t}\leq   \mathbb{E}_{U\sim \mu_D}|\langle\psi|U|+^n\rangle|^{2t}+\varepsilon\leq 2\sqrt{\frac{t!}{2^{nt}}}.
\end{equation}

We can split the expectation value in~\cref{eq:momentbound} into the discrete part and the angles.
Applying Markov's inequality over the distribution of $\boldsymbol{A}$ we find
\begin{equation}
    \mathrm{Pr}_{\boldsymbol{A}}[\mathbb{E}_{\boldsymbol{\theta}}|\langle\psi|\prod_i e^{\sum_{i=1}^m\mathrm{i}\theta_i\bigotimes_{j=1}^{n}Z^{\boldsymbol{A}_{ij}}}|+^n\rangle|^{2t}\geq (2t!/2^{nt})^{\frac14})]\leq (2(t!/2^{nt})^{\frac14}).
\end{equation}
Therefore, with high probability over the $Z$ strings we have 
\begin{equation}\label{eq:momentboundforfixedZstrings}
    \mathbb{E}_{\boldsymbol{\theta}}|\langle\psi|e^{\sum_{i=1}^m\mathrm{i}\theta_i\bigotimes_{j=1}^{n}Z^{\boldsymbol{A}_{ij}}}|+^n\rangle|^{2t}\geq (2t!/2^{nt})^{\frac14}).
\end{equation}
We can now use Markov's inequality again to prove that most states for any state $|\psi\rangle$ almost all states $U|0^n\rangle$ have small overlap. 
Now consider any $\boldsymbol{A}$ such that~\cref{eq:momentboundforfixedZstrings} holds.
For any $\delta>0$, we have
\begin{align}
\begin{split}
    \mathrm{Pr}_{\boldsymbol{\theta}}[|\langle\psi|\prod_i e^{\mathrm{i}\theta_i\bigotimes_{j=1}^nZ^{\boldsymbol{A}_{ij}}}|+^n\rangle|^{2}\geq (1-\delta)]&=\mathrm{Pr}_{\boldsymbol{\theta}}[|\langle\psi|e^{\sum_{i=1}^m\mathrm{i}\theta_i\bigotimes_{j=1}^{n}Z^{\boldsymbol{A}_{ij}}}|+^n\rangle|^{2t}\geq (1-\delta)^{t}]\\
    & \leq \frac{\mathbb{E}_{\boldsymbol{\theta}}|\langle\psi|e^{\sum_{i=1}^m\mathrm{i}\theta_i\bigotimes_{j=1}^{n}Z^{\boldsymbol{A}_{ij}}}|+^n\rangle|^{2t}}{(1-\delta)^{t}}\\
    &\leq 2^{-\frac{1}{4} [n-\log_2(t)+\log_2(1/(1-\delta)]t}.
    \end{split}
\end{align}
In particular, even if we choose $1-\delta=2^{-n/8}$ and $t=\poly(n)$, we find that 
\begin{align}
    \bra{\psi}e^{\sum_{i=1}^m\mathrm{i}\theta_i\bigotimes_{j=1}^{n}Z^{\boldsymbol{A}_{ij}}}|+^n\rangle|^2\leq 2^{-n/8}
\end{align}
with probability $1-2^{-\frac{1}{10}nt}$.
From a union bound, we then find that there are at least $2^{cnt}$ distinct instances that are (almost) mutually orthogonal to each other.
\end{proof}

\end{appendix}

\end{document}